\documentclass[12pt]{amsart}

\usepackage{amsmath,amsfonts,amssymb,amsthm,epsfig,color, graphicx, natbib, appendix,comment,subcaption, amssymb,xcolor,geometry,bm,enumerate}

\usepackage{marginnote}

\usepackage{setspace} %
\usepackage[colorlinks=true,
    linkcolor=blue,
    filecolor=magenta,      
    urlcolor=blue]{hyperref}
\usepackage{tikz}
\usetikzlibrary{fit,positioning}
\usetikzlibrary{shapes,arrows}
\usetikzlibrary{intersections}
\usetikzlibrary{external}
\usepackage[nameinlink]{cleveref}

\DeclareMathOperator*{\diag}{diag}
\newtheorem{theorem}{Theorem}
\newtheorem{lemma}{Lemma}
\newtheorem*{definition}{Definition}
\newtheorem{assumption}{Assumption}

\newtheorem{proposition}{Proposition}
\newtheorem{corollary}{Corollary}
\theoremstyle{definition}
\newtheorem{example}{Example}
\theoremstyle{remark}

\newtheorem*{example2}{Example}

\newtheorem*{example*}{Example}
\theoremstyle{remark}

\theoremstyle{remark} \newtheorem{fact}{Fact}

\DeclareMathOperator*{\argmax}{arg\,max}
\DeclareMathOperator*{\argmin}{arg\,min}

\usepackage{tikz}
\usetikzlibrary{positioning,quotes}

\definecolor{darkblue}{rgb}{0.0, 0.0, 0.55}
\hypersetup{colorlinks,linkcolor={darkblue},citecolor={darkblue},urlcolor={darkblue}} 
\allowdisplaybreaks
\title{Incentive Design with Spillovers}

\author{Krishna Dasaratha \and Benjamin Golub \and Anant Shah}

\thanks{Dasaratha: krishnadasaratha@gmail.com, Boston University. Golub: ben.golub@gmail.com, Northwestern University. Shah: anantshah2026@u.northwestern.edu, Northwestern University. We thank Yann Calv\'{o} L\'{o}pez, Yu-Chi Hsieh, Vitalii Tubdenov, and Luis Henrique Ribeiro Linhares for excellent research assistance. We are grateful to Drew Fudenberg and his reading group for a very valuable discussion, and also for helpful comments  (in random order) to Erik Madsen, Ilya Segal, George Mailath, Aravindan Vijayaraghavan, Stephen Morris, Omer Tamuz, Marina Halac, Thomas Steinke, Marzena Rostek, Doron Ravid, Peter DeMarzo, Jose Betancourt Valencia, Roberto Corrao, Jacopo Perego, Juan Ortner, Alex Wolitzky, Jason Hartline, Evan Sadler, Adam Szeidl, Jeff Zwiebel, Sam Wycherley, Michael Ostrovsky, Elliot Lipnowski, Ravi Jagadeesan, Daniel Kr\"{a}hmer, Yifan Dai, Melika Liporace, Piero Gottardi, and Michael Powell, as well as many seminar and conference participants.}

\newcommand{\production}{Y}

\newcommand{\eqlbactions}{\bm{a}^{*}}

\newcommand{\actionagt}[1]{a_{#1}}
\newcommand{\network}{\bm{G}}

\newcommand{\subnetwork}{\widetilde{\bm{G}}}
\newcommand{\sharesvector}{\bm{\tau}}
\newcommand{\optsharesvector}{\bm{\tau}^{*}}
\newcommand{\sharesmatrix}{\bm{T}}

\newcommand{\sharesagt}[1]{\tau_{#1}}

\newcommand{\centrality}{c}
\newcommand{\productivity}{\alpha}
\newcommand{\balanceconstant}{\lambda}

\begin{document}

\begin{abstract}
A principal uses payments conditioned on stochastic outcomes of a team project to elicit costly effort from the team members. We develop a multi-agent generalization of a classic first-order approach to contract optimization by leveraging methods from network games. The main results characterize the optimal allocation of incentive pay across agents and outcomes. Incentive optimality requires equalizing, across agents, a product of (i) individual productivity (ii) organizational centrality and (iii) responsiveness to monetary incentives. We specialize the model to explore several applied questions, including whether compensation should reward individual ability or collaborativeness and how the strength of complementarities shapes pay dispersion.

\end{abstract}

\begin{titlepage}
\date{\today}
\maketitle
\thispagestyle{empty} 

\end{titlepage}

\section{Introduction}

A popular method of motivating members of a team is giving them performance incentives that depend on jointly achieved outcomes. Examples of such incentives include startup executives receiving firm stock and a marketing team receiving bonuses for achieving a sales target. How should such incentive schemes be designed and how should they take into account the team's production function?

We examine these questions in a simple model of a team working on a joint project. Each member chooses how much costly effort to exert. These actions jointly determine a real-valued \emph{team performance}---for example, the quality of a product---according to a sufficiently smooth, increasing function of the efforts, which may entail interactions such as complementarities among agents' efforts. Any performance level determines a probability distribution over observable project \emph{outcomes}. For example, the outcome may be the revenue from a project, which is increasing in non-contractible project quality, but stochastic due to factors outside the team's control. Although it is not possible to write contracts contingent on individual actions or the realized team performance, the principal can commit to a contract specifying nonnegative payments to each agent contingent on each possible project outcome. The principal's goal is to design this contract in a way that maximizes profit: revenue minus compensation.

The setting builds on the classic \citet{holmstrom1979moralhazard} model, in which a single agent produces work of a non-contractible quality resulting in an observable outcome.\footnote{In particular, the  obstacles to perfect contracting are the same: moral hazard and limited liability.}  In our setting, the analogue of quality is a jointly achieved performance. How incentive spillovers across agents matter for contract design---a central issue for modern firms---is not well understood, despite the immense amount learned about contract design since Holmstr\"{o}m's work.
In this paper, we make progress on this problem by leveraging some ideas from network theory.

To illustrate the basic importance of incentive spillovers, imagine that the principal slightly adjusts the contract of a particular agent, Bob, in a way that motivates him to work harder. In team production, changing one team member's action can change other agents' private returns to effort, holding fixed their own contracts. Those whose efforts are complements to Bob's are now motivated to work harder, while those whose efforts are substitutes to his have incentives to free-ride on his higher effort. Optimal contracts therefore depend on both the direct effect of each agent's action on team performance and the responses that action induces from others.

The paper makes two contributions: first, it provides a general characterization of optimal contracts in the presence of incentive spillovers; second, it applies this framework to understand how exogenous features of production in an organization shape optimal compensation outcomes.

The condition that constitutes the first contribution is stated in terms of three kinds of quantities that can be associated to each agent at any contract. The first quantity, an agent's \emph{marginal productivity}, is the partial derivative of team performance in an individual's action, holding others' actions fixed. The second quantity is called an agent's \emph{centrality}: a measure of connectedness in a network---which we are about to discuss in detail---reflecting incentive spillovers, with connectedness to more productive agents weighted more.  The relevance of these quantities for contract optimality should be intuitive in view of what we have said. The third quantity is an agent's \emph{responsiveness} to additional incentives, which involves the marginal utility of money and the curvature of the agent's cost. It accounts for the fact that an agent who has a low valuation of an additional dollar and finds it considerably more costly to take a larger action is, all else equal, a less appealing recipient of incentive pay. Our main result is that, when the binding incentive constraints are local, optimal contracts satisfy a balance condition: the product of the three quantities described above is equal across all agents receiving any incentive pay.

The spillover matrix used to define centrality plays a key role in our analysis. Economically, it summarizes how changing one agent's effort shifts other agents' marginal returns to effort, inducing strategic responses that propagate through the team: complements work harder, while substitutes have incentives to free-ride. Mathematically, the network is based on the Hessian of the team performance function evaluated at the equilibrium induced by the optimal contract---that is, the cross-partial derivatives of output as we differentiate in two agents' efforts. Because this network is pinned down by the technology's local curvature, it can in principle be estimated from data on how output changes as agents' actions change. It is worth emphasizing that our general model makes no parametric assumptions and thus allows quite flexible production functions for the team.\footnote{To take just one example, the team's production function could be an arbitrary polynomial, with monomials of arbitrary degree reflecting outputs generated through the joint efforts of arbitrary subsets of the team---e.g., complementing one another in threes, fours, and so on.} Nevertheless,  we show that the team's production function matters only through its first derivatives and its Hessian at the optimal contract. This allows us to leverage methods from well-understood network games whose payoffs are quadratic polynomials to analyze the spillover effects of locally perturbing incentive pay under arbitrary contracts. We use this key reduction to characterize the first-order conditions that determine the principal's optimal allocation of incentives both across agents and across outcomes.

There are two immediate corollaries of this characterization, new to the contract theory literature, that demonstrate its economic content. First, it implies a ranking of agents by compensation. If all agents have identical utility functions of money and these functions are concave, then those with a higher ``productivity times centrality" index must be paid more in every outcome. This ranking favors those agents who have high marginal productivity themselves and whose actions are also valuable to others' incentives directly and indirectly Second, the balance condition implies a ranking of observable \emph{outcomes} (e.g., revenue realizations) based on their incentive power, and this ranking holds across agents. Payments concentrate on outcomes with higher ``incentive power''---those where the ratio of probability to marginal probability is favorable for motivating effort. This generalizes the single-agent result of \cite{holmstrom1979moralhazard} to our multi-agent setting.

The balance condition is a critical step in our analysis of the problem, but all its terms are endogenous---evaluated at an equilibrium played by the agents at a contract chosen by the principal. To obtain more direct guidance, one needs to understand how exogenous features of the organization---its technological complementarities and measurable agent characteristics---map into compensation outcomes.  The remaining sections take up this challenge by introducing several different parametric structures that admit explicit solutions, including standard network games settings and CES production. 

We begin in \Cref{ss:rankone} with a setting where agents differ in the contribution their effort makes to performance in isolation, which we call \emph{standalone productivity}, and in how much their effort complements others', with each agent's complementarity described by a single parameter. In this tractable case, the balance condition yields explicit pay formulas: optimal compensation depends on a combination of an agent's standalone productivity and his complementarity parameter. As spillovers strengthen, compensation shifts sharply from being standalone-productivity-driven to being complementarity-driven. A numerical example suggests accounting for this effect can be very important: a principal neglecting spillovers when they are large fails to achieve most of the potential profits.

The rank-one case used to make these points has the property that an exogenous complementarity parameter for each agent maps monotonically into that agent's centrality. The relation is more complex with general complementarity networks. \Cref{s:applyparametric} works in a setting  with equal standalone productivities and gives a closed-form characterization of centralities and optimal compensation. Centrality need not be monotone in complementarity parameters: strengthening complementarities can weaken an agent's network position and pay because the principal rebalances the optimal contract and equilibrium actions.

Our applications so far have been set in the linear-quadratic framework, which is standard in network games but in which all complementarities are structurally bilateral. A different case of practical interest is one in which complementarities arise at the level of the whole team rather than between pairs. Motivated by this, \Cref{s:CES} specializes the framework to Cobb--Douglas and CES production functions capturing this feature. There, we show that the balance condition again yields clean closed-form solutions. This allows us to address an important outcome: pay dispersion. The key insight is that pay dispersion is governed by the elasticity of substitution: highly substitutable efforts justify concentrating pay on high-productivity agents, while strong team-level complementarity induces an optimizing principal to allocate pay more equally.

\subsection*{Related literature}

In the contract theory literature,  the \citet{holmstrom1979moralhazard} model---studying incentives for a single agent under moral hazard and imperfect observability---is a special case of our multi-agent setup. We use the first-order approach (see \citet{rogerson1985first} and \citet{jewitt1988justifying}). To our knowledge, there is not much work on how first-order conditions for contract optimality depend on spillovers.\footnote{\cite{itoh1991incentives} allows for a form of spillovers in a two-agent model.} Indeed, the extensive literature on moral hazard beginning with \citet{holmstrom1982moral}  focuses mainly on different questions. In that strand, a key question is how a principal can use observed outcomes to separately detect agents' deviations from a desired action profile, often a nearly first-best one \citep[see, e.g.,][]{mookherjee1984optimal,legrosmatsushima,legros1993efficient}.  Several features of our model prevent such schemes.\footnote{In particular, the contractible outcome (which has only finitely many possible values) is stochastically determined by a one-dimensional team performance, and there is limited liability.} 
In this type of situation, when observability and fine-grained auditing  of effort are fundamentally constrained, we examine how optimal contracts depend on spillovers in the production function.

Some of the closest work on optimal incentives in the presence of spillovers is found in the literature on networks. This includes, in addition to work already mentioned, papers such as \citet*{candogan2012optimal}, \citet{bloch2016targeting}, \citet{belhaj2018targeting}, \citet*{galeotti2020targeting}, \citet*{gaitonde2020adversarial},  \citet{shi2022optimal}, and \citet{demarzo2023contracting}.   Our main contribution to this literature is a study of a natural non-parametric formulation, both in terms of the production function and the form of incentives. We show that network game techniques permit some general characterizations of optimal contracts without the parametric assumptions common in the network games literature.    When we specialize to a canonical parametric environment in \Cref{s:applyparametric}, we also contrast the more specific implications of our analysis with existing parametric networks models---of which the closest is the contemporaneous work by \citet*{CMO} on optimal linear incentive contracts in quadratic network games. Both our analysis and results end up being quite different.

The problem of designing multi-agent contracts has also recently attracted attention in a new algorithmic contract theory literature---e.g., \citet*{dutting2023multiagent}, \citet*{ezra2023inapproximability}, and \citet*{ezramultiagentcombinatorial}.  A starting point of this work is that many standard approaches to finding optimal team contracts may make heavy demands on the analyst's knowledge of the entire production function and ability to perform computations on it. This literature studies environments with finitely many actions where combinatorial problems create obstacles to tractable optimization and examines whether contracts can be devised that achieve some fraction of optimal performance. Our approach is very different methodologically in that actions and team performance are continuous, but the results give a perspective---complementary to the algorithmic contract theory work---on parsimonious ways to assess and improve on contract performance. Recent work by \citet{zuo2024optimizing}, which discusses the structure of optimization problems in a model closely related to some of our special cases, shows that there are interesting computational questions in our more continuous type of setting as well.

\section{Model} \label{sec:generalmodel}

We first present the model and then, in \Cref{sec:remarks_model}, we comment on some issues of interpretation.

There are $n$ agents, $N=\{1,2,\ldots,n\}$. The agents take real-valued actions $\actionagt{i} \geq 0$, which can be interpreted as effort levels. These jointly determine a team \emph{performance}, given by a function $Y : \mathbb{R}^{n}_{\geq 0} \to \mathbb{R}_{\geq 0}$, which we assume is twice differentiable and strictly increasing in each of its arguments. This team performance determines the project \emph{outcome}, an element of the finite set $\mathcal{S}$. The probability of the outcome $s$ is $P_{s}(Y)$, where for any $s \in \mathcal{S}$, the function $P_{s}(\cdot)$ is strictly positive and twice differentiable.\footnote{The assumption that a probability of outcome function is strictly positive is not crucial to the analysis, but does simplify some statements.} 

There is also one principal. (When we use pronouns, we use ``she'' for the principal and ``he'' for an agent.) The principal receives \emph{revenue} $v_{s}$ from the outcome $s$.\footnote{This should be interpreted as the principal's valuation of that outcome realizing, gross of any payments she will make to the agents.}  The principal can make payments contingent on the project outcome but not on agents' efforts or the team performance $Y$.  The principal commits to a non-negative payment contingent on the outcome. If outcome $s$ is realized, the principal pays $\tau_{i}(s)$  to agent $i$. The function $\sharesvector : \mathcal{S} \to \mathbb{R}_{\geq 0}^n$  is called a \emph{contract}. 

We consider risk-averse agents\footnote{We will specialize to the case where agents are risk neutral at times. Even in that case, the principal cannot simply ``sell the firm'' due to our limited liability assumption and the presence of multiple agents, so the problem remains interesting.} and a risk-neutral principal.\footnote{The characterization of an optimal contract and its consequences can be extended to the case of a risk-averse principal.} The utility to agent $i$ from a monetary transfer is given by the function $u_{i} : \mathbb{R}_{\geq 0} \to \mathbb{R}_{\geq 0}$, which is strictly increasing, concave, and differentiable. Each agent also has a  cost-of-effort function $C_{i} : \mathbb{R}_{\geq 0} \to \mathbb{R}_{\geq 0}$, which is strictly increasing, strictly convex, and twice differentiable in that agent's action. The marginal cost at zero action is zero: $C_{i}'(0)=0$. Given the contract $\bm{\tau}$, agent $i$ chooses $a_i$ to maximize his expected payoff from transfers minus his cost,
\begin{equation*}
    \mathcal{U}_{i} = \sum_{s \in \mathcal{S}}P_{s}\left(Y(a_i,a_{-i})\right)u_{i}\left(\tau_{i}(s)\right) - C_{i}(a_{i}).
\end{equation*} 

The payoff for the principal given a contract $\sharesvector$ and team performance $Y$ is the expectation of expected revenue minus promised transfers:
\begin{equation*}
    \sum_{s \in \mathcal{S}}\left(v_{s} - \sum_{i}\tau_{i}(s)\right)P_{s}(Y).
\end{equation*}

The timing is as follows: The principal commits to a contract $\sharesvector$, and then agents simultaneously choose actions. Our solution concept for the game among the agents is pure strategy Nash equilibrium; in the remainder of the paper, we simply use the word \emph{equilibrium} to refer to this solution.

There may be multiple equilibria under some contracts. Given a contract $\sharesvector$, we assume that agents play an equilibrium $\eqlbactions(\sharesvector)$ maximizing the principal's expected payoff. Under this selection, a principal's payoff under a contract is well-defined if at least one equilibrium exists. Among such contracts, a contract $\bm{\tau}$ is \emph{optimal} if no other contract $\widetilde{\bm{\tau}}$ gives the principal a higher payoff. Implicit in this definition is the assumption that contracts without equilibria can never be optimal.

Our analysis will not rely on the existence of an optimal contract, but the following argument shows an optimal contract exists if we impose a bit of additional structure.
\begin{fact} Suppose that the space of contracts giving the principal a non-negative payoff is compact.\footnote{We exclude any contracts where no equilibrium exists from this space.} Then an optimal contract exists. \label{fact:opt_exist} \end{fact}

As examples, this holds if the outcome is binary or if the outcome probabilities $P_s(Y)$ are uniformly bounded away from zero. The proof uses a compactness argument, taking a sequence of contracts whose payoffs converge to the supremum of attainable principal payoffs, along with their corresponding equilibria. By compactness of the contract space and action space, we can extract convergent subsequences of both the contracts and equilibrium actions. The limit contract achieves the supremum payoff because equilibria are upper-hemicontinuous in the contract and the principal's payoff function is continuous in both contracts and actions.

\subsection{Simple success-or-failure environments} \label{sec:simple_environments}

We introduce a class of \emph{simple success-or-failure environments} that will be helpful for illustrating our results.

There are two possible outcomes $s \in \{0,1\}$. The revenues from these outcomes are normalized so that $v_{1}=1$ and $v_{0}=0$. These can be interpreted as success or failure of the project. The probability of success is $P(\production)$, where the function $P(\cdot)$ is strictly increasing and twice differentiable.

Each agent has a quadratic cost of effort.\footnote{This is not substantively restrictive in that, under smoothness assumptions, a monotone transformation can be applied to $a_i$ to achieve this form.} Agent $i$ maximizes the expected payoff given by the expression:  $$ \mathcal{U}_{i} = P(Y)u_{i}\left(\tau_{i}(1)\right) + \left(1-P(Y)\right)u_{i}\left(\tau_{i}(0)\right) - \frac{a_i^2}{2}.$$ 

\begin{fact} \label{fact:no_pay_failure} In simple success-or-failure environments, it is optimal for the principal to pay nothing at the failure outcome---that is, $\tau_{i}(0) = 0$ for all agents $i$. A contract can then be represented by a $n$-dimensional vector $\sharesvector \in \mathbb{R}_{\geq 0}^{n}$ consisting of payments in the success outcome. \end{fact}

The reason is simple: agents' incentives depend only on the difference $u_{i}\left(\tau_i(1)\right)-u_{i}\left(\tau_i(0)\right)$ between utility conditional on success and failure, so we can shift payments and assume $\tau_i(0)=0$ and $u_i(0)=0$. This shift can only improve the principal's payoff, so it is without loss of optimality in the principal's problem. We can  interpret $\tau_i(1)$ as an equity share in the project's output.

\subsubsection{Linear-quadratic production functions} \label{ex:main}
Within this class of environments, we can define an important leading example, which we will analyze in more detail after giving our main result.

Consider \emph{linear-quadratic production functions} of the form
\begin{equation}
  Y(\bm{a}) = \sum_{i=1}^n k_i a_i + \frac{1}{2}\sum_{i,j=1}^n G_{ij} a_i a_j,
  \label{eq:lq_general}
\end{equation}
where $k_i$ captures agent $i$'s standalone productivity and the entrywise nonnegative symmetric matrix $\bm{G}$ is the technological complementarity matrix. Its off-diagonal entries capture bilateral complementarities between agents' efforts, while $G_{ii}$ captures curvature in agent $i$'s own action.\footnote{This functional form is a second-order approximation to any smooth production function $Y$, with $\bm{G}$ being the Hessian of $Y$ with respect to actions at equilibrium.}

Fix a concave success probability function $P(\cdot)$. In the appendix, \Cref{l:uniqueeq} shows that a fixed contract $\bm{\tau}$ induces a unique Nash equilibrium. Under those conditions, the equilibrium actions $\bm{a}^*$ and team performance $Y^* = Y(\bm{a}^*)$ solve \begin{equation} [\bm{I}-P'(Y^*) \bm{T} \bm{G}]\eqlbactions = P'(Y^*) \bm{T} \bm{k},\label{eq:Nash_conditions_example} \end{equation} where $\bm{T}$ is a diagonal matrix with entries ${T}_{ii} = u_{i}(\sharesagt{i})$. The equilibrium structure is reminiscent of standard models of network games  with linear best responses (\cite{BCZ-06}), but two nuances complicate the situation: non-linear success probability $P(Y)$, as well as the appearance of the endogenous object $\bm{T}$ as part of the matrix $P'(Y^*)\bm{TG}$ that captures incentive spillovers.

An instructive special case arises when the symmetric matrix $\bm{G}$ has rank one: 
$\bm{G} = \rho \bm{\beta}\bm{\beta}^\top$ for some nonnegative vector $\bm{\beta}$ and scalar $\rho > 0$. In this case, the production function can be written as
\begin{equation}
  Y(\bm{a}) = \sum_{i=1}^n k_i a_i + \frac{\rho}{2}\left(\sum_{i=1}^n \beta_i a_i\right)^2,
  \label{eq:production_tech}
\end{equation}
where $\beta_i$ can be interpreted as the degree to which agent $i$'s effort complements others', and $\rho$ governs the overall strength of complementarities. 

Both the case with general $\bm{G}$ and the rank-one case will be important for applications.

\subsection{Remarks on the model} \label{sec:remarks_model}

The team performance $Y$ is real-valued, but the outcome $s$ is discrete. This assumption need not substantively restrict the scope of the model since the outcome can be, for example, a revenue rounded to the nearest cent. On the other hand, the fact that outcomes are mediated by a one-dimensional performance level is important. We discuss in the paper's concluding remarks the issue of extending the analysis to outcomes determined by a higher-dimensional function of efforts.

We assume that the firm's output is the only contractible consequence of any agent's effort. In other words, agents cannot be paid directly for their efforts $a_i$. In this we follow the literature stemming from \citet{holmstrom1982moral}, which is motivated by the practical features of contracts and the fact that $a_i$ and $Y$ are abstractions that may not have any measurable real-world counterpart, among other considerations.

We do not impose restrictions on the space of contracts beyond non-negativity of payments. (In particular, we allow the principal to lose money at some outcomes with the hope of inducing a higher action.) In practice, compensation may be restricted to particular simpler classes of contracts. For example, equity contracts award an agent the same share of profit in every state. In many such settings, analogues of our balance condition can be derived, as we demonstrate for the case of equity contracts in \Cref{a:optimalequitypay}.

We do not assume that equilibria exist under all contracts in the formulation of the model or our analysis. For the contract that pays zero in all states, there is always an equilibrium in which agents take zero actions; this ensures that the principal's value is well-defined. For other contracts, existence needs to be analyzed in the environment of interest; for example, \Cref{l:uniqueeq} establishes existence and uniqueness in the parametric model of \Cref{ex:main}.

\section{Optimal Contracts}
\label{s:generalintensivemargin}

This section presents and interprets the characterization of optimal contracts, beginning by defining its ingredients.

\subsection{Key objects} \label{ss:notation} We start with some notation. Fix a contract $\sharesvector$ and a corresponding principal-optimal equilibrium $\bm{a}^*(\sharesvector)$. Let $Y^*$ be the team performance under this equilibrium. 

\begin{definition}[Marginal productivity] Let $\nabla Y(\eqlbactions)$ be the gradient of $Y(\cdot)$ at $\eqlbactions$, restricted to the agents that take a strictly positive action. We define the \textit{marginal productivity} vector $\bm{\productivity}$ as \begin{align*}
    \bm{\productivity} := \nabla Y(\eqlbactions).
\end{align*} 
\end{definition}

The entry $\productivity_{i}$ captures the marginal effect of $i$'s action on team performance. We use \emph{productivity} as shorthand for marginal productivity. Here, as in other cases, the quantities in the expression for $\bm{\alpha}$ are evaluated at the equilibrium $\bm{a}^*$, under a certain contract $\bm{\tau}$---but we often omit this dependence on the equilibrium and the contract notationally for brevity.

\bigskip

We turn next to centrality, which concerns how incentives propagate through the team. Let $\bm{G}$ denote the \emph{technological complementarity matrix}: the Hessian matrix of $Y$ with respect to agent actions, restricted to agents that take a strictly positive action in $\eqlbactions$. Formally, for agents $j$ and $k$ such that $a_{j}^*>0$ and $a_{k}^*>0$, define \begin{align*}
    G_{jk} := \frac{\partial^2 Y}{\partial a_{k} \partial a_{j}}.
\end{align*} Let the \textit{marginal payment utility} matrix $\bm{U}$ be a diagonal matrix where \begin{align*}
    U_{jj} := \sum_{s \in \mathcal{S}}P_{s}'(Y^*)u_{j}(\tau_{j}(s))
\end{align*} is the marginal change in agent $j$'s utility from payments when team performance increases. The increase in $Y$ changes all the probabilities of outcomes, and the agent's utility from these outcomes is given by $u_j(\tau_j(s))$, where the contract is held fixed.

\begin{definition}[Spillover matrix and centrality]
Let $\bm{H}$ be the diagonal matrix with $H_{jj} := C_{j}''(a_{j}^*)$. Whenever $\bm H$ is nonsingular, define the \emph{spillover matrix}
\begin{equation*}
    \bm S:=\bm H^{-1}\bm U\bm G.
\end{equation*}
Whenever $\bm I-\bm S$ is nonsingular, the \emph{centrality} vector is
\begin{equation} \label{eq:centrality_definition}
\bm{\centrality}^{T} := \bm{\productivity}^{T}\left[\bm{I}-\bm S \right]^{-1}.
\end{equation}
\end{definition}
The entry $\centrality_{i}$  can be thought of as the {total effect on team performance} induced by a marginal change in agent $i$'s  incentive to increase $a_i$. This effect is inclusive of  spillovers on others' incentives through strategic interactions. To explain the role of $\bm{H}$: when we analyze how equilibrium actions vary with contract perturbations, an agent's best response is less sensitive to incentives when $C''_j(a_j^*)$ is larger, and $\bm{H}$ allows us to capture this effect. Centrality is well-defined when the inverse in the formula exists, which will be ensured under the conditions of our results.

Productivity and centrality take simpler forms in simple success-or-failure environments. In that case, the matrix $\bm{H}$ is equal to the identity matrix $\bm{I}$ because the cost of each agent's action is $\frac{1}{2}a_i^2$. So

\begin{equation}\label{eq:prodcentsf}\alpha_i = \frac{\partial Y}{\partial a_i} \quad \text{ and } \quad
\bm{\centrality}^\top = \bm{\productivity}^\top\left[\bm{I}- P'(Y) \bm{T} \bm{G}\right]^{-1},
\end{equation}
where $\bm{T} = \operatorname{diag}\!\left(u_1(\tau_1), \ldots, u_n(\tau_n)\right)$.  If we further specialize to the rank-one setting of \Cref{ex:main}, then the technological complementarity matrix is $\bm{G} = \rho \bm{\beta}\bm{\beta}^\top$ and productivities are
$$\alpha_i = \frac{\partial Y}{\partial a_i} = k_i + \rho \beta_i x^*,$$
where $$x^* := \bm{\beta}^\top \bm{a}^* = \sum_{j=1}^n \beta_j a_j^*$$ is the complementarity-weighted sum of equilibrium actions.

\subsection{Balance condition across agents} 
\label{ss:balancegeneral}

In this section, we present our main result: a balance condition across agents at each outcome realization.

We will state a necessary condition for an optimal contract $\bm{\tau}^*$ under the following assumption, which will be maintained from now on.
\begin{assumption} \label{as:balancederive} A differentiable selection $\eqlbactions(\sharesvector)$ from the equilibrium correspondence can be defined in a neighborhood of $\optsharesvector$.
\end{assumption}
This assumption stipulates that equilibrium varies differentiably as we slightly perturb the contract in a neighborhood of the optimum.  \Cref{a:appendix_indifferences} gives sufficient conditions for the applicability of the first-order approach, including explicit ones on the primitives that extend those familiar from single agent problems  \citep{rogerson1985first}. (See also 
\Cref{s:applications} for specific production functions $Y(\bm{a})$ under which the first-order approach applies.) 
The one-dimensional aggregator $Y(\bm{a})$ plays a crucial role here. Whereas in general providing sufficient conditions for the first-order approach in arbitrary multi-agent settings is quite challenging, in our framework it is much more tractable.

The following result characterizes optimal contracts.
\begin{theorem}
\label{t:generalmodel}
    Suppose $\optsharesvector$ is an optimal contract and $Y^*$ is the induced team performance. There exist constants $\balanceconstant_{s}$ such that for any agent $i$ receiving a positive payment under an outcome $s$, we have \[ \label{eq:balancecondition} \productivity_{i}\centrality_{i} \cdot \frac{u_{i}'(\tau_{i}^*(s))}{C''_i(a_i^*)} = \balanceconstant_{s}\tag{BC}.\]Moreover, the outcome-dependent constants $\balanceconstant_s$ satisfy $\balanceconstant_{s} \propto \frac{P_{s}(Y^*)}{P_{s}'(Y^*)}.$
\end{theorem}

This result says that optimal incentives require balance to hold, with the product on the left being equal across agents. More informally, the balance condition (\ref{eq:balancecondition}) states that under an outcome $s$,
$$\text{productivity}_i \cdot \text{centrality}_i \cdot \text{responsiveness}_i \text{ } = \text{ constant}$$
for all agents paid when that outcome occurs. Below, we will give more intuition for why this is a necessary condition.

In simple success-or-failure environments with all agents having utility function $u_i(\tau)=\tau$, we obtain a simpler statement: there exists a constant $\lambda$ such that
$$\alpha_i c_i  = \lambda$$
for all agents $i$ receiving positive payments conditional on success. Recall that in this case $\alpha_i$ and $c_i$ take the simpler forms given in (\ref{eq:prodcentsf}).

We note that the proof does not rely on the induced team performance $Y^*$ being optimal. The balance condition at the optimal contract would hold if the principal instead wanted to implement any desired level of performance with minimal (expected) transfers to agents. Solving for the minimal payments implementing a given outcome distribution  and then optimizing over outcome distributions is a standard approach in single-agent settings (e.g., \cite{grossman1983implicit}). With multiple agents, the first step also involves allocating incentives across agents.

The key to the proof, formalized in the following lemma, is calculating the effect on team performance of increasing an agent's payment under a given outcome. \Cref{as:balancederive} ensures that these perturbations are well-defined.

\begin{lemma}
    \label{l:changeteamperformance}
    Suppose $\optsharesvector$ is an optimal contract with corresponding equilibrium actions $\eqlbactions$ and team performance $Y^*$. %
    Consider any agent $i$ receiving a positive payment at some outcome. For any outcome $s$, the derivative of team performance in $\tau_{i}(s)$, evaluated at $\bm{\tau}^*$, is \begin{align*}
      \frac{d Y}{d \tau_{i}(s)} = l P'_{s}(Y^*) \productivity_{i} \centrality_{i}\cdot \frac{ u_{i}'(\tau^*_{i}(s))}{C''_i(a_i^*)} ,
    \end{align*}  where $l$ is independent of $i$ and $s$.
\end{lemma}

A complete proof for the result above is provided in \Cref{a:detailedproofs}. We provide some intuition for the various terms in the expression in the lemma.

\medskip

\textit{Intuition for the proof:} The lemma characterizes the effect on team performance of increasing the transfer to agent $i$ under outcome $s$. We can decompose this effect as the product of: 
\begin{enumerate}[(i)]
\item a factor $P'_{s}(Y^*) \productivity_{i}$ capturing the sensitivity of the probability of the outcome $s$ to $i$'s effort; 
\item an action-responsiveness factor $\frac{ u_{i}'(\tau_{i}(s))}{C''_i(a_i^*)}$: the numerator governs how the transfer perturbation changes $i$'s marginal incentive, and the denominator converts that marginal-incentive change into a direct action response;
\item a term $\centrality_{i}$ capturing the spillovers from changing $i$'s incentive to exert effort; 
\item the constant $l$, which depends on the curvature of the probability $P_s(Y)$. 
\end{enumerate}

We focus on the first three factors and defer treatment of the fourth term, which is not central in the basic intuition, to the formal proof in the appendix.

The first two factors, (i) and (ii), measure the principal's ability to directly induce agent $i$ to change his action by rewarding him when outcome $s$ is realized. The change in agent $i$'s marginal payoff from effort as $\tau_i(s)$ increases is $P_s'(Y^*)\frac{\partial Y}{\partial a_i}u_i'(\tau_i(s))$. Dividing by $C_i''(a_i^*)$ converts this marginal-incentive change into the direct change in $i$'s optimal action. An agent is easier to incentivize when his marginal utility of money is higher and when his marginal cost of effort increases less steeply.

Multiplying by term (iii), the centrality $\centrality_i$, translates from this direct effect on $i$'s action to the overall change in equilibrium team performance. This effect is at the heart of our use of network theory; a detailed intuition for this appears in \Cref{sec:ripples} just below.

\medskip

First, though, we sketch why \Cref{l:changeteamperformance} implies \Cref{t:generalmodel}. (A formal proof is provided in \Cref{a:detailedproofs}.) We want to show that the balance condition \begin{equation*}
    \productivity_{i}\centrality_{i} \cdot \frac{u_{i}'(\tau_{i}^*(s))}{C''_i(a_i^*)} = \productivity_{j}\centrality_{j} \cdot \frac{u_{j}'(\tau_{j}^*(s))}{C''_j(a_j^*)},
\end{equation*} must hold under an optimal contract. Suppose that the principal would benefit from a slightly higher team performance (the case  in which the principal prefers a slightly lower team performance proceeds analogously). \Cref{l:changeteamperformance} shows that the change in team performance from increasing agent $i$'s payment under outcome $s$ is equal to $ \productivity_{i}\centrality_{i}\cdot \frac{u_{i}'(\tau_{i}^*(s))}{C''_i(a_i^*)}$ times terms independent of $i$, and similarly for agent $j$. If we had\begin{equation*}
    \productivity_{i}\centrality_{i}\cdot \frac{u_{i}'(\tau_{i}^*(s))}{C''_i(a_i^*)} > \productivity_{j}\centrality_{j}\cdot \frac{u_{j}'(\tau_{j}^*(s))}{C''_j(a_j^*)},
\end{equation*} it would be profitable for the principal to pay agent $i$ slightly more and agent $j$ slightly less under outcome $s$. The same argument holds in the opposite direction, so the balance condition is necessary for the contract to be optimal.

\medskip

\subsubsection{What centrality captures} \label{sec:ripples}

This section gives an intuition for the role of network centrality in our characterization, and the details of how network centrality is applied in our setting. For simplicity, we focus on the simple success-or-failure case, with $\bm{H}=\bm{I}$ (which is essentially a normalization), and recall $\bm S=P'(Y)\bm T\bm G$ is the spillover matrix induced by the complementarity matrix $\bm G$.

When the spectral radius of $\bm S$ is less than 1, we can write $$\bm{\centrality}^\top = \bm{\productivity}^\top \left[\sum_{\ell=0}^{\infty} \bm S^\ell\right].$$ Economically, the powers capture the effects of the initial increase in $i$'s action ($\ell=0$), the resulting changes in each agent's best response ($\ell=1$), the further changes in best responses induced by these, etc. The entries of $\bm S$ describe how $j$'s increased effort spills over to affect $i$'s marginal incentives, accounting for technological complementarities and contracts. The form of spillovers makes sense: the direct dependence of $i$'s best response on $j$'s action is proportional to $P'(Y)$---the sensitivity of success probability to team performance $Y$; $u_{i}\left(\tau_i\right)$---agent $i$'s utility conditional on success; and to the complementarity of $i$ with $j$. Unlike in standard models of network games with linear best-responses, the spillover matrix $\bm S$, and therefore the resulting centrality measure, depend endogenously on the contract and equilibrium.

In terms of graph statistics, entry $(i,j)$ of $\bm S^\ell$ counts walks of length $\ell$ from $j$ to $i$ in the network of spillovers, with each step in the walk from a node $j'$ to a node $i'$ contributing a factor $S_{i'j'}$ to the weight of a walk.\footnote{See \citet{BCZ-06} and \citet*{bloch2023centrality} for more on these centrality measures.} To compute centralities, the entries of these powers are then multiplied by productivities $\alpha_i$: an agent $j$ is more central if the ripple effects of motivating $j$ increase the incentives of more productive agents.

\subsubsection{Further comments on the result and the modeling behind it. }
As our discussion makes clear, \Cref{t:generalmodel} essentially follows from a decomposition of the terms in the first-order conditions for contract optimality. One contribution of our work, as we have already mentioned, is formulating a framework in which the first-order approach is valid under conditions that extend those justifying the first-order approach from single-agent contracts  (\Cref{a:appendix_indifferences}).

Another contribution is formulating the framework so that these conditions coming from the first-order approach are as simple and intuitive as possible. In particular, if we simply differentiate $i$'s incentives in $j$'s actions, second derivatives of $P(Y)$ show up. But such terms do not appear in our definition of centrality. In the proof of our main result, we show that changing the second derivative $P''(Y^*)$ rescales all relevant expressions involving the spillover matrix $\bm S$ by a common factor, and so second derivative terms can be ignored when  checking whether redistributing incentives between agents is profitable. This simplification  relies crucially on the model feature that agents' actions contribute to production via one-dimensional team performance $Y$.

In the remaining sections of the paper, we build on the insights from the implicit characterization and derive explicit connections between economic primitives and optimal contracts when the production function takes various functional forms. 

\subsection{Comparisons across agents and outcomes.}
\label{s:generalmarginalutils}

Before turning to these more applied consequences of \Cref{t:generalmodel}, we give several basic corollaries describing how incentives are distributed across agents and across outcomes. We first give conditions implying an ordinal ranking of payments to agents that does not depend on the particular outcome. We then provide a multi-agent version of a result from \cite{holmstrom1979moralhazard} on the optimal allocation of payments across outcomes rather than agents.

 \subsubsection{Ranking agents.} Agents can be ranked in terms of payments at the optimal contract. To see this, we establish a relationship between the marginal utilities of agents. An implication of \Cref{t:generalmodel} is that the ratio between any two agents' marginal utilities is the same at every outcome for which both receive positive transfers.

\begin{corollary}
\label{c:marginalutil1}
Consider an optimal contract $\optsharesvector$. Let $\mathcal{S}_{ij}^*$ be the set of outcomes at which agents $i$ and $j$ both receive a positive payment. For any outcome $s \in \mathcal{S}^*_{ij}$, we have \begin{align*}
        \frac{u_{i}'(\tau^*_{i}(s))}{u_{j}'(\tau^*_{j}(s))} = \frac{\productivity_{j}\centrality_{j}}
        {\productivity_{i}\centrality_{i} } \cdot \frac{C''_i(a_i^*)}{C''_j(a_j^*)}.
    \end{align*}
\end{corollary}

Intuitively, since outcome probabilities are determined by a joint team performance, agents' incentives should vary across outcomes in similar ways. The corollary formalizes this intuition in terms of marginal utilities in each outcome.\footnote{This contrasts with a literature on optimal compensation when the observed outcome can be used to identify individuals who deviated from a desired level of effort (e.g., \cite{holmstrom1982moral} and \cite{legros1993efficient}).} %

The corollary only applies when the set of outcomes at which agents $i$ and $j$ both receive a payment is non-empty. Determining when an agent is paid at a given outcome can be complicated in general, but it is easy to construct settings where the corollary applies. In \Cref{a:sufficientconditionspay}, for example, we give a class of environments in which an Inada condition guarantees that all agents are paid at all outcomes where $P'_s(Y^*)>0$ (and no other outcomes). 

When agents have identical utility functions, agents can be ranked so that an optimal contract provides stronger incentives to more highly ranked agents.

\begin{proposition}
\label{p:rankingpayments}
    Suppose that $\optsharesvector$ is an optimal contract. If a pair of agents $i$ and $j$ have identical strictly concave utility functions $u_{i}(\cdot) = u_{j}(\cdot)$, then $$
        \tau^*_{i}(s) \geq \tau^*_{j}(s) \text{ for all } s \in \mathcal{S} \text{ or }
        \tau^*_{j}(s) \geq \tau^*_{i}(s) \text{ for all } s \in \mathcal{S}$$  (or both).
\end{proposition}

The intuition is simple: for two agents that derive the same value from a monetary transfer, the agent with a greater overall effect on team performance at the optimal contract must be receiving a higher payment.

When all agents have an identical utility function, the optimal contract induces a complete ranking on the agents. The relative magnitude of payments across agents depends on the environment. This becomes evident in the parametric example discussed further in \Cref{s:applyparametric}.

\subsubsection{Allocation across outcomes.} \label{ss:outcomespositivepayments} A second implication of the main balance result is a relationship between a single agent's marginal utility across outcomes. 

\begin{corollary}
\label{c:marginalutil2}
    Suppose $\optsharesvector$ is an optimal contract and $Y^*$ is the induced team performance. If agent $i$ receives positive payments under outcomes $s_1$ and $s_2$, then \begin{align*}
        \frac{u_{i}'(\tau_{i}^*(s_{1}))}{u_{i}'(\tau_{i}^*(s_{2}))} = \frac{P_{s_{1}}(Y^*)}{P'_{s_{1}}(Y^*)} \cdot \frac{P'_{s_{2}}(Y^*)}{P_{s_{2}}(Y^*)}.
    \end{align*}
\end{corollary}

The corollary states that the marginal utility under each outcome is proportional to the probability of that outcome divided by the marginal change in that probability as team performance increases. That is, agents are paid more in outcomes that are less likely and more responsive to team performance. This result generalizes a result in the single-agent setting of \citet{holmstrom1979moralhazard} concerning how a single agent's payments should be allocated across outcomes.

A straightforward application of \Cref{c:marginalutil2} characterizes the set of outcomes at which an agent receives a positive payment. If an agent receives a positive payment at some outcome, the outcomes at which he receives a positive payment must all either have a positive marginal probability at equilibrium team performance, or a negative marginal probability. When the team performance function $Y(\cdot)$ is strictly increasing in each of its arguments, the outcomes at which an agent receives a positive payment all have a positive marginal probability at equilibrium team performance. (This is formalized as \Cref{l:positivepaymentoutcomes} in the Appendix).

In the special case that an agent is risk-neutral, the corollary has a stronger implication: agents are only paid under the outcome(s) maximizing $P'_s(Y^*)/P_s(Y^*)$. It is straightforward to construct functions $P_s(Y)$ where there can only be one such outcome. Risk neutrality removes the principal's motives to diversify payments across outcomes, so compensation can be concentrated in the outcome providing the most incentive power.

\section{Implications: Whom to incentivize in teams?}
\label{s:applications}

The balance condition in \Cref{t:generalmodel} characterizes optimal contracts 
implicitly. We now derive explicit solutions and comparative statics in two 
tractable settings that illuminate how exogenous parameters determine the balance condition and the resulting contracts.

First, we 
study a setting in which agents differ in their intrinsic capacities for productivity and complementarity. A single complementarity parameter governs network interactions with all agents, yielding a rank-one network. This exercise allows us to examine which of these exogenous traits should primarily drive compensation.

A second application focuses on better understanding centrality in general networks. By turning off differences in standalone productivity, we can explicitly characterize how a heterogeneous technological complementarity matrix shapes agents' centralities at the optimal contract. We find some interesting non-monotonicities that sharply distinguish our theory from standard network models. This illustrates both the importance and the subtlety of the fact that centralities are evaluated at endogenously determined responses to optimal contracts, rather than in an exogenous network.

\subsection{Rank-one networks: Standalone productivity versus complementarity}
\label{ss:rankone}

We return to the rank-one, linear-quadratic setting of \Cref{ex:main}. 
This setting allows heterogeneity in both standalone productivities $k_i$ and 
complementarity parameters $\beta_i$, while maintaining tractability through the 
rank-one structure of the Hessian $\bm{G} = \rho \bm{\beta}\bm{\beta}^\top$. Throughout this section, we maintain the assumption that the parameters satisfy the bounds $k_{i} \in [\underline{k},\overline{k}]$ for some $\overline{k}>\underline{k} > 0$ and $\beta_{i} \in [\underline{\beta},\overline{\beta}]$ for some $\overline{\beta}>\underline{\beta} > 0$. 

\subsubsection{Characterization of optimal contracts}

Consider any contract $\bm{\tau}$ with corresponding equilibrium actions $\bm{a}^*$ 
and team performance $Y^*$. Define the \emph{spillover matrix}
\begin{equation*}
    \bm{S} := \rho P'(Y^*) \bm{T} \bm{\beta} \bm{\beta}^{\top},
\end{equation*}
where $\bm{T} = \diag\!\left(u_1(\tau_1), \ldots, u_n(\tau_n)\right)$ collects 
utilities from the success payments (with $\tau_i=\tau_i(1)$, recalling \Cref{fact:no_pay_failure}). The spectral radius 
$\mu(\bm{S})$ of this matrix determines whether standalone productivity or complementarity
governs optimal pay.

A key result, which follows from our general balance condition, characterizes optimal contracts in terms of a quadratic form in the primitive parameters.

\begin{lemma} \label{l:optprodcomp}
   Consider any optimal contract $\bm{\tau}^*$ in the rank-one setting of \Cref{ex:main}. 
   There exist constants $B_{1} \geq 0$ and $B_{2} \geq 0$ such that for all agents $i$ receiving positive pay,
   \begin{align*}
       u_{i}'\left(\tau^*_{i}\right) \propto \left(k_{i}^2 + B_{1}k_{i}\beta_{i} + B_{2}\beta_{i}^2\right)^{-1}.
   \end{align*}
\end{lemma}

This lemma shows that optimal pay depends on a quadratic form in the agent's 
standalone productivity $k_i$ and complementarity $\beta_i$. The coefficients $B_1$ and
$B_2$ depend on the equilibrium and determine the relative importance 
of standalone productivity versus complementarity in setting compensation.

\subsubsection{Spillover regimes}

The following theorem establishes that the strength of spillovers determines 
which characteristic---standalone productivity or complementarity---dominates optimal pay.

\begin{theorem} \label{t:spilloverregimes}
   In the rank-one setting of \Cref{ex:main}, fix all primitives other than the complementarity parameter $\rho$. Consider a sequence of complementarity parameters $\rho^{(m)}$ with corresponding optimal contracts $\bm{\tau}^{(m)}$ and induced spillover matrices $\bm{S}^{(m)}$. Write $\mu^{(m)}$ for the spectral radius of $\bm{S}^{(m)}$. For any pair of agents $i$ and $j$ that receive positive pay for all sufficiently large $m$,
   \begin{enumerate}[(i)]
       \item If $\mu^{(m)} \to 0$, then
       $\displaystyle\lim_{m\to\infty}\frac{u_i'(\tau_i^{(m)})}{u_j'(\tau_j^{(m)})}=\frac{k_j^2}{k_i^2}$.
       \item If $\mu^{(m)} \to 1$, then
       $\displaystyle\lim_{m\to\infty}\frac{u_i'(\tau_i^{(m)})}{u_j'(\tau_j^{(m)})}=\frac{\beta_j^2}{\beta_i^2}$.
   \end{enumerate}
\end{theorem}

The theorem establishes a sharp dichotomy.
When spillovers are weak (low $\mu$), 
the propagation of incentives across agents is limited, so pay primarily reflects 
direct productivity $k_i$. When spillovers are strong (high $\mu$), the network 
of complementarities dominates, and an agent's position in the interaction 
structure---captured by $\beta_i$---becomes the primary determinant of compensation. We note that the spectral radius $\mu^{(m)}$ varies endogenously with the complementarity parameter $\rho^{(m)}$, fixing the rest of the environment.

Beyond delineating these regimes, \Cref{t:spilloverregimes} reveals that 
heterogeneity in productivity or complementarity is, in a sense, \emph{magnified} under 
optimal contracts: the ratio of marginal utilities depends on the \emph{square} 
of the ratio of the relevant parameters. This is a consequence (a non-obvious one, in our view) of the specific type of incentive problem we have here, where the proceeds of a single output must be shared out to compensate the agents.

The characterization in \Cref{l:optprodcomp} suggests a regression specification 
for studying compensation in teams. If $k_i$ and $\beta_i$ can be measured (or 
proxied), then optimal pay should depend on a quadratic form in these variables. 
\Cref{t:spilloverregimes} further predicts that the relative weights should vary systematically with the strength of 
complementarities: in industries or teams with stronger spillovers, the $\beta$ 
terms should receive more weight.

We provide a numerical example to illustrate how the optimal contract's behavior depends on the strength of spillovers, as described above. This effect is particularly relevant in environments where the team includes members with high standalone productivities and weak complementarities. 

We define a team of four agents, partitioned evenly into two groups $\mathcal{T}_{1}$ and $\mathcal{T}_{2}$. Agents in  group $\mathcal{T}_{1}$ have standalone productivity $k = 0.3$ and complementarity $\beta = 0.15$, whereas agents in group $\mathcal{T}_{2}$ have standalone productivity $k = 0.15$ and complementarity $\beta = 0.3$. So the first group has higher standalone productivity and weaker complementarity. The probability of success is an exponential function: \begin{align*}
    P(Y) = 1-e^{-Y} \quad \text{for all } Y \geq 0.
\end{align*} Each agent obtains utility $
    u_{i}(\tau_{i}) = \sqrt{\tau_{i}}$ from a transfer $\tau_i$.

\Cref{subfig:specradvary} plots the spectral radius of the spillover matrix induced by the optimal contract as $\rho$ varies. By symmetry, the two agents in a group $\mathcal{T}_{i}$ receive identical pay under an optimal contract. \Cref{subfig:ratpayrho} illustrates how the ratio of payments to agents in $\mathcal{T}_{1}$ relative to those in $\mathcal{T}_{2}$ varies with $\rho$. 

As depicted in \Cref{subfig:ratpayrho}, when $\rho$ is close to $0$, cross-sectional variation in pay is driven primarily by differences in standalone productivities: agents in $\mathcal{T}_1$ receive approximately sixteen times the payment of agents in $\mathcal{T}_{2}$. For intermediate values of $\rho$,  payments at the optimal contract are approximately equal across agents. For large values of $\rho$, it is optimal to pay agents in $\mathcal{T}_2$ much more than agents in $\mathcal{T}_1$.

\begin{figure*}[t!]
    \centering
\begin{subfigure}[t]{0.5\textwidth}
        \centering
        \includegraphics[width=\textwidth]{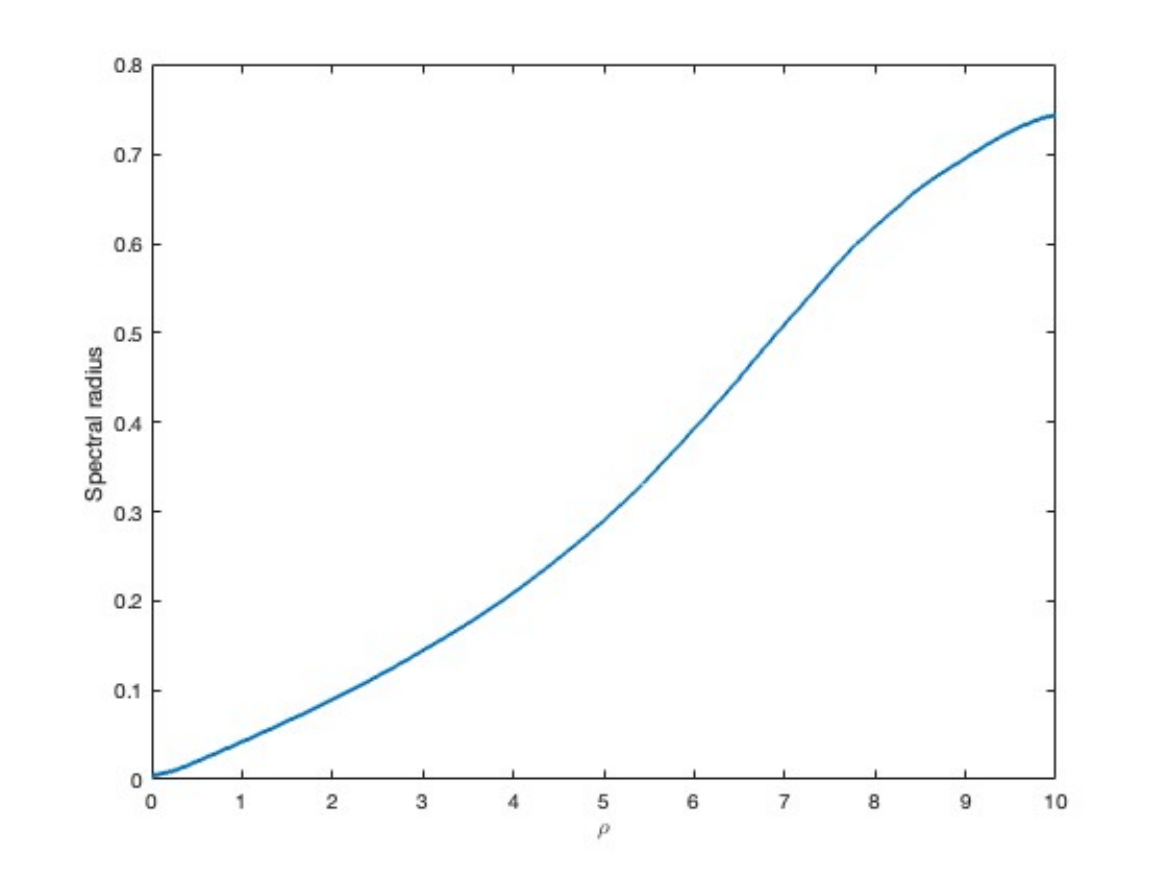}
        \caption{Spectral radius of the spillover matrix}
        \label{subfig:specradvary}   
    \end{subfigure} 
    ~
    \begin{subfigure}[t]{0.5\textwidth}
        \centering
        \includegraphics[width=\textwidth]{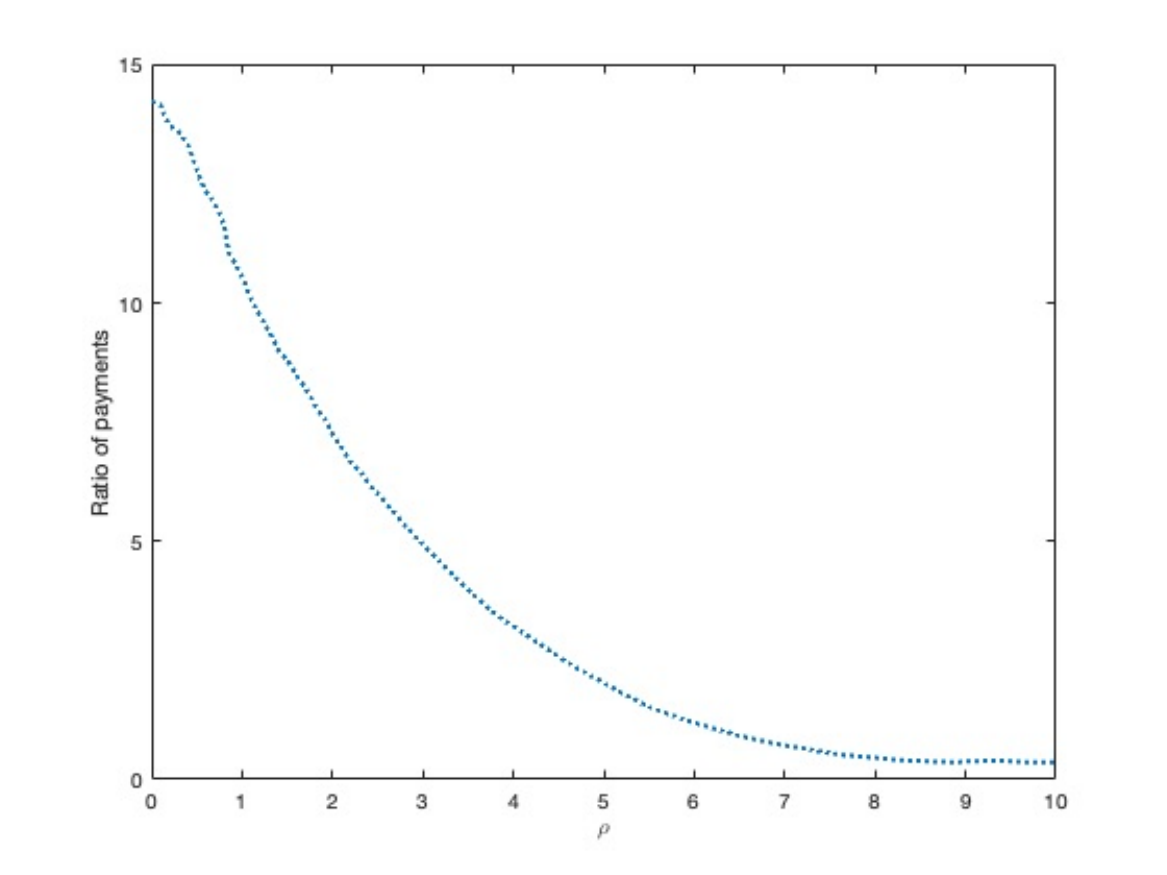}
    \caption{Ratio of payments}
    \label{subfig:ratpayrho}
    \end{subfigure} 
    \label{fig:rank1example}
    \caption{Panel (a) displays the spectral radius of the spillover matrix as $\rho$ varies over the range $[0,10]$. Panel (b) reports the ratio of payments to agents in group $\mathcal{T}_{1}$ relative to those in group $\mathcal{T}_{2}$ at the optimal contract as a function of $\rho$.}
\end{figure*}

\subsubsection{The importance of spillovers in contract design} The preceding analysis describes how the structure of an optimal contract changes with spillovers. To underscore the importance of accounting for spillovers, we now compare the principal's profits under the optimal contract to profits under a misspecified contract neglecting spillovers.

Formally, the misspecified contract is the optimal contract if $\rho = 0$ (and all other parameters are unchanged). This misspecified contract equalizes 
$    k_{i}^2u_{i}'(\tau_{i}^h)$ across agents. In contrast, as established in \Cref{l:optprodcomp}, the correctly-specified optimal contract $\bm{\tau}^*$ equalizes $$
    \left(k_{i}^2 + B_{1}k_{i}\beta_{i} + B_{2}\beta_{i}^2\right) u_{i}'(\tau_{i}^*)
 $$ across agents for some constants $B_{1}>0$ and $B_{2}>0$. Comparing profits across these contracts therefore isolates the cost to the principal of ignoring spillover effects.

We perform this comparison in the setting of the numerical example in the preceding section with global complementarity parameter $\rho=3$. We modify the probability of success to be an exponential function \begin{align*}
    P(Y) = 1-e^{-\xi Y} \quad \text{for all } Y \geq 0
\end{align*} with a flexible parameter $\xi \geq 0$ (which was set to $1$ in the previous example).

\Cref{fig:hillclimbratio} depicts the ratio of profits at an optimal contract $\bm{\tau}^*$ to profits at the benchmark contract $\bm{\tau}^h$. The optimal contract significantly outperforms the benchmark as long as failures are not too rare. The gap grows very large when the probability of success becomes a steeper function of team performance (large $\xi$).

\begin{figure}[t!]
    \centering
    \includegraphics[width=0.5\textwidth]{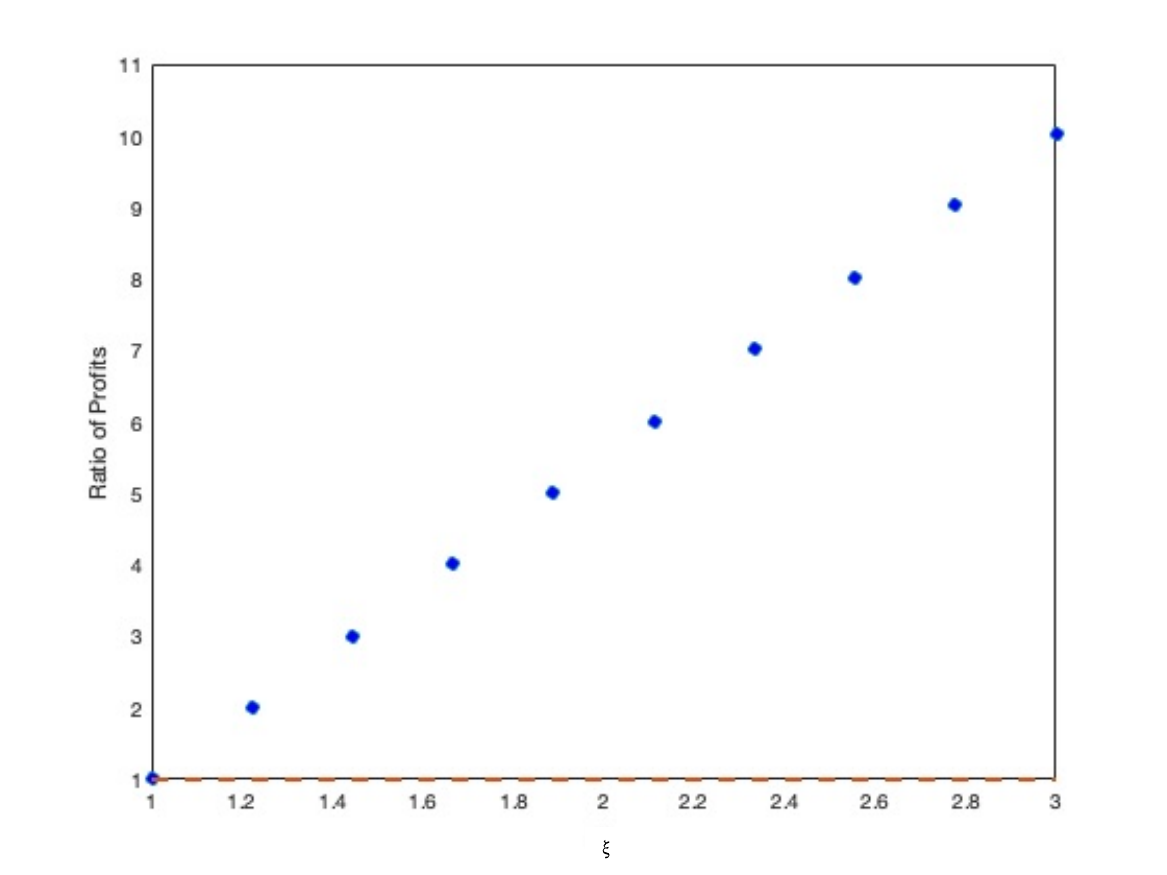}
    \caption{The ratio of profits at an optimal contract and profits at a benchmark contract which ignores spillovers. The red dashed line indicates a ratio of $1$.}
    \label{fig:hillclimbratio}
\end{figure}

\subsection{General networks}
\label{s:applyparametric}

The rank-one setting in \Cref{ss:rankone} allowed rich individual heterogeneity but restricted the complementarity structure. Ideally, we would have a similar level of insight into how the endogenous moving parts of our balance condition depend on primitives with arbitrary network structures $\bm{G}$. Unfortunately, we do not know how to obtain explicit characterizations in the linear-quadratic setting  with general networks and heterogeneous standalone productivities.

Considerable insight can be gained, however, from a natural special case. Staying within the linear-quadratic success-or-failure setting of \Cref{ex:main}, we make all standalone productivities 
equal and assume risk-neutrality, i.e. $u_i(\tau_i) = \tau_i$, for all agents. This setting, with an arbitrary $\bm{G}$, will let us get some traction on how network structure alone shapes optimal contracts.

 Consider any contract $\bm{\tau}$ with corresponding equilibrium actions $\bm{a}^*$. At equilibrium, an agent $i$ chooses a positive action if $\tau_i>0$ and an action of zero otherwise. An agent is said to be \emph{active} under a given contract $\sharesvector$ if he receives a positive payment $\tau_{i} > 0$ and \emph{inactive} otherwise.\footnote{The extensive margin leads to a rich set of questions---which agents are active at an optimal contract. We leave these questions for future work. } We will focus on characterizing the optimal contract payments and equilibrium actions among active agents.

\begin{proposition} \label{t:optsharescharacterization} In the setting of this section, suppose $\sharesvector^*$ is an optimal contract and $\eqlbactions$ and $Y^*$ are the induced equilibrium actions and team performance, respectively. The following properties are satisfied:
    \begin{enumerate}
        \item For any two active agents $i$ and $j$, we have $\productivity_i = \productivity_j$ and $\centrality_i = \centrality_j$.
           \item Balanced neighborhood actions: There is a constant $\balanceconstant'>0$ such that for all active agents $i$, we have $(\network \bm{a}^*)_i = \balanceconstant'$.
        \item Balanced neighborhood equity: There is a constant $\balanceconstant>0$ such that for all active agents $i$, we have $(\network \sharesvector^*)_i = \balanceconstant$.
    \end{enumerate}
\end{proposition}

The first part of the result states that all active agents have equal productivities and equal centralities. This is much stronger than what the balance condition guarantees, which is that $\alpha_i c_i$ is constant across active agents $i$.

The property of balanced neighborhood actions states that for each active agent $i$, the sum of actions of neighbors of $i$, weighted by the strength of $i$'s connections to those neighbors in $\network$, is equal to the same number, $\balanceconstant'$.  Similarly, the property of balanced neighborhood equity says that for each active agent $i$, the sum $\sum_j G_{ij} \tau_j$ of shares given to neighbors of $i$, weighted by the strength of $i$'s connections to those neighbors in $\network$, is equal to the same number (i.e., is not dependent on $i$).

Why must neighborhood action balance hold? Intuitively, if balance failed and an agent $i$ had strong links and neighbors taking high actions, it would be profitable to reallocate incentives from other agents to $i$. (This relies on equal standalone productivities and did not need to hold in \Cref{ss:rankone}.)  One implication of this balance is that differences in the exogenous complementarity structure are attenuated by the principal's optimization. If an agent $i$ has stronger links $G_{ij}$, then to satisfy balanced neighborhood actions, i.e. holding $\sum_j G_{ij} a_j$ constant, $i$'s neighbors must take lower actions. This ``balancing'' force diminishes the complementarities that $i$ experiences. So, while the example here has a lot of structure---and, in particular, the equality of standalone productivity plays an important role---it cleanly makes the point that comparative statics as we vary ``complementarity parameters'' in general networks are subtle.

In fact, the structure of this setting permits an explicit solution for optimal equity shares and the associated equilibrium actions, which allows us to explore these subtleties more precisely. Assume the relevant adjacency matrix $\network$ is invertible, which holds for generic weighted networks. At an optimal solution, the payments and actions of active agent $i$ satisfy $$\tau_{i}^* \propto \left(\subnetwork^{-1}\boldsymbol{1}\right)_i,$$ where $\subnetwork$ is the subnetwork of active agents for that payment allocation. This follows immediately from the balance conditions in  \Cref{t:optsharescharacterization}. Thus, the sum of row $i$ in $\widetilde{G}^{-1}$ is an index characterizing $i$'s compensation and equilibrium action. Thanks to this explicit description, we can see that these outcomes behave quite differently from standard measures of centrality such as Katz--Bonacich centrality. In particular, unlike the case with such standard measures, increasing the strength of an agent's links in $\bm{G}$ need not increase the compensation index, as we illustrate in the next section.  

\subsubsection{Comparative statics}
\label{sec:compstatics}

In this section, we explore how the optimal contract, as well as the agents' and principal's payoffs, vary with the technology of production. The simple form of the team performance function $Y$ in our environment, as well as the explicit characterization of incentives and outcomes, facilitate this exercise. The results demonstrate some interesting tensions between the principal's and the agents' interests. 

We look at how the principal's and agents' payoffs vary as the network changes. 

\begin{proposition}
\label{p:profitsnetworkperturb}
The principal's payoff is weakly increasing in the edge weight $G_{ij}=G_{ji}.$
\end{proposition}

The principal obtains weakly higher profits from an increase in edge weights. However, it need not be the case that agents prefer such a perturbation. We will illustrate this through a network on three agents (see \Cref{l:threeagentcomplete}).  %

Without loss of generality, we can assume $G_{12} \geq G_{13} \geq G_{23}$ and choose the normalization $G_{12}=1$, so that the adjacency matrix is 
$$\network = 
\begin{bmatrix} 0 & 1 & G_{13} \\
1 & 0 & G_{23} \\
G_{13} & G_{23} & 0
\end{bmatrix}.$$

\begin{figure}
    \centering
        \begin{tikzpicture} [node distance={30mm}, agt1/.style = {draw,circle,color=blue!60,fill=blue!20}, agt2/.style = {draw,circle,color=red!50,fill=red!20}, agt3/.style = {draw,circle,color=orange!100,fill=orange!30}]
\node[agt1] (1) {$1$};
\node[agt2] (2) [below left of=1] {$2$};
\node[agt3] (3) [below right of=1] {$3$};
\draw (2) edge["$G_{12}$"] (1);
\draw (3) edge["$G_{23}$"] (2);
\draw (1) edge["$G_{13}$"] (3);
\end{tikzpicture}
    \caption{Three agent weighted graph with weights $G_{12}, G_{13},$ and $G_{23}$.}
    \label{l:threeagentcomplete}
\end{figure}

\begin{figure}
    \centering
    \begin{subfigure}{0.7\textwidth}
    \centering
    \includegraphics[width=\textwidth]{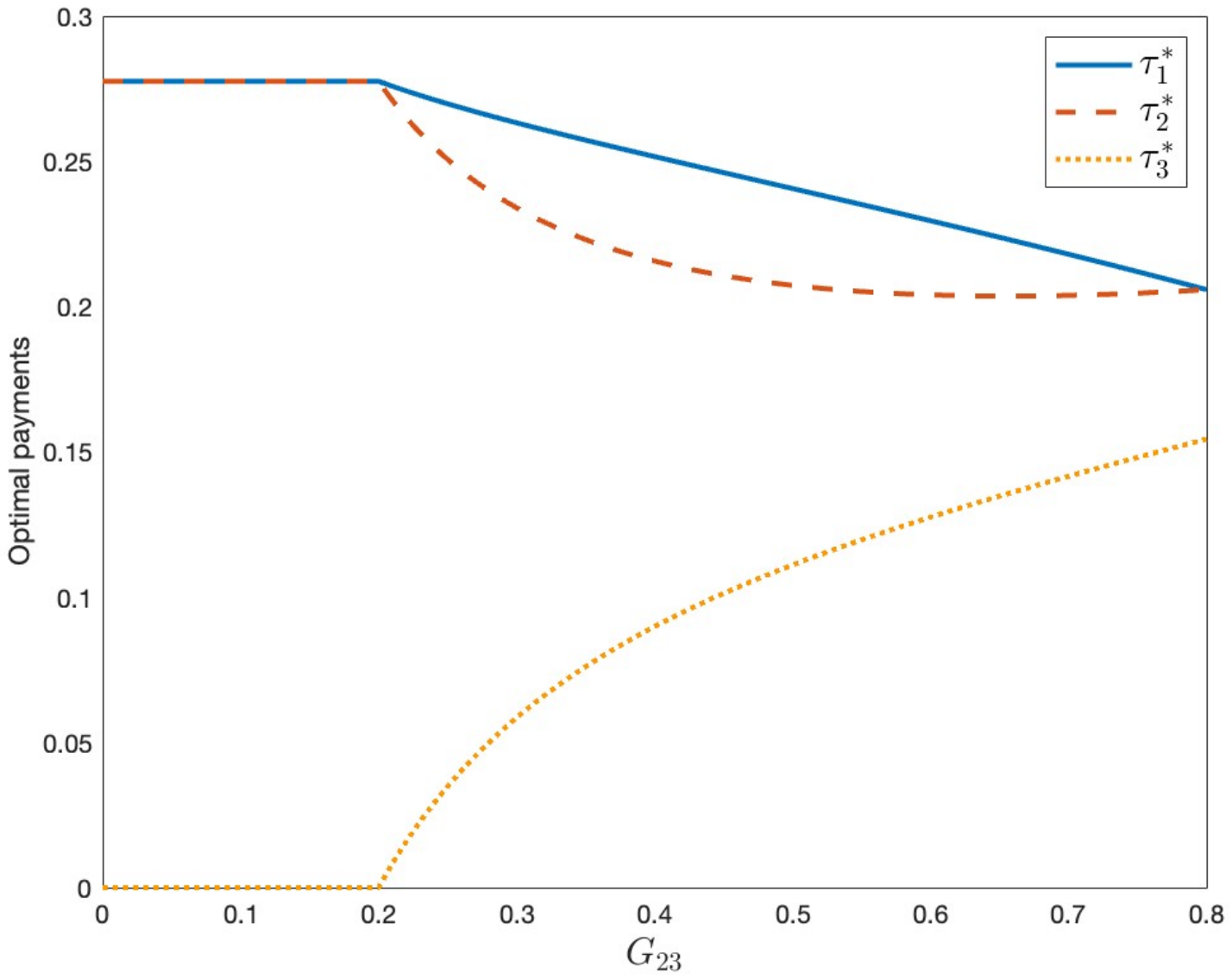}
    \caption{Optimal payments}
    \label{fig:optshares}
    \end{subfigure}
    \begin{subfigure}{0.7\textwidth}
        \centering
     \includegraphics[width=\textwidth]{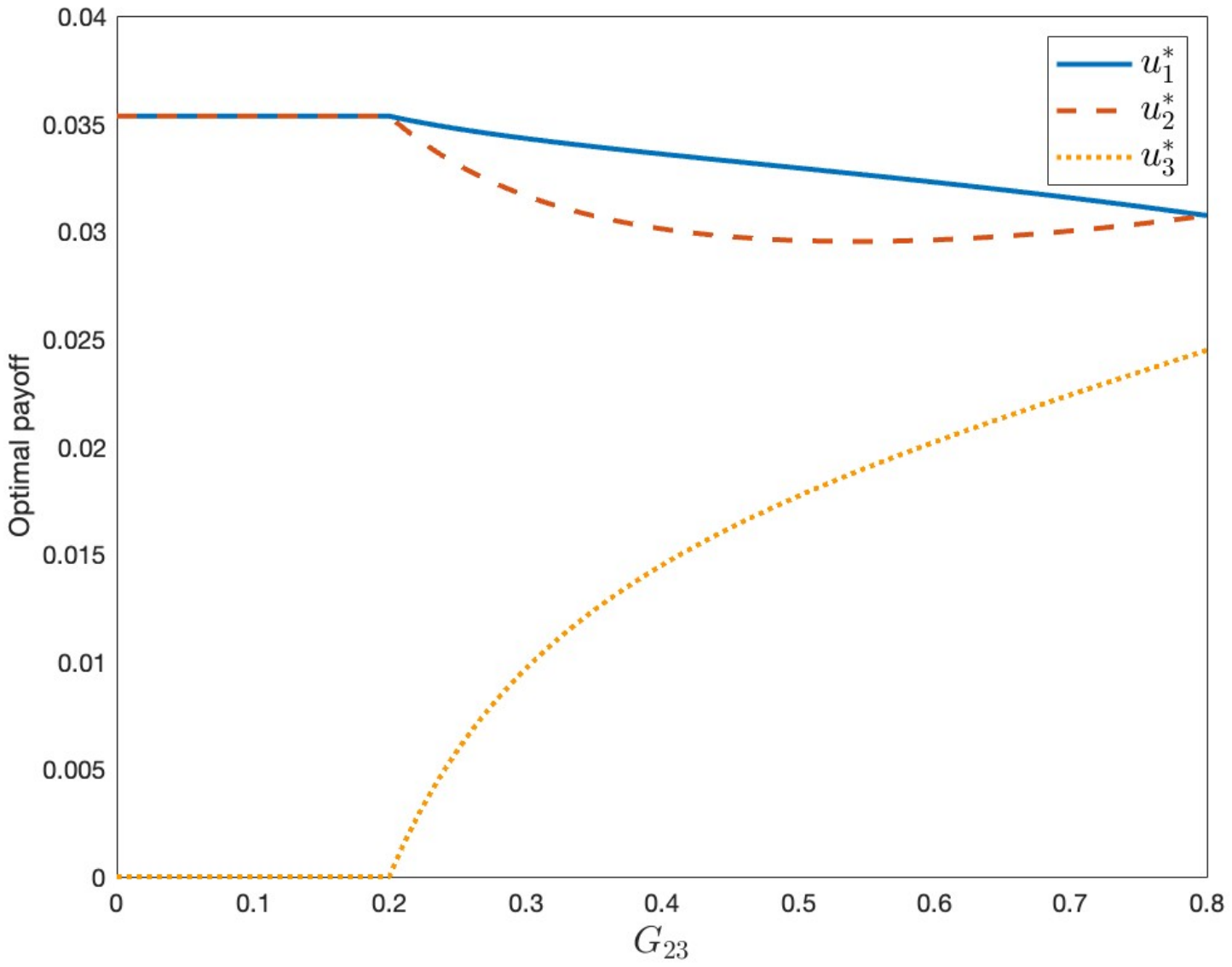}
     \caption{Payoffs under optimal payments}
     \label{fig:optpayoff}
    \end{subfigure}
    \caption{The optimal payments and resulting equilibrium payoffs as a function of the weight $G_{23}$. Here $G_{13}=0.8$ while $P(Y) = \min\{0.5Y,1\}$ (this could be replaced by a smooth, strictly concave $P$ without affecting the results.). In both diagrams, the curve corresponding to agent $1$ is the topmost (solid blue) one; the curve corresponding to agent $2$ is the second from the top (dashed red); and the curve corresponding to agent $3$ is the lowest (dotted orange) one.}
    \label{fig:optshares_payoffs}
\end{figure}

\Cref{fig:optshares_payoffs} shows the optimal payments and the corresponding equilibrium payoffs as we vary the link weight $G_{23}$, under parameter values specified in the caption. 
\Cref{fig:optshares} depicts optimal payments to each agent as a function of $G_{23}$. The payment is non-monotonic in own links: once payments are nonconstant in the strength of that link, increasing $G_{23}$ initially decreases agent $2$'s payment. The numerical example also illustrates a corresponding non-monotonicity in payoffs: strengthening one of an agent's links can decrease his equilibrium payoff under the optimal contract. \Cref{fig:optpayoff} depicts the equilibrium payoffs under optimal payments as a function of $G_{23}$. Strengthening the link between agents $2$ and $3$ can \emph{decrease} the resulting payoffs for agents $1$ and $2$. 

To give some intuition, observe that we are strengthening a link $G_{23}$ between two agents with a common neighbor $1$. When we do so, we may want to increase incentives to one of these neighbors. It turns out this agent is $3$, who is now connected enough to include in the team. But the balanced neighborhood equity condition requires that as we increase $\tau_3,$ we must keep the total equity $G_{12} \tau_2 + G_{13}\tau_3$ of $1$'s neighbors equal to the total equity in other neighborhoods. So to allocate more shares to agent $3$ while maintaining the balance condition, the principal must either take shares away from agent $2$ or give out substantially more shares to agents. The example demonstrates that the principal can prefer the former to the latter. More generally, we expect such a force to push toward non-monotonicities when we strengthen a link between two agents with one or more common neighbors (though there can be other countervailing forces).

This finding contrasts with an intuition that one might have from the network games literature, that agents are better off from becoming more central. Under fixed payments, all agents' payoffs are monotone in the network. In the present setting, however, agent $2$ can benefit from weakening one of his links. There is therefore a tension between the network formation incentives of the principal and the agents. Agents may not be willing to form links that would benefit the principal or the team as a whole, even if link formation has no technological cost.

\section{Production functions with broader complementarities}\label{s:CES}

The balance condition applies beyond network games. We briefly illustrate 
with two standard production functions where complementarities operate at the team level rather than through bilateral network links. We show that our framework yields 
simple characterizations of optimal contracts in these classical settings as well and that these characterizations have implications for the amount of pay dispersion in teams.

We derive optimal contracts for Cobb-Douglas and constant elasticity of substitution production functions. The solutions describe how much incentives should be concentrated on more productive agents, and our main finding is that the optimal contract exhibits more unequal compensation when agents are more substitutable.

We work again in the simple success-or-failure environment with risk-neutral agents.\footnote{The assumption that all agents have the same quadratic effort cost $C_i(a_i) = a_i^2/2$ is not essential: it is straightforward to extend the subsequent analysis to heterogeneous quadratic costs by rescaling actions.} For Cobb-Douglas production, team performance is $$Y(\bm{a}) = \prod_{i=1}^n a_i^{\gamma_i},$$
where the factor shares $\gamma_i$ can differ across agents. For constant elasticity of substitution (CES) production, team performance is $$Y(\bm{a}) = \left(\sum_{i=1}^n \gamma_i a_i^{\rho}\right)^{\kappa/\rho}$$
for a non-zero $\rho$. The parameter $\kappa$ captures returns to scale and the factor shares $\gamma_i>0$ can again differ across agents. Finally, the elasticity of substitution between agents' efforts is $\frac{1}{1-\rho}$.

We will characterize optimal contracts when the first-order approach is valid, and can establish sufficient conditions for this:

\begin{fact} \label{fact:cesstrictstable}
\Cref{as:balancederive} holds if $P(\cdot)$ is concave and:
\begin{enumerate}[(i)]
    \item For Cobb-Douglas, $\gamma_i \leq 1$ for all $i$ and $\sum_{i=1}^n \gamma_i < 2$,
    \item For CES, $\rho \leq 1$ and $\kappa \in (0,1]$.
\end{enumerate}
\end{fact}

One can directly characterize the optimal contract by applying \Cref{t:generalmodel} (see \cite{zuo2024optimizing}, which solves a closely related example in the Cobb-Douglas case). We take an alternate approach of transforming our problem to a simpler one and then applying \Cref{t:generalmodel}.

For Cobb-Douglas, consider an equivalent transformed problem in which we replace $Y$ with $\widetilde{Y} = \log(Y)$ and $P(Y)$ with $\widetilde{P}(\widetilde{Y}) = P(\exp(\widetilde{Y}))$. The problem is now separable:
$$\widetilde{Y}(\bm{a}) = \sum_{i=1}^n \gamma_i \log(a_i).$$
For CES, consider a transformed problem in which we replace $Y$ with $\widetilde{Y} = \frac{1}{\rho} Y^{\frac{\rho}{\kappa}}$ and $P(Y)$ with $\widetilde{P}(\widetilde{Y}) = P((\rho \widetilde{Y})^{\frac{\kappa}{\rho}})$. Analogous to Cobb-Douglas, the problem is separable:
\begin{equation}\widetilde{Y}(\bm{a}) = \frac{1}{\rho}\sum_{i=1}^n \gamma_i a_i^{\rho}. \label{eq:Y_transformed} \end{equation} These transformations do not change the optimal contract or the corresponding equilibrium actions.\footnote{For Cobb--Douglas production and CES production with $\rho<0$, we restrict attention to positive action profiles, the domain on which the transformations below are defined.} But because the transformed problems have additively separable team performances, we can now easily calculate agents' productivities and centralities.

This lets us obtain the following characterizations of optimal contracts:
\begin{proposition} \label{p:equityces} In the setting of this section, suppose \Cref{as:balancederive} holds.
    \begin{enumerate}
        \item 
For Cobb-Douglas production, the optimal payments $\tau^*_i$ are proportional to $\gamma_i$.
\item For CES production, the non-zero optimal payments\footnote{If $\rho=1$, then only agents with maximal $\gamma_i$ receive positive transfers. See \Cref{fact:cesrhoone}.}  $\tau^*_i$ are proportional to $\gamma_i^{\frac{1}{1-\rho}}$.   \end{enumerate}
\end{proposition}

The Cobb-Douglas solution is intuitive: the principal pays more to agents with higher factor shares. Turning to CES, we see that pay dispersion under the optimal contract depends on the elasticity of substitution $\frac{1}{1-\rho}$. When this elasticity is high, so agents' efforts are highly substitutable, the optimal contract pays agents with higher factor shares much more. Conversely when the elasticity of substitution is low, so agents' efforts are highly complementary, there is little pay dispersion.

Several limiting cases of the CES result are illustrative. In the limit  as $\rho \rightarrow -\infty$ with $\sum_i \gamma_i=1$, the production function converges to the Leontief production function $Y(\bm{a})=(\min_i a_i)^\kappa$. In this case, it is optimal to pay all agents equally because inducing one agent to take a higher action than others does not improve team performance. In the limit as $\rho \rightarrow 0$, again assuming $\sum_i\gamma_i=1$, the production function approaches the Cobb--Douglas, where output is quite responsive to each agent's own effort, and contributes to transformed output in proportion to $\gamma_i$. Then compensation turns out to be linear in $\gamma_i$ as in part (a) of the proposition. Finally, note that for any $\rho$, the relative compensation of agents does not depend on the returns to scale $\kappa$, but the total amount paid to agents when the project succeeds (and thus the absolute payments to individual agents) does depend on $\kappa$.

The Cobb-Douglas and CES cases show that the balance condition provides explicit characterizations of optimal incentives for production functions that have been difficult to analyze in the network games literature. A key point is that because we assume outcomes in our model depend on a one-dimensional team performance (via functions $P_{s}(Y)$ that can be quite general), we can apply monotone transformations that simplify the relevant spillovers between agents.

\section{Concluding discussion}
\label{s:conclusion}

We have studied an incentive design problem for a team whose members contribute via non-contractable effort. We investigate how optimal contracts depend on the team's production function. Our main result is a necessary condition for contract optimality. We show that optimal contracts must satisfy a balance condition across agents receiving positive incentive pay. 

In applications, our analysis offers a new perspective on compensation design that comes from focusing on spillovers. The overall theme is deriving interpretable insights into the roles of standalone capabilities and collaboration. We close with a few remarks on main features of this model and potentially important variations on them.

\subsection{Participation constraints}

The setting we have presented has no formal outside options; instead of participation constraints, it uses the constraint of limited liability to make the contracting problem nontrivial. 

Nevertheless, one could take the same model of production to a setting with standard participation constraints, where some  participation constraints would bind at the principal's optimum. The analogue of  \Cref{as:balancederive} would have to be made, not for all perturbations, but in a constrained space of contract perturbations where these constraints continue to be respected (as in \Cref{ss:indifferences}). Moreover, the balance condition would need to be augmented with agent-specific Lagrange multipliers on the binding constraints. If one agent works harder, this relaxes or tightens the constraints for other agents, and the balance condition would involve the multipliers on the constraints to account for these spillovers.  We leave a full treatment of this case to future work.

\subsection{Robustness to measurement error}

Our analysis has assumed that the principal knows the network of incentive
spillovers. In practice this network would be estimated with error. Suppose a principal takes measured spillover parameters at face value, checks the balance condition at a given contract, and perturbs incentive pay when it fails.
\Cref{a:centralityestimate} takes up the question of when this strategy can work.

It turns out that certain spectral statistics of the network play a subtle role in the answer. When
spillovers are moderate---the spectral norm of the effective spillover matrix
is bounded away from its critical value---the balance condition is robust to measurement error as long as this error is not too large. By
contrast, when spillovers are strong and the team is nearly partitioned into
weakly connected subteams, even arbitrarily small measurement error can mislead the principal about which subteam
is more responsive to incentives. We can rule out this issue when the network's spectral gap, a standard measure of segregation
\citep{vonLuxburg2007,GolubJackson2012}, is bounded away from zero. This extension connects our framework to the econometrics of
networks and to recent work on interventions with noisy spillover estimates
\citep{GaleottiGolubGoyalTalamasTamuz2024}.

\subsection{Other objectives} While we have focused on a standard principal concerned with maximizing profits, the analysis of how contract perturbations affect team performance is equally relevant for understanding other natural principal objectives. For example, the objective could be maximizing a (weighted) sum of workers' utilities. The key insight is that, under suitable conditions, the first-order conditions of such an optimization problems involve only the local comparative statics of the equilibrium, and the incentive spillover forces we discuss would be equally relevant. 

\subsection{The role of a one-dimensional outcome} An important simplifying assumption throughout the analysis is that actions influence outcomes only via a one-dimensional team performance $Y(\bm{a})$. Moving beyond this assumption to settings where probabilities of outcomes depend in an arbitrary way on the full action profile is a natural direction. Though a first-order approach may continue to be fruitful, establishing conditions for its validity is likely challenging. Subject to that, we expect that generalizations of our balance conditions would hold. However, the outcome distribution may now provide more fine-grained information about an agent's effort, and the principal will use this information to design optimal incentives (as in, e.g., \cite{holmstrom1982moral} and \cite{legros1993efficient}). The generalized balance conditions must incorporate this along with agents' centralities and productivities. So the balance conditions would be suitably adjusted, and rankings such as those in \Cref{s:generalmarginalutils} would also be affected by these informational considerations.

{\footnotesize \linespread{1.0}
\bibliographystyle{ecta}
\bibliography{references.bib}
}
\appendix

\newpage

\section{Omitted proofs}
\label{a:detailedproofs}

This section proves all results from  \Cref{sec:generalmodel}, \Cref{s:generalintensivemargin}, and \Cref{s:applications}. Proofs from \Cref{s:CES} are in the Online Appendix.

\subsection{Proof of \Cref{fact:opt_exist}} For contracts in this compact space, agent best responses will also be contained in a compact space. We want to show the supremum $\pi$ of attainable principal contracts is attainable. Choose a sequence of contracts $\bm{\tau}^{(i)}$ giving payoffs converging to $\pi$ and let $(\bm{a^*})^{(i)}$ be the corresponding equilibria. (If $\pi$ is zero, then the contract giving zero payments under all outcomes is optimal. So we can assume that $\pi$ is positive and thus an equilibrium exists under $\bm{\tau}^{(i)}$ for $i$ sufficiently large.) By compactness, we can choose a subsequence of $\bm{\tau}^{(i)}$ such that $\bm{\tau}^{(i)} \rightarrow \bm{\tau}^*$ and $(\bm{a}^*)^{(i)} \rightarrow \bm{a}^*$. By upper-hemicontinuity of equilibrium, the action profile $\bm{a}^*$ is an equilibrium under $\bm{\tau}^*$. Since the principal's payoffs are continuous in the contract and actions, the contract $\bm{\tau}^*$ attains the optimal payoff $\pi$.

\subsection{Proof of \autoref{l:changeteamperformance}}
We begin by observing that under any optimal contract, \begin{align*}
    a^*_i(\bm{\tau}^*) =0 \iff \tau^*_{i}(s) = 0, \quad \text{for all } s \in \mathcal{S}.
\end{align*} That is, at an optimal contract, an agent $i$ exerts zero effort at equilibrium if and only if it does not receive a payment from the contract at any outcome.

To see why the observation holds, suppose at an optimal contract $\optsharesvector$, there is an agent $i$ who receives positive payment $\tau_{i}^*(s) > 0$ under an outcome $s$ but chooses action $a_{i}^*(\optsharesvector) = 0$. Then the principal receives a strictly higher profit under the contract $\sharesvector^\dag$ which sets $\tau_{i}^\dag(s) = 0$ for all outcomes $s$ and is otherwise equal to $\optsharesvector$. At this contract, agent $i$ chooses action $a_{i}=0$ for any profile of actions $\bm{a}_{-i}$ played by other agents. Thus, the equilibrium $\eqlbactions(\optsharesvector)$ under contract $\optsharesvector$ is also an equilibrium profile under contract $\sharesvector^\dag$. Since outcome $s$ occurs with positive probability under any team performance, the principal's expected payments to agents are strictly higher under $\optsharesvector$ than $\sharesvector^\dag$. The other direction is straightforward.

With the observation established, we analyze the change in team performance as the transfers to agents are perturbed. Consider contract $\bm{\tau}$ and any agent $i$ for which there exists an outcome $s'$ such that $\tau_{i}(s') > 0$. For any outcome $s$, consider marginally increasing $\tau_{i}(s)$. The change induced by this perturbation is \begin{equation} \label{eq:changeteamperformance}
    \frac{\partial Y}{\partial \tau_{i}(s)} =\nabla Y(\eqlbactions)^{T} \cdot \frac{\partial \eqlbactions}{\partial \tau_{i}(s)},
\end{equation} where $\eqlbactions$ is the equilibrium action profile for the contract $\sharesvector$. The substance of the proof is analyzing the second term on the right-hand side of (\ref{eq:changeteamperformance}).  

First, consider any agent $j$ such that $a_{j}^*(\sharesvector)=0$. The change in their equilibrium action due to an increase in $\tau_{i}(s)$ is zero. The utility to agent $j$ at contract $\sharesvector$ is \begin{equation*}
    \mathcal{U}_{j} = \sum_{s' \in \mathcal{S}}P_{s'}(Y^*)u_{j}(\tau_{j}(s')) - C_{j}(a_{j}).
\end{equation*} At contract $\sharesvector$, agent $j$ receives no payment under any outcome, so he has a unique best response of $a_{j}^*=0$.

It is thus without loss to analyze the change in equilibrium actions of agents $j$ that take a strictly positive action in profile $\eqlbactions$, that is, $a_{j}^* > 0$. The analysis from here on focuses on such agents, overloading notation to represent the actions of these agents by $\eqlbactions$.

We will show that the change in equilibrium actions $\eqlbactions$ as the transfer $\tau_{i}(s)$ increases is \begin{equation}
    \label{eq:eqlbactionsderivative}
    \frac{\partial \bm{a}^*}{\partial \tau_{i}(s)} = \left[\bm{I}-\bm{H}^{-1} \bm{U}\bm{G} \right]^{-1} \bm{H}^{-1} \begin{bmatrix}
                 \bm{0}  \\
                 \frac{\partial Y}{\partial a_{i}}P'_{s}(Y)u_{i}'(\tau_{i}(s)) \\
                 \bm{0}
            \end{bmatrix} + \frac{\partial Y}{\partial \tau_{i}(s)}\left[\bm{H}-\bm{U}\bm{G}\right]^{-1}\bm{d}, 
\end{equation} for some vector $\bm{d}$.

Consider the equilibrium action profile $\eqlbactions$. For an agent $j$, the first-order conditions imply $a_{j}^*$ must solve the equation \begin{equation} \label{eq:bestresponsegeneral}
    C_{j}'(a_{j}) = \left(\sum_{s' \in \mathcal{S}}P_{s'}'(Y)u_{j}(\tau_{j}(s'))\right) \frac{\partial Y}{\partial a_{j}}.
\end{equation} 

To arrive at (\ref{eq:eqlbactionsderivative}), let us implicitly differentiate (\ref{eq:bestresponsegeneral}) with respect to $\tau_{i}(s)$. The resulting expression depends on the identity of agent $j$ in comparison to $i$, the agent whose payment is perturbed. For all $j \neq i$, \begin{multline} \label{eq:implicitgeneraljneqi}
    C_{j}''(a_{j}^*)\frac{\partial a^*_{j}}{\partial \tau_{i}(s)} = \left(\sum_{s' \in \mathcal{S}}P_{s'}'(Y)u_{j}(\tau_{j}(s'))\right) \left(\sum_{k=1}^{n}\frac{\partial^2 Y}{\partial a_{k} \partial a_{j}} \cdot \frac{\partial a_{k}^*}{\partial \tau_{i}(s)}\right) + \\ \frac{\partial Y}{\partial a_{j}} \cdot \frac{\partial Y}{\partial \tau_{i}(s)} \cdot \sum_{s' \in \mathcal{S}}P_{s'}''(Y)u_{j}(\tau_{j}(s')).
\end{multline} On the other hand, for $j=i$,  \begin{multline} \label{eq:implicitgeneraljeqi}
    C_{j}''(a_{j}^*)\frac{\partial a^*_{j}}{\partial \tau_{i}(s)} = \left(\sum_{s' \in \mathcal{S}}P_{s'}'(Y)u_{j}(\tau_{j}(s'))\right) \left(\sum_{k=1}^{n}\frac{\partial^2 Y}{\partial a_{k} \partial a_{j}} \cdot \frac{\partial a_{k}^*}{\partial \tau_{i}(s)}\right) \\ + \frac{\partial Y}{\partial a_{j}}P'_{s}(Y)u_{j}'(\tau_{j}(s)) +  \frac{\partial Y}{\partial a_{j}} \cdot \frac{\partial Y}{\partial \tau_{i}(s)} \sum_{s' \in \mathcal{S}}P_{s'}''(Y)u_{j}(\tau_{j}(s')). \end{multline}  We can combine (\ref{eq:implicitgeneraljneqi}) and (\ref{eq:implicitgeneraljeqi}) to write the resulting expression in vector form below \begin{align*}
    \frac{\partial \bm{a}^*}{\partial \tau_{i}(s)} &= \left[\bm{H}-\bm{U}\bm{G}\right]^{-1} \begin{bmatrix}
                 \bm{0}  \\
                 \frac{\partial Y}{\partial a_{i}}P'_{s}(Y)u_{i}'(\tau_{i}(s)) \\
                 \bm{0}
            \end{bmatrix} + \frac{\partial Y}{\partial \tau_{i}(s)}\left[\bm{H}-\bm{U}\bm{G}\right]^{-1}\bm{d},
\end{align*} where $\bm{d}$ is a vector with $j^{\text{th}}$ element defined as \begin{align*}
    d_{j} := \frac{\partial Y}{\partial a_{j}} \cdot \sum_{s' \in \mathcal{S}}P_{s'}''(Y)u_{j}(\tau_{j}(s')).
\end{align*} The expression in (\ref{eq:eqlbactionsderivative}) follows.

Substituting (\ref{eq:eqlbactionsderivative}) into (\ref{eq:changeteamperformance}), the change in team performance as the transfer $\tau_{i}(s)$ increases is\begin{multline*}
    \frac{\partial Y}{\partial \tau_{i}(s)} = \nabla Y(\eqlbactions)^{T}  \left[\bm{I}-\bm{H}^{-1} \bm{U} \network \right]^{-1} \bm{H}^{-1}\begin{bmatrix}
                 \bm{0}  \\
                 \frac{\partial Y}{\partial a_{i}}P'_{s}(Y)u_{i}'(\tau_{i}(s)) \\
                 \bm{0}
            \end{bmatrix} + \\ \frac{\partial Y}{\partial \tau_{i}(s)}\nabla Y(\eqlbactions)^{T}\left[\bm{H}-\bm{U}\bm{G}\right]^{-1}\bm{d}.
\end{multline*} Applying the definitions of $\productivity_{i}$ and $\centrality_{i}$, we obtain \begin{align} \label{eq:implicitteamperformance}
    \frac{\partial Y}{\partial \tau_{i}(s)} = \productivity_{i}\centrality_{i}P_{s}'(Y)\cdot \frac{u_{i}'(\tau_{i}(s))}{C''_i(a_i^*)} + \frac{\partial Y}{\partial \tau_{i}(s)}\nabla Y(\eqlbactions)^{T}\left[\bm{H}-\bm{U}\bm{G}\right]^{-1}\bm{d}.
\end{align}

To complete the proof, we have two cases depending on whether $\nabla Y(\bm{a}^*)^{\top}\left[\bm{H}-\bm{U}\bm{G}\right]^{-1}\bm{d}$ is equal to $1$. First suppose this term does not equal $1$. Rearranging,
$$    \frac{\partial Y}{\partial \tau_{i}(s)} = \frac{1}{1-\nabla Y(\eqlbactions)^{T}\left[\bm{H}-\bm{U}\bm{G}\right]^{-1}\bm{d}} \cdot \productivity_{i}\centrality_{i}P_{s}'(Y)\cdot \frac{u_{i}'(\tau_{i}(s))}{C''_i(a_i^*)}  
.$$
Setting $l=\frac{1}{1-\nabla Y(\eqlbactions)^{T}\left[\bm{H}-\bm{U}\bm{G}\right]^{-1}\bm{d}}$ and observing $l$ does not depend on $i$, we obtain the desired result.

In the remaining case $\nabla Y(\bm{a}^*)^{\top}\left[\bm{H}-\bm{U}\bm{G}\right]^{-1}\bm{d}=1$, (\ref{eq:implicitteamperformance}) simplifies to $\alpha_{i}c_{i}P_{s}'(Y)\cdot \frac{u_{i}'(\tau_{i}(s))}{C''_i(a_i^*)} =0$ for any agent $i$ and outcome $s$ such that $\tau_{i}(s)>0$. We will show in \Cref{l:positivepaymentoutcomes} that any agent that receives a payment at an optimal contract does so at an outcome where $P_{s}'(Y^*)>0$. Thus $c_{i}=0$ for such agents and the balance condition in \Cref{t:generalmodel} is satisfied trivially. 

\subsection{Proof of \autoref{t:generalmodel}}
The expected payoff for the principal under contract $\sharesvector$ and corresponding equilibrium actions $\eqlbactions$ is \begin{align*}
\sum_{s' \in \mathcal{S}}\left(v_{s'} - \sum_{i \in N}\tau_{i}(s')\right)P_{s'}(Y(\eqlbactions)).
    \end{align*} Suppose $\optsharesvector$ is an optimal contract inducing equilibrium $\eqlbactions(\optsharesvector)$ with team performance $Y^*$. Consider outcome $s$ and any agent $i$ such that $\tau_{i}^*(s)>0$. 
    Then the first-order condition for $\tau^*_i(s)$ implies that \begin{align*}
        \frac{d Y}{d \tau_{i}(s)} \underbrace{\sum_{s' \in \mathcal{S}}\left(v_{s'} - \sum_{i \in N}\tau^*_{i}(s')\right)P'_{s'}(Y^*)}_{D} = P_{s}(Y^*).
    \end{align*}
The left-hand side is the benefit from increasing $\tau^*_i(s)$ while the right-hand side is the expected additional transfer required. Since $P_s(Y^*)>0$ by assumption, the summation labeled $D$ is nonzero.
    
Substituting \Cref{l:changeteamperformance} in the above equation, we obtain   
\begin{align*} l \productivity_{i}\centrality_{i}P'_{s}(Y^*)\cdot \frac{u_{i}'(\tau_{i}(s))}{C''_i(a_i^*)}  &= \frac{P_{s}(Y^*)}{D}, \\ 
\iff \productivity_{i}\centrality_{i}\cdot \frac{u_{i}'(\tau_{i}(s))}{C''_i(a_i^*)}  &= \balanceconstant_{s},
\end{align*} where $\balanceconstant_{s} = P_{s}(Y^*)/(lP_{s}'(Y^*) D)$. Observing that $\balanceconstant_{s}$ is independent of $i$, the statement of the result follows.

\subsection{Proof of \autoref{c:marginalutil1}}
    Fix any $s\in\mathcal S_{ij}^*$. By \Cref{t:generalmodel}, there is a constant $\balanceconstant_s \neq 0$ such that \begin{align*}
        \productivity_{i} \centrality_{i}\cdot \frac{u_{i}'(\tau_{i}^*(s))}{C''_i(a_i^*)}  &= \balanceconstant_{s}, \quad \text{and } \quad  \productivity_{j} \centrality_{j}\cdot \frac{u_{j}'(\tau_{j}^*(s))}{C''_j(a_j^*)}  = \balanceconstant_{s}.
    \end{align*} It follows that \begin{equation*}
        \frac{u_{i}'(\tau_{i}^*(s))}{u_{j}'(\tau_{j}^*(s))} = \frac{\productivity_{j}\centrality_{j}}{\productivity_{i}\centrality_{i}} \cdot \frac{C''_i(a_i^*)}{C''_j(a_j^*)}.
    \end{equation*} This is exactly the claimed equality.

\subsection{Proof of \autoref{p:rankingpayments}}

We prove a couple of lemmas which help in proving the proposition statement. The first lemma gives a condition which must hold for all  outcomes at which an agent receives a positive payment.

\begin{lemma}
\label{l:positivepaymentoutcomes}
    Suppose $\optsharesvector$ is an optimal contract and $Y^*$ is the induced team performance. For all $s$ in the set of outcomes $\mathcal{S}_{i}^*$ where $i$ receives a positive payment,
$        P_{s}'(Y^*) > 0.$
\end{lemma}

\begin{proof}
    Consider agent $i$ and let $\mathcal{S}^*_{i}$ be the set of outcomes at which agent $i$ receives a positive payment. If $\mathcal{S}^*_{i}$ is the empty set, the result holds vacuously. Otherwise, we will show that, either \begin{align} \label{eq:characpaymentoutcome} P'_{s}(Y^*) > 0, \text{ for all } s \in \mathcal{S}^*_{i}, \quad \text{or}, \quad P'_{s}(Y^*) < 0, \text{ for all } s \in \mathcal{S}^*_{i}.
    \end{align} 
    Recall from the proof of \Cref{t:generalmodel} that
    \begin{align*}
        l \productivity_{i} \centrality_{i} P_{s}'(Y^*) \cdot \frac{u_{i}'(\tau_{i}(s))}{C''_i(a_i^*)}  = \frac{P_{s}(Y^*)}{\sum_{s' \in \mathcal{S}}\left(v_{s'} - \sum_{i \in N}\tau_{i}(s')\right)P'_{s'}(Y^*)}, \quad \forall s \in \mathcal{S}^*_{i}.
    \end{align*} Since every outcome occurs with non-zero probability, it must be that $P_{s}'(Y^*) \neq 0$ for any outcome in $\mathcal{S}_{i}^*$. In the case $\mathcal{S}_{i}^*$ has exactly $1$ outcome, it follows that $P_{s}'(Y^*) > 0$ or $P_{s}'(Y^*) < 0$ for $s \in \mathcal{S}_{i}^*$. Thus, suppose $\mathcal{S}_{i}^*$ has at least $2$ outcomes. Taking the ratio of the above equation for any pair of outcomes $s_{1},s_{2} \in \mathcal{S}^*_{i}$, we obtain \begin{align*}
        \frac{u_{i}'(\tau_{i}^*(s_{1}))}{u_{i}'(\tau_{i}^*(s_{2}))} = \frac{P_{s_{1}}(Y^*)}{P'_{s_{1}}(Y^*)} \cdot \frac{P'_{s_{2}}(Y^*)}{P_{s_{2}}(Y^*)}.
    \end{align*} Since the utility function $u_{i}(\cdot)$ is strictly increasing, we must have either \begin{align*}
        P'_{s}(Y^*) > 0 \text{ for } s \in \{s_{1},s_{2}\}, \quad \text{or}, \quad P'_{s}(Y^*) < 0 \text{ for } s \in \{s_{1},s_{2}\}.
    \end{align*} The statement in (\ref{eq:characpaymentoutcome}) follows. We now show that $
        P_{s}'(Y^*) > 0$  for all $s \in \mathcal{S}_{i}^*.$
 The equilibrium condition for agent $i$ is \begin{align*}
    C_{i}'(a_{i}^*) &= \frac{\partial Y}{\partial a_{i}}\sum_{s \in \mathcal{S}}P_{s}'(Y^*)u_{i}(\tau^*_{i}(s)). 
\end{align*} Since $\sum_{s \in \mathcal{S}}P_{s}'(Y) = 0$, the equilibrium condition can be rewritten as \begin{align*}
    C_{i}'(a_{i}^*) &= \frac{\partial Y}{\partial a_{i}}\sum_{s \in \mathcal{S}_{i}^*}P_{s}'(Y^*)\left(u_{i}(\tau^*_{i}(s)) - u_{i}(0)\right).
\end{align*} By assumption we have a positive productivity, that is, $\frac{\partial Y}{\partial a_{i}} > 0$. The cost of effort is strictly increasing, that is, $C_{i}'(\cdot) > 0$. The utility function $u_{i}(\cdot)$ is strictly increasing in payments. Thus,  \begin{align*}
     P'_{s}(Y^*) > 0, \text{ for all } s \in \mathcal{S}^*_{i}
\end{align*} as desired.
\end{proof}

The second lemma shows the existence of a common outcome at which agents receiving a positive payment are paid. 

\begin{lemma}
\label{l:commonoutcomepay}
    Suppose $\optsharesvector$ is an optimal contract. Consider a pair of agents $i$ and $j$, each with strictly concave utility functions. If there exist outcomes $s_{i}$ and $s_{j}$ such that $\tau_{i}^*(s_{i}) > 0$ and $\tau_{j}^*(s_{j}) > 0$, then there exists an outcome $s \in \mathcal{S}$ such that \begin{align*}
    \tau_{i}^*(s) > 0 \text{ and } \tau_{j}^*(s) > 0.
\end{align*} 
\end{lemma}

\begin{proof}
    Suppose there does not exist an outcome at which both agents receive a positive payment. Thus, the payments $\tau_{i}^*(s_{j})=0$ and $\tau_{j}^*(s_{i})=0$. The KKT first-order conditions at optimal contract $\optsharesvector$ are \begin{align} \label{eq:focequal}
    lD\productivity_{k}\centrality_{k}\cdot \frac{u_{k}'(\tau_{k}^*(s_{k}))}{C''_k(a_k^*)}\cdot P_{s_{k}}'(Y^*) - P_{s_{k}}(Y^*) = 0 \quad \text{for } k \in \{i,j\}.
\end{align} In addition to the above set of equations, we also have \begin{align} \label{eq:focinequal}
    lD\productivity_{k}\centrality_{k}\cdot \frac{u_{k}'(0)}{C''_k(a_k^*)}\cdot P'_{s_{\{i,j\} \setminus k}}(Y^*) - P_{s_{\{i,j\} \setminus k}}(Y^*) \leq 0 \quad \text{for } k \in \{i,j\}.
\end{align} Recall $\mathcal{S}_{i}^*$ is the set of outcomes where agent $i$ receives a positive payment under contract $\optsharesvector$. 
 
Since $P'_{s}(Y^*) > 0$ for any $s \in \mathcal{S}_{i}^*$ (see \Cref{l:positivepaymentoutcomes}), we must have $lD \productivity_{i} \centrality_{i} > 0$. Consider the following chain of inequalities for agent $i$: \begin{align} \label{eq:agenticond}
        \frac{P_{s_{i}}(Y^*)}{P_{s_{i}}'(Y^*)} < lD\productivity_{i}\centrality_{i}\cdot \frac{u_{i}'(0)}{C''_i(a_i^*)} 
        \leq \frac{P_{s_{j}}(Y^*)}{P'_{s_{j}}(Y^*)}.
    \end{align} Both the inequalities follow from applying (\ref{eq:focequal}) and (\ref{eq:focinequal}) to agent $i$. We utilize the fact that $u_{i}(\cdot)$ is strictly concave. We also utilize the observation that, since agent $i$ and $j$ receive a positive payment at outcome $s_{i}$ and $s_{j}$, \Cref{l:positivepaymentoutcomes} tells us that $P_{s_{i}}'(Y^*)>0$ and $P_{s_{j}}'(Y^*)>0$. Following the same computation for agent $j$, we obtain the inequalities \begin{align} \label{eq:agentjcond}
        \frac{P_{s_{j}}(Y^*)}{P_{s_{j}}'(Y^*)} < lD\productivity_{j}\centrality_{j}\cdot \frac{u_{j}'(0)}{C''_j(a_j^*)} 
        \leq \frac{P_{s_{i}}(Y^*)}{P'_{s_{i}}(Y^*)}.
    \end{align} This contradicts inequality (\ref{eq:agenticond}). Thus, if two agents receive a positive payment at some (potentially different) outcomes under the optimal contract, then there must exist an outcome at which both agents receive a positive payment.
\end{proof}

\begin{proof}[Proof of \Cref{p:rankingpayments}]
    Consider agents $i$ and $j$ with identical strictly concave utility functions $u_{i}(\cdot) = u_{j}(\cdot)$. The statement trivially holds if either agent $i$ or agent $j$ receives a $0$ payment at all outcomes. Thus, consider a scenario where there exist outcomes $s_{i}$ and $s_{j}$ such that \begin{align*}
    \tau_{i}^*(s_{i}) > 0 \text{ and } \tau_{j}^*(s_{j}) > 0.
\end{align*}

By \Cref{l:commonoutcomepay}, it suffices to show that when there exists an outcome such that both agents $i$ and $j$ receive a positive payment at this outcome, then \begin{align*}
    \tau^*_{i}(s) \geq \tau^*_{j}(s) \text{ for all } s \in \mathcal{S} \text{ or }
        \tau^*_{j}(s) \geq \tau^*_{i}(s) \text{ for all } s \in \mathcal{S} 
\end{align*} (or both).

Let $\mathcal{S}^*_{ij}$ be the set of outcomes at which both agents receive a positive payment under contract $\optsharesvector$. The set $\mathcal{S}^*_{ij}$ is non-empty. We can assume without loss of generality that $\tau_{i}^*(s) \geq \tau_{j}^*(s)$ for some outcome $s \in \mathcal{S}_{ij}^*$. We show that then\begin{align*}
        \tau^*_{i}(s) \geq \tau^*_{j}(s) \text{ for all } s \in \mathcal{S}.
    \end{align*} Applying \Cref{c:marginalutil1}, we must have \begin{align*}
        \left|\frac{\productivity_{i}\centrality_{i}}{C''_i(a_i^*)}\right| \geq \left|\frac{\productivity_{j}\centrality_{j}}{C''_j(a_j^*)}\right|. 
    \end{align*} Additionally, $\productivity_{i}\centrality_{i}$ and $\productivity_{j}\centrality_{j}$ are either both positive or negative. Further applying \Cref{c:marginalutil1} to any outcome $s' \in \mathcal{S}_{ij}^*$, the ratio of marginal utilities satisfies \begin{align*}
        \frac{u_{i}'(\tau^*_{i}(s'))}{u_{j}'(\tau^*_{j}(s'))} = \frac{\productivity_{j}\centrality_{j}C''_i(a_i^*)}{\productivity_{i}\centrality_{i}C''_j(a_j^*)} \leq 1.
    \end{align*} This implies that agent $i$ receives a weakly larger payment than agent $j$ under all outcomes in $\mathcal{S}_{ij}^*$, that is, \begin{align*}
        \tau_{i}^*(s) \geq \tau_{j}^*(s) \quad \textit{for all } s \in \mathcal{S}_{ij}^*.
    \end{align*} We will show that this ordering on payments holds for outcomes in the set $\mathcal{S} \setminus \mathcal{S}_{ij}^*$ as well. The ordering trivially holds at outcomes where $\tau_{j}^*(s) = 0$. Consider an outcome $s$ at which $\tau_{j}^*(s) > 0$ but $\tau_{i}^*(s)=0$. We show that such an outcome cannot exist at an optimal contract $\optsharesvector$. We showed in the proof of \Cref{t:generalmodel} that the first-order condition for the principal is \begin{align*}
        lD\productivity_{j}\centrality_{j}\cdot \frac{u_{j}'(\tau_{j}^*(s))}{C''_j(a_j^*)}\cdot P_{s}'(Y^*)-P_{s}(Y^*)=0.
    \end{align*} Since the utility to agent $j$ is strictly increasing and the cost of effort is strictly convex, this implies that \begin{align*}
        lD\productivity_{j}\centrality_{j}P_{s}'(Y^*) > 0.
    \end{align*} Recall that $ \left|\frac{\productivity_{i}\centrality_{i}}{C''_i(a_i^*)}\right| \geq \left|\frac{\productivity_{j}\centrality_{j}}{C''_j(a_j^*)}\right|$ and they are either both positive or negative. Using the fact that $u_{i}(\cdot)$ and $u_{j}(\cdot)$ are strictly increasing identical utility functions, and thus $u_{i}'(0) = u_{j}'(0) > 0$, and $P_{s}'(Y^*) > 0$ (see \Cref{l:positivepaymentoutcomes}), we conclude \begin{align} \label{eq:rankingpteameffect}
    lD\productivity_{i}\centrality_{i} \cdot \frac{u_{i}'(0)}{C''_i(a_i^*)} \cdot P_{s}'(Y^*) \geq lD\productivity_{j}\centrality_{j}\cdot \frac{u_{j}'(0)}{C''_j(a_j^*)} \cdot P_{s}'(Y^*). 
    \end{align}
    
    Now, consider the following chain of inequalities: \begin{align*}
        lD\productivity_{i}\centrality_{i}\cdot \frac{u_{i}'(0)}{C''_i(a_i^*)}\cdot P_{s}'(Y^*)-P_{s}(Y^*) & \geq lD\productivity_{j}\centrality_{j}\cdot \frac{u_{j}'(0)}{C''_j(a_j^*)}\cdot P_{s}'(Y^*)-P_{s}(Y^*), \\
        & > lD\productivity_{j}\centrality_{j}\cdot \frac{u_{j}'(\tau_{j}^*(s))}{C''_j(a_j^*)} \cdot P_{s}'(Y^*)-P_{s}(Y^*), \\
        & = 0. 
    \end{align*} The first inequality follows from (\ref{eq:rankingpteameffect}). The second inequality follows from the fact that $u_{j}(\cdot)$ is strictly concave and $lD\frac{\productivity_{j}\centrality_{j}}{C''_j(a_j^*)}P_{s}'(Y^*) > 0$. The left-hand side in the above chain of inequalities is the derivative of the principal's objective in $\tau_{i}(s)$. The derivative being positive contradicts the optimality of $\optsharesvector$, so the statement of the proposition holds. 
\end{proof}

\subsection{Proof of \autoref{c:marginalutil2}}
    Recall from the proof of \Cref{t:generalmodel} that
    \begin{align*}
        l \productivity_{i} \centrality_{i} P_{s}'(Y^*) \cdot \frac{u_{i}'(\tau_{i}^*(s))}{C''_i(a_i^*)} = \frac{P_{s}(Y^*)}{\sum_{s' \in \mathcal{S}}\left(v_{s'} - \sum_{i \in N}\tau_{i}(s')\right)P'_{s'}(Y^*)}, \quad \forall s \in \mathcal{S}^*_{i}.
    \end{align*} Taking the ratio of the above equation for any pair of outcomes $s_{1},s_{2} \in \mathcal{S}^*_{i}$, we obtain \begin{align*}
        \frac{u_{i}'(\tau_{i}^*(s_{1}))}{u_{i}'(\tau_{i}^*(s_{2}))} = \frac{P_{s_{1}}(Y^*)}{P'_{s_{1}}(Y^*)} \cdot \frac{P'_{s_{2}}(Y^*)}{P_{s_{2}}(Y^*)}.
    \end{align*} The statement is proved.

\subsection{Proof of \autoref{l:optprodcomp}} 
We begin by deriving an expression for an agent's productivity $\alpha_{i}$. By definition, \begin{align} \label{eq:prodrankone}
    \alpha_{i} = \frac{\partial Y}{\partial a_{i}} = k_{i} + \underbrace{\rho \left(\sum_{j \in \mathcal{N}}\beta_{j}a_{j}\right)}_{\widetilde{B}_{1}} \beta_{i}. 
\end{align} Because actions satisfy $a_{j} \geq 0$ for every $j$, and $\beta_{j} \geq 0$ by assumption, the term $\widetilde{B}_{1} \geq 0$. 

We next derive an expression for an agent's centrality $c_{i}$. By definition, at any contract $\bm{\tau}$ and corresponding equilibrium $\bm{a}^*$, the centrality \begin{align} \label{eq:centrankone}
    \bm{c}^{\top} = \bm{\alpha}^{\top} \left[\bm{I}-\rho P'(Y^*)\bm{T} \bm{\beta} \bm{\beta}^{\top}\right]^{-1},
\end{align} where $\bm{T} = \text{diag}\left(u_{1}(\tau_{1}),\dots,u_{n}(\tau_n)\right)$. Using the Sherman-Morrison formula, the inverse matrix is \begin{align} \label{eq:rankoneinv}
    \left[\bm{I}-\rho P'(Y^*)\bm{T} \bm{\beta} \bm{\beta}^{\top}\right]^{-1} = \bm{I} + \frac{\rho P'(Y^*) \bm{T}\bm{\beta} \bm{\beta}^\top}{1-\rho P'(Y^*) \bm{\beta}^{\top} \bm{T} \bm{\beta}}. 
\end{align} As established in the proof of \autoref{l:uniqueeq}, the spectral radius of the spillover matrix at any contract and corresponding equilibrium is strictly smaller than $1$, and in the rank-one case this radius equals $\rho P'(Y^*)\bm\beta^\top\bm T\bm\beta$, so the denominator in \eqref{eq:rankoneinv} is strictly positive. Substituting (\ref{eq:rankoneinv}) in (\ref{eq:centrankone}), we get \begin{align*}
    \bm{c}^{\top} = \bm{\alpha}^{\top} + \underbrace{\frac{\rho P'(Y^*) \bm{\alpha}^{\top}\bm{T} \bm{\beta}}{1-\rho P'(Y^*) \bm{\beta}^{\top} \bm{T} \bm{\beta}}}_{\widetilde{B}_{2}} \bm{\beta}^{\top} ,
\end{align*} where $\widetilde{B}_{2} \geq 0$. Substituting (\ref{eq:prodrankone}) in the expression above, we get \begin{align} \label{eq:centralityrankone}
    c_{i} = k_{i} + \left(\widetilde{B}_{1}+\widetilde{B}_{2}\right) \beta_{i}. 
\end{align}

Finally, consider any optimal contract $\bm{\tau}^*$. By \Cref{t:generalmodel}, there exists a constant $\lambda$ such that for any agent receiving positive pay \begin{align*}
   \alpha_{i} c_{i} u_{i}'(\tau_{i}) = \lambda.
\end{align*} Substituting (\ref{eq:prodrankone}) and (\ref{eq:centralityrankone}) in the balance condition, it follows that there exists $B_{1} \geq 0$ and $B_{2} \geq 0$ such that for all agents $i$ ,
   \begin{align*}
       u_{i}'\left(\tau_{i}^*\right) \propto \left(k_{i}^2 + B_{1}k_{i}\beta_{i} + B_{2}\beta_{i}^2\right)^{-1}.
   \end{align*} The statement is proved. 

\subsection{Proof of \autoref{t:spilloverregimes}} 

This proof utilizes \Cref{l:optprodcomp}, which characterizes an optimal contract for any given environment. We will also require a characterization of equilibrium action profiles under a given contract. We provide the characterization for any linear-quadratic production function without requiring rank-one structure, as the more general version will be useful for later proofs.

\begin{lemma}\label{l:uniqueeq}
Consider the setting of \Cref{ex:main}. Fixing $\sharesvector$, there exists a unique pure Nash equilibrium. The equilibrium actions $\eqlbactions$ and  team performance $Y^*$ solve the equations
\begin{equation} [\bm{I}-P'(Y^*) \bm{\bm{T}} \network]\eqlbactions = P'(Y^*)\bm{T} \bm{k} \text{ and }Y^* = \production(\eqlbactions),\label{eq:Nash_conditions} \end{equation} where $\bm{T}$ is the diagonal matrix with entries ${T}_{ii} = u_{i}\left(\sharesagt{i}\right)$. 
\end{lemma}

\begin{proof}
Fix a contract $\sharesvector$. Before we provide a proof, we make some useful observations. Any agent $i$ that receives zero pay $\tau_i=0$ must take a zero action in any equilibrium. Now, the probability of success $P(Y)$ is strictly increasing in $Y$ and team performance $Y(\bm{a})$ is strictly increasing in each argument. Thus, any agent that receives strictly positive pay $\tau_i>0$ has a strictly positive marginal returns to effort $\frac{\partial U_i}{\partial a_{i}} > 0$ at a zero action $a_{i}=0$. To analyze equilibrium, it thus suffices to focus on agents that receive positive pay. For the rest of the argument, we overload notation and denote positive payments made to agents by $\sharesvector$. The proof proceeds in two steps.

\textit{Step (i): An action profile $\eqlbactions$ is a pure Nash equilibrium if and only if it is a solution to the following system of equations: \begin{equation}
    [\bm{I}-P'(Y^*)\sharesmatrix \network]\eqlbactions = P'(Y^*)\bm{Tk} \text{ and }Y^* = \production(\eqlbactions). \label{eq:kktfoc}
\end{equation}} 

\textit{Proof of Step (i):} Consider any pure Nash equilibrium $\bm{a}^*$. Because agents receive positive pay and the marginal returns at a zero action is strictly positive, we have $a_{i}^*>0$ for every $i$. The KKT first-order condition for actions gives (\ref{eq:kktfoc}).

Conversely, fix $\bm a_{-i}$. Agent $i$'s payoff derivative has the sign of
\[
u_i(\tau_i)P'(Y(a_i,\bm a_{-i}))
-\frac{a_i}{k_i+\sum_{j\ne i}G_{ij}a_j+G_{ii}a_i}.
\]
The first term is weakly decreasing in $a_i$ by concavity of $P$ and monotonicity of $Y$, while the fraction is strictly increasing in $a_i$. The displayed expression is therefore strictly decreasing in $a_i$. Thus, \eqref{eq:kktfoc} makes $a_i^*$ agent $i$'s unique global best response.

\textit{Step (ii): Fix a contract $\sharesvector$. There exists a unique solution to the following equations: \begin{equation*}
    [\bm{I}-P'(Y^*)\sharesmatrix \network]\eqlbactions = P'(Y^*)\bm{Tk} \text{ and }Y^* = \production(\eqlbactions).
\end{equation*}}

\textit{Proof of Step (ii):} For $y$ satisfying $P'(y)\rho(\sharesmatrix\network)<1$, define $$\eqlbactions(y)=[\bm I-P'(y)\sharesmatrix\network]^{-1}P'(y)\bm{Tk}.$$ 
Solutions of the first-order conditions then correspond to solutions to
$$Y(\eqlbactions(y)) =y.$$

The function $Y(\eqlbactions(y))$ is strictly increasing in each coordinate of $\eqlbactions(y)$. We analyze how $\eqlbactions(y)$ changes as $y$ increases. Consider the set \begin{align*}
    y_{R} := \{y: P'(y) \rho(\sharesmatrix \network) < 1\}.
\end{align*} Observe that because $P(\cdot)$ is concave, if $y \in y_{R}$ then $y + \epsilon \in y_{R}$ for any $\epsilon > 0$. We show that constrained to the set $y_{R}$, there exists a unique fixed point to the function $Y(\eqlbactions(y))$. Each coordinate of $\eqlbactions(y)$ is weakly decreasing in $y$ since $P'(\cdot)$ is weakly decreasing (by our assumption $P(\cdot)$ is concave). So  $Y(\eqlbactions(y))$ is decreasing, meaning there is at most one solution to $Y(\eqlbactions(y)) = y$. It remains to show a solution to this equation exists.

We claim that we can find $y$ such that $Y(\eqlbactions(y)) \geq y$ and $P'(y) \rho(\sharesmatrix \network)<1$. If $P'(0) \rho(\sharesmatrix \network)<1$, the claim holds with $y=0$ since $Y(\eqlbactions(0)) \geq 0$. Otherwise, define $y_0$ by $P'(y_0) \rho(\sharesmatrix \network)=1$. A solution to this equation exists since $P'(y)$ is continuous and converges to zero as $y\rightarrow \infty$. Then $Y(\eqlbactions(y)) \rightarrow \infty$ as $y \rightarrow y_0$ from above, so we have $Y(\eqlbactions(y_0+\epsilon)) \geq y_0+\epsilon$ for $\epsilon>0$ sufficiently small. This completes the proof of the claim.

Since $Y(\eqlbactions(y))$ is decreasing in $y$, we can also choose $y$ large enough such that $y > Y(\eqlbactions(y))$. Since $Y(\eqlbactions(y))$ is continuous in $y$, by the intermediate value theorem this function has a fixed point, denoted by $y^*$. We conclude that there exists a unique solution to $Y(\eqlbactions(y))  = y$ in the set $y_{R}$ and a corresponding profile $\eqlbactions$ of equilibrium actions.

It remains to show that there does not exist an equilibrium $\eqlbactions$ with corresponding team performance $Y^*$ such that $P'(Y^*) \rho(\sharesmatrix \network) \geq 1$. The case $\sharesvector=0$ is immediate as the only equilibrium is $\eqlbactions=0$. Take $\sharesvector$ not identically zero and suppose there exists an equilibrium $\eqlbactions$ such that $P'(Y^*) \rho(\sharesmatrix \network) \geq 1$. It must solve: \begin{equation} \label{eq:necsufeqlb}
    [\bm{I}-P'(Y^*)\sharesmatrix \network]\eqlbactions = P'(Y^*)\bm{Tk} \text{ and }Y^* = \production(\eqlbactions).
\end{equation} By the Perron-Frobenius theorem,\footnote{For this argument, it is without loss to assume the matrix $\sharesmatrix\network$ is irreducible. If not, since $\network$ is symmetric and $\tau_{i}>0$ for some agent $i$, we can rewrite $\sharesmatrix\network$ in a block diagonal form with irreducible blocks. Then $P'(Y^*)\rho(\sharesmatrix \network)$ must be an eigenvalue of at least one block of the matrix $P'(Y^*) \sharesmatrix \network$. We can drop agents in all other blocks and apply the remainder of the argument to this block.} there exists a left-eigenvector $\bm{v}$ of the matrix $P'(Y^*)  \sharesmatrix \network$ such that $\bm{v}$ has strictly positive entries. Multiplying the LHS of (\ref{eq:necsufeqlb}) by the vector $v$, we get \begin{align*}
        \bm{v}^{T}[\bm{I}-P'(Y^*)\sharesmatrix \network] \eqlbactions &= [1 - P'(Y^*) \rho(\sharesmatrix \network)]\bm{v}^{T} \eqlbactions \\
        & \leq 0,
\end{align*} where the inequality follows from the assumption $P'(Y^*)\rho(\sharesmatrix \network) \geq 1$ and the fact that $\eqlbactions$ has strictly positive elements. However, we also compute
\begin{align*}
        \bm{v}^{T}[\bm{I}-P'(Y^*)\sharesmatrix \network] \eqlbactions &= \bm{v}^{T}P'(Y^*)\bm{Tk}  \text{ by }(\ref{eq:necsufeqlb})\\
        & > 0,
\end{align*} 
where the inequality holds because the entries of $\bm{v}$ are all positive and the entries of $\bm{\tau}$ are all non-negative and not identically zero. This is a contradiction, so there does not exist an equilibrium $\eqlbactions$ with corresponding team performance $Y^*$ such that $P'(Y^*) \rho(\sharesmatrix \network) \geq 1$. We conclude the equilibrium described above is the unique one.
\end{proof}

We can now prove the theorem. Consider any optimal contract $\bm{\tau}$ and the corresponding equilibrium $\bm{a}^*$, which (by the previous lemma) is given by \begin{align} \label{eq:eqlbactrankone}
    \bm{a}^* = P'(Y^*) \left[\bm{I}-\rho P'(Y^*)\bm{T} \bm{\beta} \bm{\beta}^{\top}\right]^{-1} \bm{T} \bm{k},
\end{align} where $\bm{T} = \text{diag}\left(u_{1}(\tau_{1}),\dots,u_{n}(\tau_{n})\right)$. Under this characterization, it can be established that the constants in \Cref{l:optprodcomp} have the following expression: \begin{align} \label{eq:b1}
    B_{1} &= \rho P'(Y^*) \bm{k}^{\top} \bm{T} \bm{\beta} \cdot \frac{3-2\rho P'(Y^*)\bm{\beta}^{\top}\bm{T}\bm{\beta}}{\left(1-\rho P'(Y^*)\bm{\beta}^{\top}\bm{T}\bm{\beta}\right)^2}, \\ \label{eq:b2} \text{ and }  B_{2} &= \left(\rho P'(Y^*) \bm{k}^{\top} \bm{T} \bm{\beta}\right)^{2} \cdot \frac{2-\rho P'(Y^*)\bm{\beta}^{\top} \bm{T} \bm{\beta}}{\left(1-\rho P'(Y^*)\bm{\beta}^{\top}\bm{T}\bm{\beta}\right)^3}. 
\end{align} Given the above expressions, we establish Part (i) and Part (ii). 

Before turning to each part of the argument, we characterize the spectral radius of the spillover matrix. Consider the sequence of environments in the theorem, and define
\begin{align*}
    \bm T^{(m)}=\operatorname{diag}\!\left(u_1(\tau_1^{(m)}),\ldots,u_n(\tau_n^{(m)})\right).
\end{align*}
Denote the induced team performance by $Y^{*(m)}$. The spillover matrix is
\begin{align*}
    \bm S^{(m)}=\rho^{(m)}P'(Y^{*(m)})\bm T^{(m)}\bm\beta\bm\beta^{\top},
\end{align*}
which is rank one. Its unique non-zero eigenvalue is
\begin{align} \label{eq:specrad}
    \mu^{(m)}=\rho^{(m)}P'(Y^{*(m)})\bm\beta^{\top}\bm T^{(m)}\bm\beta
\end{align}
with associated eigenvector $\bm T^{(m)}\bm\beta$. Let $B_1^{(m)}$ and $B_2^{(m)}$ denote the expressions in \eqref{eq:b1} and \eqref{eq:b2}, respectively, evaluated in environment $m$.

To establish Part (i), we want to prove that as $m\to\infty$ under $\mu^{(m)}\to0$, the ratio of agent payments converges as follows: \begin{align*}
    \lim_{m \to \infty} \frac{u_{i}'\left(\tau_{i}^{(m)}\right)}{u_{j}'\left(\tau_{j}^{(m)}\right)} = \frac{k_{j}^2}{k_{i}^2}.
\end{align*} By \Cref{l:optprodcomp}, the ratio of payments for any pair of agents $i$ and $j$ is \begin{align} \label{eq:ratutilrankone}
    \frac{u_{i}'\left(\tau_{i}^{(m)}\right)}{u_{j}'\left(\tau_{j}^{(m)}\right)} = \frac{k_{j}^2 + B_{1}^{(m)}k_{j}\beta_{j} + B_{2}^{(m)}\beta_{j}^2}{k_{i}^2 + B_{1}^{(m)}k_{i}\beta_{i} + B_{2}^{(m)}\beta_{i}^2}
\end{align} We will show that $B_{1}^{(m)}$ and $B_{2}^{(m)}$, given by \eqref{eq:b1} and \eqref{eq:b2} evaluated in environment $m$, converge to $0$ as $m\to\infty$ under $\mu^{(m)}\to0$. By assumption and \eqref{eq:specrad}, we have \begin{align*}
    \lim_{m \to \infty} \rho^{(m)} P'(Y^{*(m)}) \bm{\beta}^{\top} \bm{T}^{(m)} \bm{\beta} = 0.
\end{align*} It remains to analyze the convergence of $\rho^{(m)} P'(Y^{*(m)}) \bm{k}^{\top} \bm{T}^{(m)} \bm{\beta}$. Consider the following chain of inequalities: \begin{align*}
    \rho^{(m)} P'(Y^{*(m)}) \bm{\beta}^{\top} \bm{T}^{(m)} \bm{\beta} \geq \underline{\beta} \rho^{(m)} P'(Y^{*(m)}) \bm{1}^{\top} \bm{T}^{(m)} \bm{\beta}
     \geq ({\underline{\beta}}/{\overline{k}}) \rho^{(m)} P'(Y^{*(m)}) \bm{k}^{\top} \bm{T}^{(m)} \bm{\beta}.
\end{align*} Since $\underline{\beta} / \overline{k} > 0$ and the left-hand side converges to $0$, the term $\rho^{(m)} P'(Y^{*(m)}) \bm{k}^{\top} \bm{T}^{(m)} \bm{\beta}$ converges to $0$. Substituting in \eqref{eq:b1} and \eqref{eq:b2}, it follows that $B_{1}^{(m)}$ and $B_{2}^{(m)}$ both converge to $0$. Thus, \begin{align*}
    \lim_{m \to \infty} \frac{u_{i}'\left(\tau_{i}^{(m)}\right)}{u_{j}'\left(\tau_{j}^{(m)}\right)} = \frac{k_{j}^2}{k_{i}^2}.
\end{align*}

To establish Part (ii), we want to prove that as $m\to\infty$ under $\mu^{(m)}\to1$, the ratio of agent payments converges as follows: \begin{align*}
    \lim_{m \to \infty} \frac{u_{i}'\left(\tau_{i}^{(m)}\right)}{u_{j}'\left(\tau_{j}^{(m)}\right)} = \frac{\beta_{j}^2}{\beta_{i}^2}.
\end{align*} Rewriting the ratio of payments using (\ref{eq:b1}) and (\ref{eq:b2}), we have that \begin{align*}
    \frac{u_{i}'\left(\tau_{i}^{(m)}\right)}{u_{j}'\left(\tau_{j}^{(m)}\right)}
    =\frac{
    \begin{aligned}
    &\left(1-\mu^{(m)}\right)^3k_{j}^2 \\
    &\quad + \rho^{(m)} P'(Y^{*(m)}) \bm{k}^{\top} \bm{T}^{(m)} \bm{\beta} (1-\mu^{(m)})(3-2\mu^{(m)})k_{j}\beta_{j} \\
    &\quad + \left(\rho^{(m)} P'(Y^{*(m)}) \bm{k}^{\top} \bm{T}^{(m)} \bm{\beta}\right)^2(2-\mu^{(m)})\beta_{j}^2
    \end{aligned}
    }{
    \begin{aligned}
    &\left(1-\mu^{(m)}\right)^3k_{i}^2 \\
    &\quad + \rho^{(m)} P'(Y^{*(m)}) \bm{k}^{\top} \bm{T}^{(m)} \bm{\beta} (1-\mu^{(m)})(3-2\mu^{(m)})k_{i}\beta_{i} \\
    &\quad + \left(\rho^{(m)} P'(Y^{*(m)}) \bm{k}^{\top} \bm{T}^{(m)} \bm{\beta}\right)^2(2-\mu^{(m)})\beta_{i}^2
    \end{aligned}
    }
\end{align*} It suffices to establish that $\rho^{(m)} P'(Y^{*(m)}) \bm{k}^{\top} \bm{T}^{(m)} \bm{\beta}$ does not converge to $0$ and remains bounded as $m\to\infty$ under $\mu^{(m)}\to1$. Consider the following chain of inequalities: \begin{align*}
    \rho^{(m)} P'(Y^{*(m)}) \bm{k}^{\top} \bm{T}^{(m)} \bm{\beta} \geq \underline{k} \rho^{(m)} P'(Y^{*(m)}) \bm{1}^{\top} \bm{T}^{(m)} \bm{\beta} \geq \frac{\underline{k}}{\overline{\beta}} \rho^{(m)} P'(Y^{*(m)}) \bm{\beta}^{\top} \bm{T}^{(m)} \bm{\beta}.
\end{align*} Since $\underline{k} / \overline{\beta} > 0$ and $\mu^{(m)}$ converges to $1$, it must be that $\rho^{(m)} P'(Y^{*(m)}) \bm{k}^{\top} \bm{T}^{(m)} \bm{\beta}$ does not converge to $0$. Next, \begin{align*}
     \rho^{(m)} P'(Y^{*(m)}) \bm{k}^{\top} \bm{T}^{(m)} \bm{\beta} \leq \overline{k} \rho^{(m)} P'(Y^{*(m)}) \bm{1}^{\top} \bm{T}^{(m)} \bm{\beta} \leq \frac{\overline{k}}{\underline{\beta}} \rho^{(m)} P'(Y^{*(m)}) \bm{\beta}^{\top} \bm{T}^{(m)} \bm{\beta}.
\end{align*} These inequalities establish that $\rho^{(m)} P'(Y^{*(m)}) \bm{k}^{\top} \bm{T}^{(m)} \bm{\beta}$ remains bounded. Thus, \begin{align*}
    \lim_{m \to \infty} \frac{u_{i}'\left(\tau_{i}^{(m)}\right)}{u_{j}'\left(\tau_{j}^{(m)}\right)} = \frac{\beta_{j}^2}{\beta_{i}^2}.
\end{align*} The statement is proved.

\subsection{Proof of \Cref{t:optsharescharacterization}} As established in \Cref{s:applyparametric}, under an optimal contract $\bm{\tau}^*$, the product $\alpha_{i} c_{i}$ is equal across all active agents. 

\begin{lemma}
\label{l:equalizebonacichendo}
    If $\productivity_{i}\centrality_{i}$ is constant across all active agents, then $\productivity_{i}$  is constant across all active agents.
\end{lemma}

\begin{proof}
    By \Cref{t:generalmodel}, we have that $$\productivity_i \centrality_i \text{ is constant across agents $i$} .$$ We derive an expression for productivity $\bm{\alpha}$ and centrality $\bm{c}$ that will be used throughout the proof. The productivity $\bm{\alpha}$ is given by \begin{align*}
        \bm{\alpha} &= \bm{1} + \bm{G}\bm{a}^*, \\
        &= \bm{1} + P'(Y^*)\bm{G}\left[\bm{I}-P'(Y^*)\bm{TG}\right]^{-1}\bm{\tau} \quad (\text{by } \autoref{l:uniqueeq}), \\
        &= \left[\bm{I}-P'(Y^*)\bm{GT}\right]^{-1}\bm{1},
    \end{align*} where $\bm{1}$ denotes the vector whose entries are all equal to one. Furthermore, the centrality $\bm{c}$ is defined as \begin{align*}
        \bm{c} = \left[\bm{I}-P'(Y^*)\bm{GT}\right]^{-1} \bm{\alpha}. 
    \end{align*}
    
    Suppose that there exist two agents $i^{*} \in N$ with $i^{*} = \argmin_{k \in N}\productivity_{k}$ and $j^{*} \in N$ with $j^{*} = \argmax_{k \in N} \productivity_{k}$ such that $\productivity_{i^{*}} < \productivity_{j^{*}}$.\footnote{We are grateful to Michael Ostrovsky for suggesting the argument in the next paragraph.} 

Then we have that, for agent $i^{*}$, \begin{equation} \productivity_{i^{*}}\centrality_{i^*} < \productivity_{i^{*}}\productivity_{j^{*}}\sum_{j \in N}\left[\bm{I} - P'(Y^*) \network \sharesmatrix \right]^{-1}_{i^*j} = (\productivity_{i^{*}})^{2}\productivity_{j^{*}}, \label{eq:ineq1}\end{equation} using the maximality of $\productivity_{j^*}$ among the $\productivity_j$ and  the definitions of $\centrality_{i^*}$ and $\productivity_{i^*}$. But we similarly have that, for agent $j^{*}$, \begin{equation}\productivity_{j^{*}}\centrality_{j^*} > \productivity_{j^{*}}\productivity_{i^{*}}\sum_{i \in N}\left[\bm{I}-P'(Y^*) \network \sharesmatrix \right]^{-1}_{j^*i} = \productivity_{i^{*}}(\productivity_{j^{*}})^{2}. \label{eq:ineq2} \end{equation} \Cref{t:generalmodel} implies that $\productivity_{i^{*}}\centrality_{i^*}=\productivity_{j^{*}}\centrality_{j^*}$ for any two agents $i^*$ and $j^*$, and so combining (\ref{eq:ineq1}) and (\ref{eq:ineq2}) implies
 $$(\productivity_{i^*})^2\productivity_{j^*} > \productivity_{i^*}(\productivity_{j^*})^2.$$
This contradicts our assumption $\productivity_{j^*}>\productivity_{i^*}$, so we must have $\productivity_i$ equal to some constant $\balanceconstant_1$ for all $i$ in $N$.    
\end{proof}

Part (b) follows because $\alpha_i=1+(\bm G\bm a^*)_i$, so Part (a) implies that $(\bm G\bm a^*)_i$ is constant across active agents.
For Part (c), the equilibrium first-order condition gives $a_i^*=P'(Y^*)\tau_i^*\alpha_i$; since $P'(Y^*)$ and $\alpha_i$ are common across active agents, $\bm\tau^*$ is a common scalar multiple of $\bm a^*$, and Part (b) implies that $(\bm G\bm\tau^*)_i$ is constant across active agents.

\subsection{Proof of \Cref{p:profitsnetworkperturb}} 
Suppose $\optsharesvector$ is an optimal contract for network $\network$ with equilibrium team performance $Y^*(\network,\optsharesvector)$. Consider a perturbed network $\widetilde{\bm{G}}$ generated by increasing edge weight $G_{ij}$ to $G_{ij} + \epsilon$ for some $\epsilon > 0$. We will show that contract $\optsharesvector$ performs weakly better on network $\widetilde{\bm{G}}$ than on $\network$. Since we will be comparing $\optsharesvector$ across networks, we suppress the dependence of equilibrium team performance on the contract. 

Consider contract $\optsharesvector$ and network $\widetilde{\bm{G}}$. We want to show that the equilibrium team performance $Y^*(\widetilde{\network})$ is at least $Y^*(\network)$. The equilibrium actions solve \begin{align*}
    \eqlbactions(\widetilde{\bm{G}}) = P'(Y^*(\widetilde{\bm{G}}))\left[\bm{I} - P'(Y^*(\widetilde{\bm{G}}))\bm{T}^* \widetilde{\network}\right]^{-1} \optsharesvector.
\end{align*} Suppose $Y^*(\widetilde{\bm{G}}) < Y^*(\network)$. It follows that $\eqlbactions(\widetilde{\bm G})$ is pointwise weakly greater than $\eqlbactions(\bm G)$ by the Neumann series and the nonnegativity of the spillover matrix. However, this is a contradiction to $Y^*(\widetilde{\bm{G}}) < Y^*(\network)$ because \begin{align*}
    Y(\bm{a},\bm{G}) = \sum_{i}a_{i} + \frac{1}{2}\sum_{i,j}G_{ij}a_{i}a_{j}.
\end{align*} Thus, we must have $Y^*(\widetilde{\bm{G}}) \geq Y^*(\network)$. Moreover, $1-\sum_i\tau_i^*\geq0$ because the optimal contract weakly dominates the zero contract. The profits to the principal under contract $\optsharesvector$ are thus weakly higher on network $\widetilde{\bm{G}}$ than on network $\network$:\begin{align*}
    \left(1-\sum_{i}\tau_{i}^*\right)P\left(Y^*(\widetilde{\bm{G}})\right) \geq \left(1-\sum_{i}\tau_{i}^*\right)P\left(Y^*\left(\network\right)\right).
\end{align*} Finally, the optimal contract for network $\widetilde{\bm{G}}$ must deliver at least as high a payoff as contract $\optsharesvector$ does on network $\widetilde{\bm{G}}.$

\newpage

\clearpage

\pagenumbering{arabic}
\setcounter{page}{1}

\setcounter{section}{0}
\setcounter{subsection}{0}
\setcounter{subsubsection}{0}
\setcounter{equation}{0}
\setcounter{figure}{0}
\setcounter{table}{0}

\setcounter{theorem}{0}
\setcounter{lemma}{0}
\setcounter{proposition}{0}
\setcounter{corollary}{0}
\setcounter{assumption}{0}
\setcounter{conjecture}{0}
\setcounter{example}{0}
\setcounter{remark}{0}
\setcounter{fact}{0}

\renewcommand{\thesection}{OA.\Alph{section}}
\renewcommand{\thesubsection}{\thesection.\arabic{subsection}}
\renewcommand{\thesubsubsection}{\thesubsection.\arabic{subsubsection}}

\renewcommand{\theHsection}{OA.\Alph{section}}
\renewcommand{\theHsubsection}{OA.\Alph{section}.\arabic{subsection}}
\renewcommand{\theHsubsubsection}
  {OA.\Alph{section}.\arabic{subsection}.\arabic{subsubsection}}

\renewcommand{\theequation}{OA.\arabic{equation}}

\renewcommand{\thetheorem}{OA.\arabic{theorem}}
\renewcommand{\thelemma}{OA.\arabic{lemma}}
\renewcommand{\theproposition}{OA.\arabic{proposition}}
\renewcommand{\thecorollary}{OA.\arabic{corollary}}
\renewcommand{\theassumption}{OA.\arabic{assumption}}
\renewcommand{\thefact}{OA.\arabic{fact}}

\renewcommand{\theHequation}{onlineappendix.\Alph{section}.\arabic{equation}}
\renewcommand{\theHfigure}{onlineappendix.\Alph{section}.\arabic{figure}}
\renewcommand{\theHtable}{onlineappendix.\Alph{section}.\arabic{table}}

\renewcommand{\theHtheorem}{onlineappendix.\Alph{section}.\arabic{theorem}}
\renewcommand{\theHlemma}{onlineappendix.\Alph{section}.\arabic{lemma}}
\renewcommand{\theHproposition}{onlineappendix.\Alph{section}.\arabic{proposition}}
\renewcommand{\theHcorollary}{onlineappendix.\Alph{section}.\arabic{corollary}}
\renewcommand{\theHassumption}{onlineappendix.\Alph{section}.\arabic{assumption}}
\renewcommand{\theHconjecture}{onlineappendix.\Alph{section}.\arabic{conjecture}}
\renewcommand{\theHexample}{onlineappendix.\Alph{section}.\arabic{example}}
\renewcommand{\theHremark}{onlineappendix.\Alph{section}.\arabic{remark}}
\renewcommand{\theHfact}{onlineappendix.\Alph{section}.\arabic{fact}}

\title{Online Appendix:
Incentive Design With Spillovers}

\date{\today}

\makeatletter
\let\thankses\@empty
\makeatother
\maketitle

\setcounter{section}{1}

\section{Omitted Proofs}

\subsection{Proof of \Cref{fact:cesstrictstable}} We begin by providing sufficient conditions for strict best responses and stable equilibrium actions under Cobb-Douglas production, before turning to the CES specification. 

\textit{Proof of Part (i)}. Consider any contract $\bm{\tau}$. For Cobb-Douglas production, team performance is \begin{align*}
    Y(\bm{a}) = \prod_{i=1}^{n}a_{i}^{\gamma_{i}}. 
\end{align*} We begin by establishing strictness of best responses. The utility of an agent is \begin{align*}
    \mathcal{U}_{i} = \tau_{i}P(Y) - \frac{a_{i}^2}{2}. 
\end{align*} The second derivative of $\mathcal{U}_{i}$ with respect to $a_{i}$ is \begin{align} \label{eq:sdutil}
    \frac{\partial^2 \mathcal{U}_{i}}{\partial a_{i}^2} = \tau_{i} P''(Y) \left(\frac{\partial Y}{\partial a_{i}}\right)^2 + \tau_{i}P'(Y)\frac{\partial^2 Y}{\partial a_{i}^2} - 1. 
\end{align} The second derivative of $Y(\cdot)$ with respect to $a_{i}$ is $\gamma_{i}\left(\gamma_{i}-1\right)Y/a_{i}^2$. Thus, (\ref{eq:sdutil}) is strictly negative when $P(\cdot)$ is concave and $\gamma_{i} \leq 1$. 

We next establish the stability of equilibrium actions at the optimal contract. Observe that the action profile $a_{i}^*=0$ for every agent $i \in \mathcal{N}$ is a Nash equilibrium for any contract $\bm{\tau}$. Consequently, for a given contract, the principal's preferred equilibrium is the one --- if it exists --- in which all agents exert positive effort. For any contract such that $\tau_{i}=0$ for some agent $i \in \mathcal{N}$, the only equilibrium is one where every agent takes a zero action. For the remainder of the proof, we focus on contracts where every agent receives positive pay. 

We show that when the factor shares also satisfy \begin{align*}
    \sum_{i \in \mathcal{N}}\gamma_{i} < 2,
\end{align*} a unique equilibrium with positive actions exists and is  differentiable with respect to the  contract $\bm{\tau}$. The result relies on the following characterization of equilibrium team performance: given $\bm{\tau}$, an equilibrium with positive actions exists if and only if there is a solution to \begin{align} \label{eq:tpcharaccd}
    Y^{\frac{2-\sum_{i}\gamma_{i}}{\sum_{i}\gamma_{i}}} \cdot \frac{1}{P'(Y)} = \prod_{i}\left(\tau_{i}\gamma_{i}\right)^{\frac{\gamma_{i}}{\sum_{i}\gamma_{i}}}.
\end{align} To see why, suppose a positive equilibrium $\bm{a}^*$ exists and let $Y^*$ denote the corresponding team performance. Agents' actions must satisfy the equilibrium first-order conditions \begin{align*}
    a_{i}^* = \sqrt{\tau_{i}\gamma_{i}P'(Y^*)Y^*}\quad \text{for every } i \in \mathcal{N}. 
\end{align*} Raising both sides to the power $\gamma_{i}$, multiplying across all agents and taking the power of that equation to $\left(2/\sum_{i}\gamma_{i}\right)$ yields (\ref{eq:tpcharaccd}). For the converse direction, suppose there exists a solution to (\ref{eq:tpcharaccd}) which we denote by $Y^*$. The action profile defined by \begin{align*}
    a_{i}^* = \sqrt{\tau_{i}\gamma_{i}P'(Y^*)Y^*} \quad \text{for every } i \in \mathcal{N},
\end{align*} then constitutes a Nash equilibrium. This follows because when $\gamma_{i} \leq 1$, each agent's utility is strictly concave in own action, implying that any profile satisfying the first-order conditions is indeed a Nash equilibrium.  

We now establish that when $\sum_{i}\gamma_{i}<2$, (\ref{eq:tpcharaccd}) admits a unique positive solution for any given contract $\bm{\tau}$. By assumption, $Y^{\frac{2-\sum_{i}\gamma_{i}}{\sum_{i}\gamma_{i}}}$ is strictly increasing and $P'(\cdot)$ is weakly decreasing by concavity of $P(\cdot)$. Thus the left-hand side of (\ref{eq:tpcharaccd}) is strictly increasing which implies a unique positive solution exists. Moreover, because $P(\cdot)$ is smooth, the left-hand side of (\ref{eq:tpcharaccd}) is smooth, implying that the solution $Y^*(\bm{\tau})$ is differentiable with respect to the contract $\bm{\tau}$. Combining this with the uniqueness of the positive equilibrium action profile, it follows that the principal's preferred equilibrium is differentiable at every contract $\bm{\tau}$. 

\textit{Proof of Part (ii).} A similar analysis applies to CES production. Consider any contract $\bm{\tau}$. The team performance is \begin{align*}
    Y(\bm{a}) = \left(\sum_{i=1}^{n}\gamma_{i}a_{i}^{\rho}\right)^{\kappa / \rho}.
\end{align*} We begin by establishing strictness of best responses. The second derivative of $Y(\cdot)$ with respect to $a_{i}$ is \begin{align*}
    \frac{\partial^2 Y}{\partial a_{i}^2} = \gamma_{i} \kappa a_{i}^{\rho-2} \left(\sum_{\ell=1}^{n}\gamma_{\ell}a_{\ell}^{\rho}\right)^{\frac{\kappa}{\rho}-2} \left[\gamma_{i}\left(\kappa-1\right)a_{i}^{\rho} + (\rho-1)\sum_{j \neq i}\gamma_{j}a_{j}^{\rho}\right].
\end{align*} Thus, the second derivative of $\mathcal{U}_{i}$ with respect to $a_{i}$, as defined in (\ref{eq:sdutil}), is strictly negative when $P(\cdot)$ is concave, $\kappa \in (0,1]$ and $\rho \leq 1$. Under these conditions, each agent's utility is strictly concave in own action, implying that the best responses are unique. 

We next establish the stability of equilibrium actions at the optimal contract. We show that when $\kappa<2$, a unique Nash equilibrium exists and is differentiable with respect to the contract $\bm{\tau}$. The result relies on the following characterization: given $\bm{\tau}$, a positive equilibrium exists if and only if there is a solution to \begin{align} \label{eq:tpces}
    Y^{\frac{2-\kappa}{\kappa}} \cdot \frac{1}{P'(Y)} = \kappa \left(\sum_{i=1}^{n}\gamma_{i}\left(\tau_{i}\gamma_{i}\right)^{\frac{\rho}{2-\rho}}\right)^{\frac{2-\rho}{\rho}}. 
\end{align} To see why, suppose an equilibrium $\bm{a}^*$ exists and let $Y^*$ denote the corresponding team performance. Agents' actions must satisfy the equilibrium first-order conditions \begin{align*}
    a_{i}^* = \kappa \tau_{i} \gamma_{i} P'(Y^*) \left(Y^*\right)^\frac{\kappa-\rho}{\kappa} \left(a_{i}^*\right)^{\rho-1}.
\end{align*} Taking equilibrium action $a_{i}^*$ to the left-hand side and taking the power of that equation to $\rho/(2-\rho)$ yields \begin{align*}
    \left(a_{i}^*\right)^{\rho} = \left(\kappa \tau_{i} \gamma_{i} P'(Y^*)\right)^{\frac{\rho}{2-\rho}} \left(Y^*\right)^{\frac{\kappa-\rho}{\kappa} \cdot \frac{\rho}{2-\rho}} \quad \text{ for every } i \in \mathcal{N}.
\end{align*} Multiplying the equation by $\gamma_{i}$ and summing across all agents gives \begin{align*}
    \left(Y^*\right)^{\rho/\kappa} = \left(Y^*\right)^{\frac{\kappa-\rho}{\kappa} \cdot \frac{\rho}{2-\rho}} \left(\kappa P'(Y^*)\right)^{\frac{\rho}{2-\rho}}\sum_{i=1}^{n}\gamma_{i}\left(\tau_{i}\gamma_{i}\right)^{\frac{\rho}{2-\rho}}. 
\end{align*} Eq. (\ref{eq:tpces}) then follows by moving the terms involving $Y^*$ and $P'(Y^*)$ to the left-hand side and taking the power of that equation to $(2-\rho)/\rho$. For the converse direction, suppose there exists a solution to (\ref{eq:tpces}) which we denote by $Y^*$. The action profile defined by \begin{align*}
    a_{i}^* = \left(\kappa \tau_{i} \gamma_{i} P'(Y^*)\right)^{\frac{1}{2-\rho}} \left(Y^*\right)^{\frac{\kappa-\rho}{\kappa} \cdot \frac{1}{2-\rho}} \quad \text{ for every } i \in \mathcal{N},
\end{align*} then constitutes a Nash equilibrium. This follows because when $\rho\leq1$ and $\kappa \in (0,1]$, each agent's utility is strictly concave in own action, so any profile satisfying the first-order conditions constitutes a Nash equilibrium. 

By assumption that $\kappa<2$, the term $Y^{(2-\kappa)/\kappa}$ is strictly increasing in $Y$, and since $P'(\cdot)$ is weakly decreasing by concavity of $P(\cdot)$, the left-hand side of (\ref{eq:tpces}) is strictly increasing. It follows that equation (\ref{eq:tpces}) admits a unique positive solution for any given contract $\bm{\tau}$. 

Because $P(\cdot)$ is smooth, the left-hand side of (\ref{eq:tpces}) is smooth, implying that the solution $Y^*(\bm{\tau})$ is differentiable with respect to the contract $\bm{\tau}$. The statement of the result follows. 

\subsection{Proof of \Cref{p:equityces}} \textit{Proof of Part (a):} Recall that, after applying the monotone transformation $\widetilde{Y} = \log Y$ and $\widetilde{P}(\widetilde{Y}) = P(\text{exp}(\widetilde{Y}))$, team performance can be written as \begin{align*}
    \widetilde{Y}(\bm{a}) = \sum_{i \in \mathcal{N}}\gamma_{i}\log \left(a_{i}\right). 
\end{align*} Consider an optimal contract $\bm{\tau}^*$ and equilibrium $\bm{a}^*$. Each agent's first-order condition is \begin{align} \label{eq:cdfoc}
    a_{i}^* = \frac{\tau_{i}\gamma_{i}\widetilde{P}'(\widetilde{Y}^*)}{a_{i}^*}.
\end{align} We now compute the quantities in \Cref{t:generalmodel}. Differentiating $\widetilde{Y}$, we compute productivity to be $\widetilde{\alpha}_{i} = \gamma_{i} / a_{i}^*$. Since the transformed problem is separable, the matrix $\frac{\partial^2 \widetilde{Y}}{\partial a_i \partial a_j}$ of spillovers is diagonal with entries $\frac{\partial^2 \widetilde{Y}}{\partial a_i^2}   = -\frac{\gamma_i}{a_i^2}.$
We next compute agents' centralities at equilibrium to be\begin{align*}
    \widetilde{c}_{i} &= \widetilde{\alpha}_{i} \left(1 + \frac{\tau_i \widetilde{P}'(\widetilde{Y}^*)\gamma_i}{(a_i^*)^2} \right)^{-1}, \\
    &= \frac{\widetilde{\alpha}_{i}}{2} \quad \text{ by (\ref{eq:cdfoc})}.
\end{align*} We can now apply \Cref{t:generalmodel}, which states that under any optimal contract satisfying \Cref{as:balancederive}, the quantity $\frac{\widetilde{\alpha}^2_i}{2}$ is equal for all agents and thus all agents have the same productivity. This tells us that equilibrium actions are proportional to $\gamma_i$, each agent's factor share in the original production function. Given this, \Cref{eq:cdfoc} implies that each agent's payment $\tau_i$ under the optimal contract is proportional to his factor share $\gamma_i$.

\begin{fact}\label{fact:cesrhoone}
If $\rho=1$, only agents with maximal $\gamma_i$ receive positive payments at an optimal contract.
\end{fact}
\begin{proof}
When $\rho=1$, transformed production is $\widetilde Y(\bm a)=\sum_i\gamma_i a_i$, so $\widetilde\alpha_i=\widetilde c_i=\gamma_i$.
For any paid agent $i$, the principal's first-order condition gives $lD\widetilde P'(\widetilde Y^*)\gamma_i^2=\widetilde P(\widetilde Y^*)>0$, so $lD\widetilde P'(\widetilde Y^*)>0$. From this it follows that all agents who are paid must have identical $\gamma_i$, and only those with the maximal $\gamma_i$ can be paid.
\end{proof}

\textit{Proof of Part (b) for $\rho<1$:} It is again useful to transform team production: we consider an equivalent transformed problem in which we replace $Y$ with $\widetilde{Y} =  \frac{1}{\rho}Y^{\rho/\kappa}$ and $P(Y)$ with $\widetilde{P}(\widetilde{Y}) = P((\rho \widetilde{Y})^{\kappa/\rho})$. The transformed problem is:
\begin{equation}\widetilde{Y}(\bm{a}) = \frac{1}{\rho}\sum_{i=1}^n \gamma_i a_i^{\rho}. \label{eq:Y_transformed_pf} \end{equation}
This transformation again does not change the optimal contract or the corresponding equilibrium actions. The role of $\rho$ in the transformation is to ensure that $\widetilde{Y}$ is an increasing function of actions.

We characterize how the optimal contract divides payments among agents who do receive positive payments. Each of these agents has first-order conditions \begin{equation}\label{eq:CES_FOC}a_i^* = \tau_i \widetilde{P}'(\widetilde{Y}^*) \gamma_i \cdot (a_i^*)^{\rho-1}.\end{equation} We next compute productivities and centralities. Differentiating $\widetilde{Y}$, productivities are
$$ \widetilde{\alpha}_i = \gamma_i a_i^{\rho-1}.$$
The matrix $\frac{\partial^2 \widetilde{Y}}{\partial a_i \partial a_j}$ of spillovers is again diagonal with entries $\frac{\partial^2 \widetilde{Y}}{\partial a_i^2}   = \gamma_i (\rho-1)a_i^{\rho-2} .$
So agents' centralities at equilibrium are
\begin{align*}\widetilde{c}_i & =\widetilde{\alpha}_i\left(1 -\tau_i \widetilde{P}'(\widetilde{Y}^*)\gamma_i (\rho-1) \cdot (a_i^*)^{\rho-2} \right)^{-1} \\ & = \frac{\widetilde{\alpha}_i}{2-\rho} \text{ by  \cref{eq:CES_FOC}}. \end{align*}

We can now apply \Cref{t:generalmodel}, which states that under any optimal contract satisfying \Cref{as:balancederive}, the quantity $\frac{\widetilde{\alpha}^2_i}{2-\rho}$ is equal for all agents and thus all agents have the same productivity. This tells us that $(a_i^*)^{1-\rho}$ is proportional to $\gamma_i$. Given this, (\ref{eq:CES_FOC}) implies that each agent's payment $\tau_i$ under the optimal contract is proportional to $\gamma_i^{\frac{1}{1-\rho}}$.

\section{Applicability of the First-Order Approach}\label{a:appendix_indifferences} 

A key assumption in our analysis is that when we perturb the optimal contract in any direction, the induced equilibrium varies in a differentiable way (\Cref{as:balancederive}). We now discuss sufficient conditions for this assumption and provide balance conditions that hold in its absence.

We begin by observing that \Cref{as:balancederive} holds as long as the optimal contract induces a strict equilibrium and a stability property is satisfied. By the implicit function theorem, \Cref{as:balancederive}  is implied by the following two conditions:

\begin{assumption}\label{as:strictstable}
\begin{enumerate}[(a)]
\item (Invertibility of utility Hessian) The  matrix $$
    \left(\frac{\partial^2 \mathcal{U}_{i}}{\partial a_{j} \partial a_{i}}\right)_{i,j}$$ is non-singular at contract $\optsharesvector$ and corresponding equilibrium $\eqlbactions$.
\item (Strictness) The equilibrium $\eqlbactions(\sharesvector^*)$ is strict.
\end{enumerate}
\end{assumption}
This assumption gives more explicit conditions that guarantee a first-order approach applies. Part (a) is weaker than requiring stability of equilibrium under best-reply dynamics.\footnote{If we solve for local best-reply dynamics the Jacobian of the agents' first-order conditions in others' actions is given by the Hessian above. Stability requires that this Jacobian be well-defined and have some further properties.} Part (b) can impose more substantive restrictions, as we now discuss.

\Cref{ss:first-order-multi-agent} shows we can extend the sufficient conditions for strictness from \cite{rogerson1985first} to our multi-agent setting. \Cref{ss:indifferences} presents an analogous balance condition which only requires differentiability in \textit{some} directions, and therefore can provide insights even when equilibrium is not strict under the optimal contract.

\subsection{Sufficient conditions for strictness of equilibrium} \label{ss:first-order-multi-agent} This section establishes conditions on the environment that guarantee strictness of equilibrium actions.  Specifically, equilibrium actions are strict when the probability functions satisfy the \textit{monotone likelihood ratio property} and the \textit{convexity of distribution function property} (as in \citet{rogerson1985first}), and when the team performance is concave in each agent's action.

We can assume without loss of generality that the possible outcomes are $S=\{1,\hdots,|S|\}$ and that outcomes are ordered in increasing value to the principal: $v_1 < v_2 < \hdots < v_{|S|}$. We will use two conditions from \citet{rogerson1985first}:

\begin{assumption}
    \label{as:mlrpcdfc}
    The collection of distribution functions $\{P_{s}(\cdot)\}_{s \in \mathcal{S}}$ satisfies: \begin{enumerate}
        \item \textbf{Monotone Likelihood Ratio Property (MLRP)}: the ratio $P_{s}'(Y)/P_{s}(Y)$ is strictly increasing in the outcome $s$ for any team performance.
        \item \textbf{Convexity of distribution function condition (CDFC)}: for any outcome $s$, the cumulative distribution \begin{align*}
            \sum_{k=1}^{s}P_{k}(Y)
        \end{align*} is convex in team performance. 
    \end{enumerate}
\end{assumption}

MLRP requires that low-valued outcomes become less likely as team performance increases. A key implication is that the cumulative distribution function of any outcome is weakly decreasing in team performance. CDFC strengthens this by requiring convexity of the cumulative distribution, which will be crucial for the strictness of best responses.

A comparison with the single-agent setting is instructive. \citet{rogerson1985first} shows that with one agent taking a one-dimensional action, these conditions are sufficient for a first-order approach. We are able to extend the analysis to multi-agent settings because team performance in our framework depends only on a one-dimensional aggregate of individual actions. To do so, we will require \Cref{as:mlrpcdfc} as well as a concavity assumption on $Y(\bm{a})$. The assumption on $Y(
\bm{a})$ is satisfied automatically in our network games setting and is satisfied for Cobb-Douglas and CES production functions whenever all factor shares are at most $1$. 
\begin{theorem}
    \label{t:eqlbstrict} Suppose the collection of outcome probability functions $\{P_{s}(\cdot)\}_{s \in \mathcal{S}}$ satisfies Assumption~\ref{as:mlrpcdfc} and the conditions in \Cref{fact:opt_exist} are satisfied. Additionally, assume that the team performance function $Y(\cdot)$ is concave in each agent's action, i.e., the second derivative $\frac{\partial^2 Y}{\partial a_i^2}$ is non-positive for every agent $i \in \mathcal{N}$. Then there exists an optimal contract $\bm{\tau}^*$ such that the resulting equilibrium action profile $\bm{a}^*(\bm{\tau}^*)$ is strict. 
\end{theorem}

Recall that the principal chooses a contract to maximize expected profit, taking into account the agents' equilibrium actions in response. The principal's problem, denoted by (\ref{eq:principalmax}), is given by: \begin{align*} \label{eq:principalmax}
    \max_{\bm{\tau},\bm{a}^*(\bm{\tau})} \quad & \sum_{s=1}^{|S|}\left(v_{s}-\sum_{i \in \mathcal{N}}\tau_{i}(s)\right)P_{s}(Y), \\
    \text{s.t} \quad & \bm{a}^*(\bm{\tau}) \text{ is an equilibrium}, \tag{PM}\\
    & \tau_{i}(s) \geq 0 \quad \text{for every } i, \text{ and outcome } s. 
\end{align*} 

To prove \Cref{t:eqlbstrict}, we begin by analyzing a relaxation of the principal's profit maximization problem. In this relaxed formulation, the principal chooses a contract $\bm{\tau}$ and an action profile $\bm{a}^*(\bm{\tau})$ that satisfies the first-order (stationary) conditions of the agents' utility maximization problems. The objective is to maximize expected profits subject to these first-order conditions. The relaxation, denoted by (\ref{eq:principalmaxrelax}), is given by: \begin{align*} \label{eq:principalmaxrelax}
    \max_{\bm{\tau},\bm{a}^*(\bm{\tau})} \quad & \sum_{s=1}^{|S|}\left(v_{s}-\sum_{i \in \mathcal{N}}\tau_{i}(s)\right)P_{s}(Y), \\
    \text{s.t} \quad & \frac{\partial \mathcal{U}_{i}}{\partial a_{i}} = 0 \quad \text{for every } i \in \mathcal{N}, \tag{RPM}\\
    & \tau_{i}(s) \geq 0 \quad \text{for every } i, \text{ and outcome } s. 
\end{align*} 

We now justify the use of this relaxation. Any optimal solution to (\ref{eq:principalmax}) is a feasible solution to (\ref{eq:principalmaxrelax}). To see why, any optimal action for each agent must satisfy that agent's first-order condition. Additionally, any agent that takes a zero action must be paid zero at the optimal contract (else the principal could induce the same equilibrium by providing zero payments to these agents, which would strictly increase profit). By our assumption that the marginal cost at a zero action is zero, such an agent's marginal utility is zero. Thus, all equilibrium actions---whether interior or at the boundary---satisfy the first-order conditions that define the constraint in (\ref{eq:principalmaxrelax}).

The primary challenge lies in establishing that an optimal solution to (\ref{eq:principalmaxrelax}) is a feasible solution to (\ref{eq:principalmax}). Consequently, by the justification in the paragraph above, it would be an optimal solution to the latter. Toward this end, we will show that each agent's utility function is strictly concave at any optimal solution to (\ref{eq:principalmaxrelax}). Feasibility in (\ref{eq:principalmax}) and strictness of the equilibrium action profile will follow.

We begin with a lemma showing it is optimal to pay agents more at better outcomes:

\begin{lemma}
    \label{l:monotonicity} Consider any optimal solution $(\bm{\tau}^*,\bm{a}^*(\bm{\tau}^*))$ to (\ref{eq:principalmaxrelax}). For any agent $i \in \mathcal{N}$, \begin{align*}
        \tau_{i}^*(s+1) \geq \tau_{i}^*(s) \quad \text{ for all } s \in \mathcal{S}.
    \end{align*}  Furthermore, for any agent that receives a positive payment, there must be a strict increase in payments at some outcome. 
\end{lemma}

\begin{proof}
Consider any outcome $s \in \mathcal{S}$ that provides a payment to \textit{some} agent $i$. The principal's KKT first-order condition implies \begin{align} \label{eq:kktratio}
        \frac{1}{u_{i}'(\tau_{i}^*(s))} = \lambda_{i} \frac{\partial Y}{\partial a_{i}} \frac{P_{s}'(Y)}{P_{s}(Y)},
    \end{align} where $\lambda_{i}$ is the Lagrange multiplier corresponding to that agent's marginal utility constraint. We will establish that $\lambda_{i} > 0$ for any agent that receives a payment. Suppose this was not the case. Then, any outcome at which this agent is paid, it must be $P_{s}'(Y) < 0$ (the left hand side of (\ref{eq:kktratio}) is always positive). But this means that the agent must be taking a zero action because the marginal utility is non-positive. This contradicts optimality because the principal can induce the same equilibrium action by giving this agent zero payments and obtain a strictly higher profit. Thus, $\lambda_{i} > 0$ for any agent that receives a payment. 
    
For the rest of the argument, it suffices to focus on these agents --- payments are trivially monotonic for an agent that does not receive a payment at \textit{any} outcome.

Consider any agent $i$ and suppose $\tau_{i}^*(s) > 0$, but $\tau_{i}^*(s+1)=0$. We will show that the principal can strictly improve profits by providing a payment to agent $i$ at outcome $(s+1)$. Recall, the KKT first-order condition for $\tau_{i}^*(s)$, established in (\ref{eq:kktratio}), is \begin{align*}
    1 - u_{i}'(\tau_{i}^*(s)) \lambda_{i} \frac{\partial Y}{\partial a_{i}} \frac{P_{s}'(Y)}{P_{s}(Y)} = 0
\end{align*} By assumption, $\tau_{i}^*(s+1)<\tau_{i}^*(s)$ and the outcomes satisfy MLRP, that is, $$\frac{P_{s+1}'(Y)}{P_{s+1}(Y)} > \frac{P_{s}'(Y)}{P_{s}(Y)}.$$ Combined with the observation $\lambda_{i} > 0$, we get \begin{align} \label{eq:nextoutcomeineq}
    1 - u_{i}'(\tau_{i}^*(s+1)) \lambda_{i} \frac{\partial Y}{\partial a_{i}} \frac{P_{s+1}'(Y)}{P_{s+1}(Y)} < 0.
\end{align} 

The derivative of the Lagrangian w.r.t $\tau_{i}(s+1)$ is \begin{align*}
    \frac{\partial \mathcal{L}}{\partial \tau_{i}(s+1)} = -P_{s+1}(Y)\left(1 - u_{i}'(\tau_{i}(s+1))\lambda_{i}\frac{\partial Y}{\partial a_{i}}\frac{P_{s+1}'(Y)}{P_{s+1}(Y)}\right) > 0,
\end{align*} where the inequality follows from (\ref{eq:nextoutcomeineq}). This contradicts the optimality of $\tau_{i}^*(s+1)=0$, so $\tau_{i}^*(s+1)>0$. Moreover, under MLRP and the concavity of $u_i(\cdot)$, analysis of (\ref{eq:kktratio}) implies that $\tau_i^*(s+1) \geq \tau_i^*(s)$.
\end{proof} 

The intuition behind the result is straightforward. Agents are more responsive to incentives at outcomes that become more likely as team performance improves. Since the principal’s value also increases with the outcome, this alignment leads the principal to offer stronger incentives at higher-valued outcomes.

We utilize the monotonicity result to prove \Cref{t:eqlbstrict}. A key step establishes that each agent's utility function is strictly concave at any optimal solution to (\ref{eq:principalmaxrelax}). 

\begin{proof}[Proof of \Cref{t:eqlbstrict}]
    Consider any optimal solution $(\bm{\tau}^*,\bm{a}^*(\bm{\tau}^*))$ to (\ref{eq:principalmaxrelax}). An argument identical to the proof of \Cref{fact:opt_exist} establishes that an optimal solution to (\ref{eq:principalmaxrelax}) exists under the conditions of \Cref{fact:opt_exist}. We can write the utility of an agent as \begin{align*}
        \mathcal{U}_{i} &= \sum_{s=1}^{|S|}u_{i}(\tau_{i}(s))P_{s}(Y) - C_{i}(a_{i}) \\
        &= u_{i}(\tau_{i}(1)) + \sum_{s=2}^{|S|}\underbrace{\left(u_{i}(\tau_{i}(s))-u_{i}(\tau_{i}(s-1))\right)}_{u_{i}^{\Delta}(\tau_{i}(s))}\sum_{k=s}^{|S|}P_{k}(Y) - C_{i}(a_{i}).
    \end{align*} The second derivative of $\mathcal{U}_{i}$ with respect to $a_{i}$ is \begin{align*}
        \frac{\partial^2 \mathcal{U}_{i}}{\partial a_{i}^2} = \frac{\partial^2 Y}{\partial a_{i}^2} \sum_{s=2}^{|S|}u_{i}^{\Delta}(\tau_{i}(s)) \sum_{k=s}^{|S|}P_{k}'(Y) + \left(\frac{\partial Y}{\partial a_{i}}\right)^2 \sum_{s=2}^{|S|}u_{i}^{\Delta}(\tau_{i}(s)) \sum_{k=s}^{|S|}P_{k}''(Y) -C_{i}''(a_{i}).
    \end{align*} We will argue that the second-derivative is strictly negative for any action profile. By \Cref{l:monotonicity}, each $u_{i}^{\Delta}(\tau_{i}(s))$ is non-negative. By concavity of $Y(\cdot)$ in each agent's action and MLRP, the first term is non-positive. By CDFC, the second term is non-positive. And finally, since the cost function is strictly convex, $C_{i}''(a_{i})>0$. This establishes the fact that the second derivative is strictly negative. Thus, the utility function of any agent is strictly concave. This has two important implications. First, any action profile where all agents' first-order conditions are satisfied is in fact an equilibrium and thus a feasible solution to (\ref{eq:principalmax}). So any optimal solution to (\ref{eq:principalmaxrelax}) is an optimal contract.  Second, the equilibrium action profile at this optimal contract is strict. The statement of the result follows.  
\end{proof}

\subsection{Directional differentiability} \label{ss:indifferences}
This section shows modified balance conditions can continue to hold when the induced equilibrium varies differentiably as the contract is perturbed in some, but not necessarily all, directions. The underlying idea is that the principal may only have available perturbations that preserve some global constraints along with the first-order conditions.

Consider an optimal contract $\bm{\tau}^*$ and a corresponding equilibrium $\eqlbactions(\bm{\tau}^*)$ that is stable but not strict. Then the implicit function theorem ensures that as we perturb the contract $\bm{\tau}^*$ to $\bm{\tau}$ locally, there is an action profile $\bm{a}(\bm{\tau})$ satisfying agents' first-order conditions that is continuously differentiable in $\bm{\tau}$. Since the equilibrium is no longer strict, however, these action profiles need not be equilibria for all $\tau$: if agent $i$ is indifferent to his equilibrium action  $a^*(\bm{\tau}^*)_i$ and some alternative, then after perturbing the contract the locally optimal action ${a}(\bm{\tau})_i$ need no longer be globally optimal. Then there could be perturbations for which the principal would prefer  $\bm{a}(\bm{\tau})$ to $\bm{a}^*(\bm{\tau}^*)$, so our full balance conditions need no longer hold. But if there are perturbations such that  $\bm{a}(\bm{\tau})$ remains an equilibrium, then we can obtain balance conditions in the directions of all such perturbations.

Consider increasing payments to agent $i$ in the direction $\bm{t}_{i}$, where $t_{i}(s)$ specifies the change in payments to agent $i$ in each outcome $s$. We hold fixed payments to all other agents $j\neq i$.
\begin{definition}
    We say \textbf{equilibrium is maintained} in direction $\bm{t}_{i}$ if there exist a one-parameter family $\bm{\tau}(x)$ of contracts with $\bm{\tau}(0)=\bm{\tau}^*$ such that $\tau_i'(0)=\bm{t}_i$, $\tau'_j(0)=0$ for $j \neq i$, and such that $\bm{a}(\bm{\tau}(x))$ is an equilibrium for $x$ in some neighborhood of $0$.
\end{definition}

When there are directions where equilibrium is maintained, we obtain balance conditions in those directions. The statements of these conditions must be modified since we now change payments under multiple outcomes.
\begin{theorem}
\label{t:balanceindifferences}
Suppose $\bm{\tau}^*$ is an optimal contract inducing equilibrium $\eqlbactions$ with team performance $Y^*$. For all directions $\bm{t}_{i}$ such that equilibrium is maintained and $\tau^*_i(s)>0$ whenever $t_i(s) > 0$, we have \begin{align*}
        q \productivity_{i}\centrality_{i}\left(\sum_{s \in \mathcal{S}}P_{s}'(Y^*)t_{i}(s)\cdot \frac{u_{i}'(\tau_{i}^*(s))}{C''_i(a_i^*)}\right) = \sum_{s \in \mathcal{S}}P_{s}(Y^*)t_i(s)
    \end{align*} for a constant $q$ that is independent of the direction $\bm{t}_{i}$ and agent $i$. 
\end{theorem}

\begin{proof}
    The proof follows a similar approach to the proof of \Cref{t:generalmodel}. Consider any agent $i$ receiving a payment in at least one outcome. Also, consider a direction of payment perturbation $\bm{t}_{i}$. As before, we can ignore agents that receive a zero payment at the optimal contract as they continue to take a zero action and not contribute to the change in equilibrium team performance. The change induced by such perturbation $\bm{t}_{i}$ is \begin{align*}
        dY(\bm{t}_{i}) = \nabla Y(\eqlbactions)^{\top} \cdot d\eqlbactions(\bm{t}_{i}),
    \end{align*} where $dY(\cdot)$ is the directional derivative of equilibrium team performance and $d\eqlbactions(\cdot)$ is the directional derivative of equilibrium actions. We state a lemma which characterizes the change in equilibrium actions for such a perturbation.

\begin{lemma}
\label{l:teamperfindiff}
    The change in equilibrium team performance as payments to agent $i$ are perturbed in direction $\bm{t}_{i}$ is \begin{align*}
        dY(\bm{t}_{i}) = l \productivity_{i}\centrality_{i} \sum_{s \in \mathcal{S}}P_{s}'(Y^*)t_{i}(s)\cdot \frac{u_{i}'(\tau_{i}^*(s))}{C''_i(a_i^*)},
    \end{align*} for some constant $l$.
\end{lemma}

\begin{proof}
Consider the equilibrium action profile $\eqlbactions$. For an agent $j$, the first-order conditions imply $a_{j}^*$ must solve the equation \begin{equation} \label{eq:bestresponsegeneralindiff}
    C_{j}'(a_{j}) = \left(\sum_{s' \in \mathcal{S}}P_{s'}'(Y)u_{j}(\tau_{j}(s'))\right) \frac{\partial Y}{\partial a_{j}}.
\end{equation} 

Let us take the directional derivative of (\ref{eq:bestresponsegeneralindiff}) in direction $\bm{t}_{i}$. The expression we obtain depends on whether $j = i$. For all $j \neq i$, \begin{multline*}
    C_{j}''(a_{j})da_{j}^*(\bm{t}_{i}) = \left(\sum_{s \in \mathcal{S}}P_{s}'(Y)u_{j}(\tau_{j}(s))\right) \sum_{k=1}^{n}\frac{\partial^2 Y}{\partial a_{k} \partial a_{j} }da_{k}^*(\bm{t}_{i}) \\ + \frac{\partial Y}{\partial a_{j}}dY(\bm{t}_{i})\sum_{s \in \mathcal{S}}P_{s}''(Y)u_{j}(\tau_{j}(s)).
\end{multline*} On the other hand, for $j=i$ \begin{multline*}
    C_{j}''(a_{j})da_{j}^*(\bm{t}_{i}) = \left(\sum_{s \in \mathcal{S}}P_{s}'(Y)u_{j}(\tau_{j}(s))\right) \sum_{k=1}^{n}\frac{\partial^2 Y}{\partial a_{k} \partial a_{j} }da_{k}^*(\bm{t}_{i}) \\ + \frac{\partial Y}{\partial a_{j}}\sum_{s \in \mathcal{S}}P_{s}'(Y)t_{j}(s)u_{j}'(\tau_{j}(s)) + \frac{\partial Y}{\partial a_{j}}dY(\bm{t}_{i})\sum_{s \in \mathcal{S}}P_{s}''(Y)u_{j}(\tau_{j}(s)).
\end{multline*}

We can combine these equations to write the resulting expression in vector form as \begin{align*}
    d\eqlbactions(\bm{t}_{i}) &= \left[\bm{H}-\bm{U}\bm{G}\right]^{-1} \begin{bmatrix}
        0 \\
        \frac{\partial Y}{\partial a_{i}} \cdot \sum_{s \in \mathcal{S}}P_{s}'(Y)t_{i}(s)u_{i}'(\tau_{i}(s)) \\
        0
    \end{bmatrix} + dY(\bm{t}_{i})\left[\bm{H}-\bm{U}\bm{G}\right]^{-1}\bm{d},
\end{align*} where $\bm{d}$ is a vector with $j^{\text{th}}$ element defined as \begin{align*}
    d_{j} := \frac{\partial Y}{\partial a_{j}} \cdot \sum_{s \in \mathcal{S}}P_{s}''(Y)u_{j}(\tau_{j}(s)).
\end{align*} The change in equilibrium team performance is \begin{multline*}
    dY(\bm{t}_{i}) = \nabla Y(\eqlbactions)^{T} \left[\bm{H}-\bm{U}\bm{G}\right]^{-1} \begin{bmatrix}
        0 \\
        \frac{\partial Y}{\partial a_{i}} \cdot \sum_{s \in \mathcal{S}}P_{s}'(Y)t_{i}(s)u_{i}'(\tau_{i}(s)) \\
        0
    \end{bmatrix} \\ + dY(\bm{t}_{i}) \nabla Y(\eqlbactions)^{T} \left[\bm{H}-\bm{U}\bm{G}\right]^{-1}\bm{d}.    
\end{multline*} Applying the definitions of $\productivity_{i}$ and $\centrality_{i}$ while rearranging the equation, we obtain \begin{align*}
    dY(\bm{t}_{i}) &= \frac{1}{1-\nabla Y(\eqlbactions)^{T} \left[\bm{H}-\bm{U}\bm{G}\right]^{-1}\bm{d}} \productivity_{i}\centrality_{i} \sum_{s \in \mathcal{S}}P_{s}'(Y^*)t_{i}(s) \cdot \frac{u_{i}'(\tau_{i}(s))}{C''_i(a_i^*)}, \\
    &= l \productivity_{i}\centrality_{i} \sum_{s \in \mathcal{S}}P_{s}'(Y^*)t_{i}(s)\cdot \frac{u_{i}'(\tau_{i}(s))}{C''_i(a_i^*)},
\end{align*}where $l = \frac{1}{1-\nabla Y(\eqlbactions)^{T} \left[\bm{H}-\bm{U}\bm{G}\right]^{-1}\bm{d}}$. Observing that $l$ does not depend on $i$, we obtain the desired result.
\end{proof}

To prove the theorem, we utilize this expression for change in team performance to analyze the principal's first-order condition. For any agent $i$, consider a direction of perturbation $\bm{t}_{i}$ such that whenever $t_{i}(s) > 0$, the optimal contract is such that $\tau_{i}^*(s) > 0$. Because equilibrium is maintained in direction  $\bm{t}_{i}$ , the following principal first-order condition must hold: \begin{align*}
    dY(\bm{t}_i{}) \underbrace{\sum_{s \in \mathcal{S}}\left(v_{s} - \sum_{i=1}^{n}\tau_{i}(s)\right)P_{s}'(Y^*)}_{D} = \sum_{s \in \mathcal{S}}P_{s}(Y^*)t_{i}(s)
\end{align*} Substituting the expression in \Cref{l:teamperfindiff}, we get \begin{align*}
    lD \productivity_{i}\centrality_{i} \sum_{s \in \mathcal{S}}P_{s}'(Y^*)t_{i}(s)\cdot \frac{u_{i}'(\tau_{i}^*(s))}{C''_i(a_i^*)} = \sum_{s \in \mathcal{S}}P_{s}(Y^*)t_{i}(s).
\end{align*} The expression in the theorem follows by taking $q = lD$.
\end{proof}

\section{Contract design under measurement error} %

\label{a:centralityestimate}

Our analysis has assumed the principal knows the production function $Y(\bm{a})$. In this section we ask the following question: when can the principal improve on a given contract using only an imperfect measurement of a team's network structure? More concretely: the principal can check our balance condition using the  observed technological complementarity matrix. Suppose that this suggests shifting incentives from one agent $j$ to another agent $i$ would be beneficial. Will this shift actually increase the principal's payoff under the true network?

We continue to work in the setting of \Cref{ex:main} with linear-quadratic production and  risk-neutral agents:
  $$Y(\bm{a}) = \sum_{i=1}^n k_i a_i + \frac{1}{2}\sum_{i,j=1}^n G_{ij} a_i a_j.$$
Suppose a principal observes productivities $\bm{\alpha}$, the derivative $P'(Y^*)$, and an imperfect measurement $\bm{G}'$ of $\bm{G}$. Furthermore, suppose the principal naively calculates the centrality vector using the measured technological complementarity matrix $\bm{G}'$ as if it were $\bm{G}$:
$$\bm{c}'=  \bm{\alpha}^{\top} (I-P'(Y^*) \bm{T} \bm{G}')^{-1}.$$
Suppose that under this calculation, the balance condition suggests increasing $\tau_i$ and decreasing $\tau_j$ will be profitable, i.e.,
\begin{equation} \alpha_i c'_i -\alpha_jc'_j > \epsilon \label{eq:apparent_reallocation} \end{equation}
for some $\epsilon>0$. We ask when the principal can guarantee that the analogous statement holds for the quantities derived from the true network, i.e., $$\alpha_i c_i - \alpha_j c_j > 0.$$
This inequality implies that marginally increasing $\tau_i$ and decreasing $\tau_j$ will indeed increase revenue. 

The measurement of the network clearly needs to have enough precision, i.e., the difference between $\bm{G}$ and $\bm{G}'$ should be small in some sense. Our (standard) measure of this difference will be the  spectral norm $\|\bm{G}-\bm{G}'\|_2$. Our question then becomes: is the principal safe in taking the approximation at face value if this difference is small? Though the answer turns out to be no in general, we can give several sufficient conditions. The first is that the spillovers between agents are not very large, as we now discuss.
\medskip

\textbf{Bounded spillovers}: Suppose that $\|\bm{T G }\|_2 P'(Y^*) < r$ for some fixed constant $r<1$. The assumption says that spillovers between agents are not too large: there is a bound  on the total effects of direct and indirect spillovers. Under this condition, the next result shows that the principal can ignore small amounts of measurement error in checking for profitable deviations. Moreover, the proof gives an explicit bound on how much measurement error is permissible.

\begin{proposition}\label{p:bounded_spillovers}
    Suppose  $\|\bm{T G}\|_2 P'(Y^*) < r$ for some fixed constant $r<1$. If $\alpha_i c_i' - \alpha_j c_j' > \epsilon$ for some $\epsilon>0$, then we can choose $\delta>0$ independent of $\bm{G}$ such that $\alpha_i c_i - \alpha_j c_j > 0$ whenever $\|\bm{G}-\bm{G}'\|_2 < \delta$.
\end{proposition}

\begin{proof}
We take $\delta$ small enough so that $r+\delta P'(0) < 1$, and therefore $r+\delta P'(Y)<1$ for all $Y$ by concavity of $P(Y)$. Note that $$\underbrace{(I - P'(Y^*)\bm{T} \bm{G} )^{-1}}_{\bm{M}}-\underbrace{(I - P'(Y^*) \bm{T} \bm{G}')^{-1}}_{\bm{M}'}  = (I - P'(Y^*)  \bm{T}\bm{G})^{-1}(P'(Y^*)\bm{T}(\bm{G}-\bm{G}'))(I - P'(Y^*) \bm{T}\bm{G}' )^{-1}.$$
Solving for the norm of $\bm{M}-\bm{M}'$, \begin{align*}
\|\bm{M} - \bm{M}'\|_2 &  =  \| (I - P'(Y^*) \bm{T}\bm{G} )^{-1}P'(Y^*)\bm{T}(\bm{G}-\bm{G}')(I - P'(Y^*) \bm{T}\bm{G}')^{-1} \|_2,
\\ & \leq \frac{\delta P'(Y^*)}{(1-r)(1-r-\delta P'(Y^*))} \quad (\text{by } \max_i \tau_i \leq 1),
 \\ & \leq  \frac{\delta P'(0)}{(1-r)(1-r-\delta P'(0))} \quad (\text{by concavity of } P(Y)).
\end{align*} 
This implies $$\|\bm{c}-\bm{c}'\|_2 \leq \|\bm{\alpha}\|_2  \cdot \frac{\delta P'(0)}{(1-r)(1-r-\delta P'(0))} .$$
So $$\|\bm{\alpha}(\bm{c}-\bm{ c'})\|_2  \leq \|\bm{\alpha}\|^2_2  \cdot \frac{\delta P'(0)}{(1-r)(1-r-\delta P'(0))}.$$
The result follows from choosing $\delta$ small enough so that the right-hand side is less than $\epsilon$.
\end{proof}

\medskip
When spillovers are very large, however, productivities and centralities can be very sensitive to small perturbations of the network. In this case, a violation of the balance condition based on the observed network need not give a profitable deviation for the principal. The following example shows that when spillovers are very large, the effects of transferring bonuses between agents may be very sensitive to mismeasurement.

\begin{example2} Consider a network with four agents and two non-zero links $G_{12}$ and $G_{34}$, with $$Y = \sum_i a_i + G_{12}a_1a_2+G_{34}a_3a_4.$$ Suppose that agents are risk-neutral and $P(Y)$ is locally linear with slope $P'(Y^*)$. By taking $Y^*$ close to zero, we can take $P'(Y^*)$ to be arbitrarily large.

The network consists of two distinct subteams with a single link in each. We now show that small measurement error can generate large observed differences in centralities between the two subteams.

Consider a contract $\tau = (\frac18,\frac18,\frac18,\frac18)$. Equilibrium actions are
$$a_1^* = a_2^* = \frac{P'(Y^*)}{8}\cdot \frac{1}{1-P'(Y^*) G_{12}/8}, a_3^* = a_4^* = \frac{P'(Y^*)}{8}\cdot \frac{1}{1-P'(Y^*) G_{34}/8}.$$
The principal observes $\bm{G}'$ and the productivities $$\alpha_1=\alpha_2 = 1+G_{12}a_1^*, \alpha_3=\alpha_4=1+G_{34}a_3^*.$$
We suppose that $\bm{G}'$ also has two non-zero links $G'_{12}$ and $G'_{34}$ (and no other links). So the principal naively calculates centralities
$$\bm{c}_1'=\bm{c}_2'=\alpha_1 \frac{1}{1-P'(Y^*) G_{12}'/8}, \bm{c}_3'=\bm{c}_4'=\alpha_3 \frac{1}{1-P'(Y^*) G_{34}'/8}.$$

Fix $\epsilon>0$. Let $$G_{12} = g + \eta/2, G_{34} =g, G_{12}'=g, G_{34}' = g+\eta,$$
where $\eta>0$ is arbitrary. So the link between $1$ and $2$ is stronger than the link between $3$ and $4$, but the principal believes the link between $3$ and $4$ is the stronger one due to measurement error. We calculate \begin{align*}\alpha_1 c'_1  & =\left(1+\frac{P'(Y^*)}{8}\cdot \frac{g + \eta/2}{1-P'(Y^*) (g+\eta/2)/8}\right)^2  \cdot \frac{1}{1-P'(Y^*) g/8} \\ & < \left(1+\frac{P'(Y^*)}{8}\cdot \frac{g}{1-P'(Y^*) g/8} \right)^2 \cdot \frac{1}{1-P'(Y^*) (g+\eta)/8}\\& =\alpha_3c'_3\end{align*} for $g$ such that $P'(Y^*)g/8$ is sufficiently close to $1$. By taking $g$ such that $P'(Y^*)g/8$ is sufficiently close to $1$, we can ensure
$\alpha_3 c'_3- \alpha_1c'_1 \geq \epsilon .$
Since $\alpha_1>\alpha_3$ and $c_1>c_3$, however, $ \alpha_3 c_3- \alpha_1c_1 < 0.$ This all holds taking $\eta$, and therefore the norm $\|\bm{G}-\bm{G}'\|_2$, to be arbitrarily small.

\end{example2}

The key issue in the example is that the principal must decide how to allocate incentives between two disconnected sub-teams. Spillovers in each of these teams are strong, so a small error in measuring the network within a team can be amplified substantially. This makes identifying which subteam is more responsive to incentives important, but measurement error may make it impossible. We can rule out this situation with a standard assumption ensuring the network cannot be partitioned into two subteams with only very weak connections between them.

\medskip

\textbf{Bounded Spectral Gap}: Suppose that the spectral gap (i.e., the difference between the absolute values of the largest and second-largest eigenvalues) of $\bm{T G}$ is bounded below by some positive constant. (This is a standard measure of segregation---see, e.g., \citet{vonLuxburg2007} and \citet{GolubJackson2012}.) Then the principal can adapt the balance condition to ensure violations guarantee profitable changes in contract as long as the measurement error $\|\bm{G}-\bm{G}'\|_2$ is not too large. 

The key idea is that the principal decomposes $\bm{G}'$ into a leading principal component and an orthogonal residual component: $$ \bm{G}' = \bm{L}' + \bm{R}'.$$ 
The apparent value of reallocating incentives from agent $i$ to agent $j$, as in (\ref{eq:apparent_reallocation}), can now be calculated as the sum of two parts: (i)  the contribution to this based on the leading component of the network and (ii) the contribution from the residual part of the network, which captures the group structure of the network. The argument shows that  (i) can be very sensitive to mismeasurement, but that mismeasurement will only affect its magnitude and not its sign. The contribution of (ii) is not very sensitive to mismeasurement because the spectral gap is not too small, which ensures the spillovers in the residual components are bounded.

To formalize this, let $\bm{v}_1,\hdots,\bm{v}_n$ and $\lambda_1,\hdots,\lambda_n$ be the eigenvectors and eigenvalues, respectively, of $\bm{T^{\frac12} G T^{\frac12}},$ chosen so that the eigenvalues satisfy $|\lambda_1|>\hdots > |\lambda_n|$, the eigenvectors form an orthonormal basis, and $\bm{v}_1$ is normalized to have positive entries. Let $\bm{v}'_1,\hdots,\bm{v}'_n$ and $\lambda'_1,\hdots,\lambda'_n$ be the eigenvectors and eigenvalues, respectively, of $\bm{T^{\frac12} G' T^{\frac12}},$ chosen in the same way.  The final result of this section shows that when measurement error is small, the principal can confirm a deviation be profitable by checking versions of the balance condition in the leading principal component $\bm{L}'$ and the residual component $\bm{R}'$. We note that the orthogonal transformation to the eigenbasis introduces some additional terms related to the contract $\bm{\tau}$.

\begin{proposition}
Suppose $P'(Y^*)\lambda_1 - P'(Y^*)|\lambda_2|>r'$ for some constant $r'>0$. If $$\tau_i^{-\frac12}\alpha_i v_{1i}' - \tau_j^{-\frac12}\alpha_j v_{1j}' > \epsilon \text{ and } \sum_{\ell>1}  \langle \bm{T}^{\frac12}\bm{c}', \bm{v}'_{\ell} \rangle( \tau_i^{-\frac12}\alpha_i  v_{\ell i}' -\tau_j^{-\frac12} \alpha_j v_{\ell j}') > \epsilon $$
for some $\epsilon>0$, then we can choose $\delta>0$ independent of $\bm{G}$ such that $\alpha_i c_i - \alpha_j c_j>0$ whenever $\|\bm{G}-\bm{G}'\|_2 < \delta$.
\end{proposition}
\begin{proof}

We will show that we can choose $\delta>0$ independent of $\bm{G}$ such that $$\tau_i^{-\frac12}\alpha_i v_{1i} - \tau_j^{-\frac12}\alpha_j v_{1j} > 0\text{ and } \sum_{\ell>1} \langle \bm{T}^{\frac12}\bm{c}, \bm{v}_{\ell} \rangle( \tau_i^{-\frac12}\alpha_i  v_{\ell i} - \tau_j^{-\frac12}\alpha_j v_{\ell j}) \geq 0 .$$
This will imply
\begin{align*}
\alpha_ic_i - \alpha_j c_j & =  \sum_{\ell =1}^n \langle \bm{T}^{\frac12}\bm{c}, \bm{v}_{\ell} \rangle( \tau_i^{-\frac12}\alpha_i  v_{\ell i} -\tau_j^{-\frac12} \alpha_j v_{\ell j})
\\ & = \langle  \bm{T}^{\frac12}\bm{c}, \bm{v}_{1}  \rangle (\tau_i^{-\frac12}\alpha_i v_{1i} - \tau_j^{-\frac12}\alpha_j v_{1j} ) + \sum_{\ell>1}\langle \bm{T}^{\frac12} \bm{c}, \bm{v}_{\ell} \rangle( \tau_i^{-\frac12}\alpha_i  v_{\ell i} - \tau_j^{-\frac12}\alpha_j v_{\ell j})
\\ & > 0
\end{align*}
as desired.

\textbf{Step (i)}: $ \tau_i^{-\frac12}\alpha_i v_{1i} - \tau_j^{-\frac12}\alpha_j v_{1j} > 0$.

We begin by bounding $\|\bm{v}_1-\bm{v}_1'\|_2$. We have
\begin{align*} \|\bm{T}^{\frac12}\bm{G}\bm{T}^{\frac12}\bm{v}_1' \|_2  &\geq \|\bm{T}^{\frac12}\bm{G}'\bm{T}^{\frac12}\bm{v}_1' \|_2 - \|\bm{T}^{\frac12}(\bm{G}' - \bm{G})\bm{T}^{\frac12}\bm{v}_1'  \|_2,
\\ & \geq \|\lambda_1' \bm{v}_1'\|_2 -\|\bm{T}^{\frac12}\|_2\|\bm{G}' - \bm{G}\|_2\|\bm{T}^{\frac12}\|_2\|\bm{v}_1'   \|_2,
\\ & \geq \lambda_1' - \delta, \\
& \geq \lambda_{1} - 2\delta.\end{align*}
The final inequality follows from the bound $\|\bm{T}^{\frac{1}{2}}\left(\bm{G}'-\bm{G}\right)\bm{T}^{\frac{1}{2}}\|_{2} < \delta$ which implies $|\lambda_{1}-\lambda_{1}'| < \delta$ by Weyl's theorem. 

On the other hand, writing $\bm{v}_1' = \sum w_i v_i$ with $\sum_i w_i^2=1,$ we have $$ \|\bm{T}^{\frac12}\bm{G}\bm{T}^{\frac12}\bm{v}_1' \|_2  = \left(\sum_i w_i^2 |\lambda_i|^2\right)^{\frac12}.$$
Combining our two expressions for $\|\bm{T}^{\frac12}\bm{G}\bm{T}^{\frac12}\bm{v}_1' \|_2$,
$$\left(\sum_i w_i^2 |\lambda_i|^2\right)^{\frac12} \geq \lambda_1 - 2\delta.$$
Since $\lambda_1-|\lambda_2| > r'/P'(Y^*)$ and $\lambda_1P'(Y^*) < 1$, this implies $w_1 \rightarrow 1$ as $\delta \rightarrow 0$. So we can choose $f(\delta)$ depending only on $r'$ and $P'(Y^*)$ and converging to $0$ as $\delta \rightarrow 0$ such that $\|\bm{v}_1-\bm{v}_1'\|_2 < f(\delta)$ whenever $\|\bm{G}' - \bm{G}\|_2 < \delta$.

Since $\tau_i^{-\frac12}\alpha_i v_{1i}' -\tau_j^{-\frac12} \alpha_j v_{1j}' > \epsilon$, we can choose $\delta$ sufficiently small so that  $$\tau_i^{-\frac12}\alpha_i v_{1i} -\tau_j^{-\frac12} \alpha_j v_{1j} > 0$$ whenever $\|\bm{v}_1-\bm{v}_1'\|_2 < f(\delta)$ and therefore whenever $\|\bm{G}' - \bm{G}\|_2 < \delta$.

\textbf{Step (ii)}: $\sum_{\ell>1} \langle  \bm{T}^{\frac12} \bm{c}, \bm{v}_{\ell} \rangle( \tau_i^{-\frac12}\alpha_i  v_{\ell i} -\tau_j^{-\frac12} \alpha_j v_{\ell j})  \geq 0$.

We can write $\bm{T}^{\frac12}\bm{G}\bm{T}^{\frac12} = \bm{H}_1 + \bm{H}_2$, where $\bm{H}_1$ is rank one with eigenvalue $\lambda_1$ and eigenvector $v_1$. Similarly we can write $\bm{T}^{\frac12}\bm{G}'\bm{T}^{\frac12} =\bm{H}_1'+\bm{H}_2'$. Then $$\sum_{\ell>1}\langle \bm{T}^{\frac12} \bm{c}, \bm{v}_{\ell} \rangle \bm{v}_{\ell}=(\bm{I}-P'(Y^*)\bm{H_2})^{-1} \bm{T}^{\frac12}\bm{\alpha} - \langle \bm{T}^{\frac12} \bm{\alpha}, \bm{v}_{1} \rangle \bm{v}_1 
$$
and 
$$\sum_{\ell>1}\langle \bm{T}^{\frac12} \bm{c}', \bm{v}'_{\ell} \rangle \bm{v}'_{\ell}=(\bm{I}-P'(Y^*)\bm{H_2}')^{-1} \bm{T}^{\frac12}\bm{\alpha} - \langle \bm{T}^{\frac12} \bm{\alpha}, \bm{v}'_{1} \rangle \bm{v}'_1 .$$

We have $\|P'(Y^*)\bm{H}_2\|_2 \leq P'(Y^*)|\lambda_2|$. Because $P'(Y^*)\lambda_1 < 1$ and $\lambda_1-|\lambda_2| > r'>0$, we can choose $r<1$ independent of $\bm{G}$ such that $ P'(Y^*)|\lambda_2|<r$. Since  $\|\bm{H}_2-\bm{H}_2'\|_2 \leq \|\bm{G}-\bm{G}'\|_2 + \|\bm{H}_1 - \bm{H}_1'\|_2$, essentially the same argument as in the proof of \Cref{p:bounded_spillovers} gives a bound on $$\|(\bm{I}-P'(Y^*)\bm{H_2})^{-1} \bm{T}^{\frac12}\bm{\alpha} -(\bm{I}-P'(Y^*)\bm{H_2}')^{-1} \bm{T}^{\frac12}\bm{\alpha} \|_2 $$
in terms of $\delta>0$. We note the bound now includes an additional factor $\max_i \tau_i^{-\frac12}$, which we can assume is well-defined by dropping any agents receiving zero payment.

The bound on $\|\bm{v}_1-\bm{v}_1'\|_2$ from Step (i) implies that we can bound $\|\langle \bm{T}^{\frac12} \bm{\alpha}, \bm{v}_{1} \rangle \bm{v}_1 -\langle \bm{T}^{\frac12} \bm{\alpha}, \bm{v}'_{1} \rangle \bm{v}'_1\|_2$ in terms of $\delta$.

Combining these bounds, we can bound $$\left\|\sum_{\ell>1}\langle \bm{T}^{\frac12} \bm{c}, \bm{v}_{\ell} \rangle \bm{v}_{\ell} - \sum_{\ell>1}\langle \bm{T}^{\frac12} \bm{c}', \bm{v}'_{\ell} \rangle \bm{v}'_{\ell}\right\|_2$$ in terms of $\delta$. So when $\delta>0$ is sufficiently small we must have
 $$\sum_{\ell>1} \langle  \bm{T}^{\frac12} \bm{c}, \bm{v}_{\ell} \rangle( \tau_i^{-\frac12}\alpha_i  v_{\ell i} -\tau_j^{-\frac12} \alpha_j v_{\ell j})  \geq 0$$ as desired.
\end{proof}

\section{Sufficient conditions for positive payments}
\label{a:sufficientconditionspay}

The balance result in \Cref{t:generalmodel} only applies to agents receiving a positive payment under a given outcome. As discussed in \Cref{s:applyparametric}, not all agents necessarily receive positive payments at the optimal contract. In this section, we provide sufficient conditions on the environment which guarantee every agent receives a positive payment at some outcome. 

\begin{assumption}
    \label{as:extensivemargin} The environment is such that:
    \begin{enumerate}
    \item The contract giving payments $\tau_i(s)=0$ for all $i$ and $s$ is not optimal. Moreover, under the optimal contract $\bm{\tau}^*$, there is an agent $i$ satisfying $\tau_{i}^*(s)>0$ for some $s$ and centrality $c_{i}>0$.  
    \item For every agent $i$,  $
            \lim_{\tau \to 0}u_{i}'(\tau) = \infty.
        $
    \item For all $i$ and $j$ and any action profile $\bm{a}$,
    \begin{align*}
        \frac{\partial^2 Y (\bm{a})}{\partial a_{j}\partial a_{i}} \geq 0.
    \end{align*}
    \end{enumerate}
\end{assumption}

Part (a) ensures that the principal finds it optimal to pay at least one agent in the team and that there is an agent who receives a payment at some outcome and has a positive centrality. Part (b) is a standard Inada condition for the agent's utility. Part (c) says that agents' actions are complements, and performance is convex in each agent's own effort.

Under these assumptions, all agents are paid precisely at the outcomes that are more likely when team performance increases slightly.

\begin{proposition}
\label{p:inadaallpay}
 Suppose $\optsharesvector$ is an optimal contract with induced team performance $Y^*$. For any agent $i$ and any outcome $s$, \begin{align*}
        \tau_{i}^*(s) > 0 \quad \textit{if and only if } \quad P_{s}'(Y^*) > 0. 
    \end{align*}
\end{proposition}

In general, it can be optimal to exclude some agents from the optimal team (by offering them no incentives). The proposition states that when agents are sufficiently risk averse and actions are complementary, it is optimal to include all agents. Moreover, all agents are paid under all outcomes that would become more likely if they increased their effort.

The remainder of this section proves the proposition. We begin by stating a key lemma, which will help show that at any optimal contract all agents have positive centralities. %

\begin{lemma}
\label{t:spilloverspectralradius}
At any optimal contract $\optsharesvector$, the spillover matrix $\bm{H}^{-1}\bm{U}\bm{G}$ has spectral radius strictly smaller than $1$. 
\end{lemma}

\begin{proof}
We will use $\rho$ to denote the spectral radius of $\bm{H}^{-1}\bm{U}\bm{G}$. We will show that the spillover matrix cannot have $\rho \geq 1$ at any optimal contract. To do so, we will construct a perturbation of the contract giving the principal a higher payoff.

By definition (see \Cref{ss:notation}), 
\begin{align*} \bm{\centrality}^{T} \left[\bm{I}-\bm{H}^{-1} \bm{U} \bm{G} \right] = \bm{\productivity}^{T}. \end{align*}
It is helpful to recall the definitions of the terms in the spillover matrix. The matrix $\bm{U}$ is diagonal with entries \begin{align*}
    U_{jj} = \sum_{s \in \mathcal{S}}P_{s}'(Y^*)u_{j}(\tau_{j}^*(s)).
\end{align*} We showed in \Cref{l:positivepaymentoutcomes} that any agent receives a positive payment under an optimal contract $\optsharesvector$ only at outcomes where $P_{s}'(Y^*) > 0$. An implication of this is that each diagonal entry in $\bm{U}$ is positive. The matrix $\network$ is non-negative by Part (c) of \Cref{as:extensivemargin}. The matrix $\bm{H}$ is diagonal with entries \begin{align*}
    H_{jj} = C_{j}''(a_{j}^*).
\end{align*} Since we assume cost functions are strictly convex, these diagonal entries are positive. It follows that $\bm{H}^{-1} \bm{U} \bm{G} $ is non-negative, so by the Perron-Frobenius theorem this matrix has a right eigenvector $\bm{p}$ with non-negative real entries and a positive real eigenvalue. Then 
\begin{align*} \bm{\centrality}^{T} \left[\bm{I}-\bm{H}^{-1} \bm{U} \bm{G} \right] \bm{p} = \bm{\productivity}^{T} \bm{p}, \end{align*} which can be simplified to \begin{align} \label{eq:centprodeigeneq}
    (1-\rho)\bm{\centrality}^{T}\bm{p} = \bm{\productivity}^{T} \bm{p}.
\end{align}

Now, suppose that $\bm{H}^{-1} \bm{U} \bm{G}$ has spectral radius $\rho \geq 1$. By assumption, the team performance $Y(\cdot)$ is strictly increasing in each of its arguments. It follows that the right-hand side of (\ref{eq:centprodeigeneq}) is positive. If the spillover matrix has spectral radius equal to $1$, the left-hand is $0$ while the right-hand is positive. Thus, we cannot have spectral radius $1$.

We must show we cannot have a spectral radius $\rho > 1$. Since $1-\rho$ is negative, this implies $\bm{\centrality}^{T} \bm{p}$ is negative. We will construct a direction $\bm{t}$ such that when the optimal contract $\optsharesvector$ is perturbed in direction $\bm{t}$, the agents' individual incentives move in direction $-\bm{p}$. The resulting change in the equilibrium team performance is proportional to $-\bm{\centrality}^{T}\bm{p}$, which is positive. This will contradict the optimality of contract $\optsharesvector$.

We now state a lemma that characterizes the change in team performance when the contract is perturbed in \textit{some} direction $\bm{t}$, where $t_{i}(s)$ specifies the change in payments to agent $i$ in each outcome $s$. (Payments only to agents receiving a positive payment are perturbed.) The result below generalizes \Cref{l:changeteamperformance} to any arbitrary direction.  

\begin{lemma}
\label{l:teamperfdirec}
    The change in equilibrium team performance as payments to agents are perturbed in direction $\bm{t}$ is \begin{align*}
        dY(\bm{t}) = l \sum_{i} \productivity_{i}\centrality_{i} \left(\sum_{s \in \mathcal{S}}P_{s}'(Y^*)t_{i}(s)\cdot \frac{u_{i}'(\tau_{i}^*(s))}{C''_i(a_i^*)}\right),
    \end{align*} for some constant $l$.
\end{lemma}

The proof of the result above follows the same approach as the proof of \Cref{l:teamperfindiff} so we omit it for brevity. We utilize this expression of change in team performance to analyze the principal's first-order condition. The derivative of the principal's objective is: \begin{align*}
    dY(\bm{t}) \underbrace{\sum_{s \in \mathcal{S}}\left(v_{s} - \sum_{i=1}^{n}\tau_{i}^*(s)\right)P_{s}'(Y^*)}_{D} - \sum_{i}\sum_{s \in \mathcal{S}}P_{s}(Y^*)t_{i}(s)
\end{align*} Substituting the expression in \Cref{l:teamperfdirec}, we get \begin{equation} \label{eq:principalfoc}
    lD \sum_{i} \productivity_{i}\centrality_{i} \sum_{s \in \mathcal{S}}P_{s}'(Y^*)t_{i}(s)\cdot \frac{u_{i}'(\tau_{i}^*(s))}{C''_i(a_i^*)} - \sum_{i}\sum_{s \in \mathcal{S}}P_{s}(Y^*)t_{i}(s).
\end{equation} 

We will show there exists a direction of perturbation $\bm{t}$ satisfying the following properties: \begin{itemize}
    \item Every element of eigenvector $\bm{p}$ satisfies $$p_{i} = -\productivity_{i}\sum_{s \in \mathcal{S}}P_{s}'(Y^*)t_{i}(s)\cdot \frac{u_{i}'(\tau_{i}^*(s))}{C''_i(a_i^*)},$$
    \item $t_{i}(s) \leq 0$ for any agent $i$ and outcome $s$, and 
    \item $t_{i}(s) < 0$ only if $\tau_{i}^*(s) > 0$. 
\end{itemize}

For each agent $i$, choose an outcome $s_{i}$ where he receives a positive payment. (Recall we have already restricted to the set of agents who receive a positive payment at some outcome, so this is possible.) At such an outcome $s_{i}$, the probability $P_{s_{i}}'(Y^*) > 0$. For outcome $s_{i}$, define \begin{align*}
    t_{i}(s_{i}) := -\frac{p_{i} C''_i(a_i^*)}{\productivity_{i}P_{s_{i}}'(Y^*)u_{i}'(\tau_{i}^*(s_{i}))}.
\end{align*} We have $t_{i}(s_{i}) \leq 0$ because $\productivity_{i} > 0$ and the entries of $\bm{p}$ are non-negative. For any other outcomes $s \in \mathcal{S} \setminus s_{i}$, define $t_{i}(s):=0$. 

Substituting in (\ref{eq:principalfoc}), the derivative of the principal's objective in direction $\bm{t}$ is \begin{align} \label{eq:focineq}
    -lD\bm{\centrality}^{T}\bm{p} - \sum_{i}P_{s_{i}}(Y^*)t_{i}(s_{i}).
\end{align} The rest of the proof establishes that the expression above is positive. First, we show that the quantity $lD>0$. By Part (a) of Assumption~\ref{as:extensivemargin}, at the optimal contract $\bm{\tau}^*$, there exists some agent $i$ satisfying the following conditions on payments and centrality: $\tau_{i}^*(s)>0$ for some outcome $s$ and $c_{i}>0$. It must be that the principal's first-order condition (see proof of \Cref{t:generalmodel}) applied to agent $i$ is satisfied with equality, that is, \begin{align*}lD\alpha_{i}c_{i}P_{s}'(Y^*)\cdot \frac{u_{i}'(\tau_{i}^*(s))}{C''_i(a_i^*)} = P_{s}(Y^*).
\end{align*} Moreover, since $c_{i}>0$ we must have $lD>0$. Combining this with $\bm{\centrality}^{T}\bm{p} < 0$ it follows that the first term $-lD\bm{\centrality}^{T}\bm{p}>0$. The inequality \begin{align}
    -lD\bm{\centrality}^{T}\bm{p} - \sum_{i}P_{s_{i}}(Y^*)t_{i}(s_{i}) > 0
\end{align} follows from noting that $t_{i}(s_{i}) \leq 0$ for every agent $i$. 

Since $\tau_{i}^*(s) > 0$ whenever $t_i(s) \neq 0$, a sufficiently small perturbation in direction $\bm{t}$ is feasible. So $\optsharesvector$ cannot be optimal, which gives a contradiction. We conclude that at the optimal contract $\optsharesvector$, the spillover matrix $\bm{H}^{-1}\bm{U}\bm{G}$ has spectral radius $\rho < 1$.
\end{proof}

To characterize whether an agent receives a positive payment, it is useful to know whether the agent's centrality is positive. We can apply \autoref{t:spilloverspectralradius} to show that the centrality of each agent receiving a positive payment is strictly positive. To see this, recall that centralities are defined by \begin{align*}
    \bm{\centrality}^{T} = \bm{\productivity}^{T}\left[\bm{I}-\bm{H}^{-1}\bm{U}\bm{G}\right]^{-1}.
\end{align*} Since at an optimal contract the spectral radius of $\bm{H}^{-1}\bm{U}\bm{G}$ is strictly smaller than $1$, we can expand the right-hand side as a power series: \begin{align*}
    \bm{\centrality}^{T} = \bm{\productivity}^{T} \sum_{k=0}^{\infty}\left(\bm{H}^{-1}\bm{U}\bm{G}\right)^k.
\end{align*} The spillover matrix $\bm{H}^{-1}\bm{U}\bm{G}$ is non-negative (see the proof of \Cref{t:spilloverspectralradius}) while each entry of $\bm{\productivity}$ is strictly positive. We conclude that the centrality $\centrality_{i}$ of each agent receiving a positive payment is strictly positive.

Unfortunately, this does not allow us to conclude that \textit{all} agents receive a payment. Suppose some agent received payment zero under $\optsharesvector$. The Inada condition guarantees a small payment to that agent under suitable outcomes would provide a large incentive to work. But whether this incentive helps the principal depends on the sign of that agent's centrality. We will now show that agents that do not receive a payment at the optimal contract have a strictly positive centrality as well.

Our definition of centrality in \Cref{ss:notation} focused on agents that receive a payment. We will extend the definition to \textit{all} agents.  To do so, we extend various other definitions to allow entries for every agent. Define the vector $\widetilde{\bm{\productivity}}$ by \begin{align*}
        \widetilde{\bm{\productivity}} := \nabla Y(\eqlbactions).
    \end{align*} Define the matrix $\widetilde{\bm{U}} \in \mathbb{R}^{n \times n}$ to be diagonal with entries \begin{align*}
        \widetilde{U}_{jj} := \sum_{s \in \mathcal{S}}P_{s}'(Y^*)u_{j}(\tau_{j}^*(s)).
    \end{align*} For any agent that does not receive positive payments under any outcome, the diagonal entry is $0$. Define  $\widetilde{\bm{G}} \in \mathbb{R}^{n \times n}$ to be the transpose of the Hessian matrix, i.e., \begin{align*}
        \widetilde{G}_{jk} := \frac{\partial^2 Y}{\partial a_{k} \partial a_{j}}.
    \end{align*} Observe that the Hessian $\bm{G}$ defined for agents that receive a payment is a submatrix of $\widetilde{\bm{G}}$. Define the matrix $\widetilde{\bm{H}} \in \mathbb{R}^{n \times n}$ to be diagonal with entries \begin{align*}
        \widetilde{H}_{jj} := C''_{j}(a_{j}^*).
    \end{align*}  We can now define all agents' centralities given the optimal contract $\optsharesvector$ by \begin{align} \label{eq:centralitynewdef}
        \widetilde{\bm{\centrality}}^{T} := \widetilde{\bm{\productivity}}^{T} \left[\bm{I} - \widetilde{\bm{H}}^{-1}\widetilde{\bm{U}}\widetilde{\bm{G}}\right]^{-1}.
    \end{align}

    \begin{lemma}
\label{l:positivecentrality}
    Suppose $\optsharesvector$ is an optimal contract. For any agent $i$, the centrality $\widetilde{c}_{i} > 0$.
\end{lemma}

\begin{proof}

We first verify for agents that receive a payment that their centrality defined in $\widetilde{\bm{\centrality}}$ is equal to their centrality as defined in $\bm{\centrality}$. For ease of notation, let the spillover matrix on \textit{all} agents $\widetilde{\bm{S}} = \widetilde{\bm{H}}^{-1}\widetilde{\bm{U}}\widetilde{\bm{G}}$ and for those with a payment $\bm{S} = \bm{H}^{-1}\bm{U}\bm{G}$.

We will show that $\widetilde{\bm{S}}$ and $\bm{S}$ have the same non-zero eigenvalues. Consequently, they have the same spectral radius. Suppose $\mu$ is an eigenvalue of $\bm{S}$ with corresponding eigenvector $\bm{v}$. Then, $\mu$ is also an eigenvalue of $\widetilde{\bm{S}}$. To see this, suppose (without loss of generality) agents with a payment are labeled $\{1,\dots,k\}$. Rows $(k+1)$ onwards in $\widetilde{\bm{S}}$ have all zeros. The matrix $\bm{S}$ is the top-left $(k \times k)$ dimensional submatrix of $\widetilde{\bm{S}}$. We can define a $n$-dimensional vector $\widetilde{\bm{v}}$ as $\widetilde{v}_{i} := v_{i}$ when $i \leq k$ and  $\widetilde{v}_{i} := 0$ when $i > k$. It is straightforward to see $\widetilde{\bm{v}}$ is an eigenvector of $\widetilde{\bm{S}}$ with eigenvalue $\mu$. We will now prove the other direction. Suppose $\widetilde{\mu}$ is a non-zero eigenvalue of $\widetilde{\bm{S}}$ with corresponding eigenvector $\widetilde{\bm{v}}$. Since rows $(k+1)$ onwards in $\widetilde{\bm{S}}$ have all zeros, it must be that $\widetilde{v}_{i} = 0$ for components $i \geq (k+1)$. But this implies $\bm{v}$, corresponding to the first $k$ components of $\widetilde{\bm{v}}$, is an eigenvector of $\bm{S}$ with eigenvalue $\widetilde{\mu}$. Applying \Cref{t:spilloverspectralradius} tells us the spectral radius of $\widetilde{\bm{S}}$ is strictly smaller than $1$. 
    
So we have the following power series expansion of (\ref{eq:centralitynewdef}):\begin{equation} \label{eq:fullcentralsum}
        \widetilde{\bm{\centrality}}^{T} = \widetilde{\bm{\productivity}}^{T} \sum_{\ell=0}^{\infty} \widetilde{\bm{S}}^{\ell}.
    \end{equation} We can write for any $\ell \geq 1$ \begin{align*}
       \widetilde{\bm{S}}^{\ell}= \begin{bmatrix}
              \bm{S}^{\ell} & \bm{J}_{\ell} \\
              \bm{0} & \bm{0}
        \end{bmatrix},
    \end{align*} where $\bm{J}_{\ell}$ is a matrix that does not contribute to the centrality. Substituting in (\ref{eq:fullcentralsum}) and noting that the first $k$ elements of $\widetilde{\bm{\productivity}}$ are just the vector $\bm{\productivity}$, we get $\widetilde{\centrality}_{i} = \centrality_{i}$ for agents that receive a payment. The centrality $\widetilde{\centrality}_{i}$ for an agent without a payment will affect overall team performance when their payment at a particular outcome is perturbed. 

All that remains to show is that $\widetilde{\centrality}_{i} > 0$ for every agent. This follows from (\ref{eq:fullcentralsum}): $\widetilde{\bm{S}}$ is non-negative and every element of $\widetilde{\bm{\productivity}}$ is strictly positive.
\end{proof}

We now complete the proof of the proposition.

\begin{proof}[Proof of \Cref{p:inadaallpay}]
The proof involves analyzing the derivative of the principal's objective with respect to payments made to the agents. Consider any agent $i$. As shown in the proof of \Cref{t:generalmodel}, the derivative of the principal's objective with respect to $\tau_{i}(s)$ is
\begin{align*}
    lD\productivity_{i}\centrality_{i}\cdot \frac{u_{i}'(\tau_{i}^*(s))}{C''_i(a_i^*)}\cdot P_{s}'(Y^*) - P_{s}(Y^*).
\end{align*}

\textit{Forward direction: Let $\mathcal{S}_{i}^*$ be the set of outcomes at which an agent $i$ receives a positive payment. Then \begin{align*}
     P'_{s}(Y^*) > 0 \text{ for all } s \in \mathcal{S}^*_{i}.
\end{align*}}

The statement of the forward direction is exactly \Cref{l:positivepaymentoutcomes}. Note that the arguments to prove the lemma did not require an Inada condition.

\textit{Backward direction:  Any agent $i$ receives a strictly positive payment at all outcomes where $P_{s}'(Y^*) > 0$.}
    
By Part (a) of \Cref{as:extensivemargin}, at the optimal contract $\optsharesvector$, there exists some agent $i$ receiving a positive payment at some outcome $s$ and has centrality $c_{i}>0$. Note that from the \textit{forward direction}, we must have $P_{s}'(Y^*) > 0$. The principal's first-order condition (see proof of \Cref{t:generalmodel}), applied to agent $i$ is \begin{align*}
    lD\productivity_{i} \centrality_{i} \cdot \frac{u_{i}'(\tau_{i}^*(s))}{C''_i(a_i^*)}\cdot  P_{s}'(Y^*) - P_{s}(Y^*) = 0.
\end{align*} This implies $lD > 0$, because $c_{i} > 0$. For any other agent $j$, recall from the proof of \Cref{t:generalmodel} that the derivative of the principal's objective in $\tau_{j}(s)$, is given by the expression \begin{align*}
        lD \productivity_{j} \centrality_{j} \cdot \frac{u_{j}'(\tau_{j}^*(s))}{C''_j(a_j^*)}\cdot  P'_{s}(Y^*) - P_{s}(Y^*).
    \end{align*} This expression was stated for any agent receiving a positive payment at some outcome, but also holds for agents receiving a zero payment at all outcomes as long as the perturbation is made at an outcome $s$ at which $P_{s}'(Y^*) > 0$.\footnote{Recall that \Cref{l:changeteamperformance} was defined for any agent receiving a positive payment at some outcome. We show that the result also holds for agents receiving zero payments at all outcomes, when the perturbation in payments is made in an outcome where $P_{s}'(Y^*) > 0$. For any agent $i$ taking action $a_{i}^* > 0$, the first-order conditions at equilibrium imply that $a_{i}^*$ solves \begin{align*}
    C_{i}'(a_{i}) = \left(\sum_{s' \in \mathcal{S}}P_{s'}'(Y)u_{i}(\tau_{i}(s'))\right) \frac{\partial Y}{\partial a_{i}}.
\end{align*} We show the equation must also hold if $a_{i}^*=0$. To see this, recall that $a_{i}^*=0$ if and only if $\tau_{i}(s)=0$ at all outcomes $s$. Consider the first order-condition when $a_i^*=0$ and $\tau_i(s)=0$ for all $s$: \begin{align*}
    C_{i}'(0) = \left(\sum_{s' \in \mathcal{S}}P_{s'}'(Y)u_{i}(0)\right) \frac{\partial Y}{\partial a_{i}}.
\end{align*} The left-hand side is zero because $C'(0)=0$. The right-hand side is zero since $\sum_{s' \in \mathcal{S}}P_{s'}'(Y)=0$. So the first-order condition holds in this case as well. It follows that the first-order condition binds when payments are perturbed for such an agent at an outcome where $P_{s}'(Y^*) > 0$.}  It is straightforward to verify that $\centrality_{j}$ is exactly the term that appears in the centrality vector $\widetilde{\bm{\centrality}}$ defined in the proof of \Cref{l:positivecentrality}.  By the Inada condition on the marginal utility function, the observation that $lD \centrality_{j} > 0$ (which holds because $\centrality_{j}>0$ as shown in \Cref{l:positivecentrality}), and the fact that $P_{s}'(Y^*)>0$ \begin{align*}
        \lim_{\tau_{j}^*(s) \to 0}lD \productivity_{j} \centrality_{j} \cdot \frac{u_{j}'(\tau_{j}^*(s))}{C''_j(a_j^*)} P'_{s}(Y^*) - P_{s}(Y^*) > 0.
    \end{align*} Thus, we cannot have $\tau_{j}^*(s)=0$ at an optimal contract.
\end{proof}

\section{Optimal equity pay}
\label{a:optimalequitypay}

The contracts in our main model are finely tailored to individual outcomes (see \Cref{c:marginalutil2}). In practice, such contracts may be difficult to implement, and firms often use simple compensation schemes. Our results can be adapted to characterize optimal contracts within a restricted class. This section provides an illustration by analyzing one widely used incentive scheme: equity pay. Note that in simple success-or-failure environments, all optimal contracts feature equity pay, but in general the optimal equity contract need not match the optimal unrestricted contract.

An equity pay contract pays each agent a fixed share $\tau_i v_s$ of the surplus $v_s$ produced by the team. For a given equity contract $\bm{\tau}$, the expected payoff to the principal is \begin{align*}
    \left(1-\sum_{i \in N}\tau_{i}\right) \sum_{s \in \mathcal{S}}v_{s}P_{s}(Y).
\end{align*} The expected payoff to agent $i$ from an equity share $\tau_{i}$ is  \begin{align*}
    \mathcal{U}_{i} = \sum_{s \in \mathcal{S}}u_{i}\left(\tau_{i}v_{s}\right)P_{s}(Y) - C_{i}(a_{i}).
\end{align*}  The result below characterizes an optimal equity contract $\bm{\tau}^*$. We maintain \Cref{as:balancederive}, which now states that there is a neighborhood of $\bm{\tau}^*$ in the space of equity contracts where $\eqlbactions(\bm{\tau})$ is continuously differentiable.

\begin{proposition}
\label{t:optimalequitypay}
    Suppose $\bm{\tau}^*$ is an optimal equity contract and $Y^*$ is the induced team performance. There exists a constant $\balanceconstant$ such that for any agent $i$ receiving a positive equity payment, we have \begin{align*} \productivity_{i}\centrality_{i} \sum_{s \in \mathcal{S}}P'_{s}(Y^*)v_{s}\cdot \frac{u_{i}'(\tau_{i}^*v_{s})}{C''_i(a_i^*)}  = \balanceconstant. \end{align*}
\end{proposition} 
The proof, which we provide below, follows a similar approach to the proof of \Cref{t:generalmodel}. It involves analyzing the effect of perturbations to equity payments on the principal's objective. Perturbations in the equity payment of an agent affect payments at all outcomes. The direct effect of increasing $\tau_{i}$ on $i$'s action is proportional to the change in marginal expected utility from payments, which is given by the expression \begin{align*}
    \productivity_{i}\sum_{s \in \mathcal{S}}P_{s}'(Y^*)v_{s}\cdot \frac{u_{i}'\left(\tau_{i}^*v_{s}\right)}{C''_i(a_i^*)}.
\end{align*} The summation captures the total direct effect of increasing an agent's equity on team performance by aggregating across outcomes, and multiplying by $\centrality_{i}$ includes indirect effects. At an optimal equity contract, the effect of perturbing equity payments on total team performance must be the same for all agents with positive equity.

In general, the balance condition in \Cref{t:optimalequitypay} characterizing optimal equity contracts does not match the condition in \Cref{s:generalintensivemargin} characterizing optimal contracts. An optimal contract fine-tunes payments at each outcome, incentivizing agents to exert optimal effort levels. Equity pay imposes a particular linear relationship between the payments for different outcomes that may be practically convenient but sacrifices some incentive power. We note the contrast with \cite{dai2022robust}, building on \cite{carroll2015robustness}, which finds that linear contracts are optimal for a principal designing team incentives that must be robust to uncertainty about the environment.

\begin{proof}[Proof of \Cref{t:optimalequitypay}]

We begin with a lemma, which adapts \Cref{l:changeteamperformance}.

\begin{lemma}
    \label{l:changeteamperformanceequity}
    Suppose $\bm{\tau}^*$ is an optimal equity contract with corresponding equilibrium actions $\eqlbactions$ and team performance $Y^*$. %
    For any agent $i$, the derivative of  team performance in $\tau_{i}$, evaluated at $\bm{\tau}^*$, is \begin{align*}
      \frac{d Y}{d \tau_{i}} = l \productivity_{i} \centrality_{i} \sum_{s \in \mathcal{S}} P'_{s}(Y^*) v_{s} \cdot \frac{u_{i}'(\tau_{i}^*v_{s})}{C''_i(a_i^*)} ,
    \end{align*}  where $l$ is independent of $i$ and $s$.
\end{lemma}

\begin{proof}
The steps taken in this proof are exactly the same as those taken in the proof of \Cref{l:changeteamperformance}. We analyze the change in team performance as the equity transfer to an agent is perturbed. Consider an equity payment scheme $\optsharesvector$ and any agent $i$. Consider marginally increasing $\tau_{i}$. The change induced by this perturbation is \begin{equation} \label{eq:changeteamperformanceequity}
    \frac{\partial Y}{\partial \tau_{i}} =\nabla Y(\eqlbactions)^{T} \cdot \frac{\partial \eqlbactions}{\partial \tau_{i}},
\end{equation} where $\eqlbactions$ is the equilibrium action profile for the contract $\sharesvector$. The substance of the proof is analyzing the second term on the right-hand side of (\ref{eq:changeteamperformanceequity}).  

As in \Cref{l:changeteamperformance}, it is without loss to analyze the change in the action of agent $i$ and actions of agents $j$ that take strictly positive actions in profile $\eqlbactions$. The analysis from here on focuses on such agents, overloading notation to represent the actions of these agents by $\eqlbactions$.

We will show that the change in equilibrium actions $\eqlbactions$ as the equity $\tau_{i}$ increases is \begin{equation}
    \label{eq:eqlbactionsderivativeequity}
    \frac{\partial \bm{a}^*}{\partial \tau_{i}} = \left[\bm{I}-\bm{H}^{-1} \bm{U}\bm{G}\right]^{-1} \bm{H}^{-1} \begin{bmatrix}
                 \bm{0}  \\
                 \frac{\partial Y}{\partial a_{i}}\sum_{s \in \mathcal{S}}P'_{s}(Y)v_{s}u_{i}'(\tau_{i}v_{s}) \\
                 \bm{0}
            \end{bmatrix} + \frac{\partial Y}{\partial \tau_{i}}\left[\bm{H}-\bm{U}\bm{G}\right]^{-1}\bm{d},
\end{equation} where $\bm{d}$ is a vector that depends on the curvature of the probability of outcome $P_{s}(\cdot)$. 

Consider the equilibrium action profile $\eqlbactions$. For an agent $j$, the first-order conditions imply $a_{j}^*$ must solve the equation \begin{equation} \label{eq:bestresponsegeneralequity}
    C_{j}'(a_{j}) = \left(\sum_{s \in \mathcal{S}}P_{s}'(Y)u_{j}(\tau_{j}v_{s})\right) \frac{\partial Y}{\partial a_{j}}.
\end{equation} 

To arrive at (\ref{eq:eqlbactionsderivativeequity}), let us implicitly differentiate (\ref{eq:bestresponsegeneralequity}) with respect to $\tau_{i}$. For all $j \neq i$, \begin{align} \label{eq:implicitgeneraljneqiequity}
    C_{j}''(a_{j}^*)\frac{\partial a^*_{j}}{\partial \tau_{i}} &= \left(\sum_{s \in \mathcal{S}}P_{s}'(Y)u_{j}\left(\tau_{j}v_{s}\right)\right) \left(\sum_{k=1}^{n}\frac{\partial^2 Y}{\partial a_{k} \partial a_{j}} \cdot \frac{\partial a_{k}^*}{\partial \tau_{i}}\right) \\ &+ \frac{\partial Y}{\partial a_{j}} \cdot \frac{\partial Y}{\partial \tau_{i}} \cdot \sum_{s \in \mathcal{S}}P_{s}''(Y)u_{j}(\tau_{j}v_{s}).
\end{align} Similarly for $j=i$,  \begin{align} \label{eq:implicitgeneraljeqiequity}
    C_{j}''(a_{j}^*)\frac{\partial a^*_{j}}{\partial \tau_{i}} &= \left(\sum_{s \in \mathcal{S}}P_{s}'(Y)u_{j}(\tau_{j}v_{s})\right) \left(\sum_{k=1}^{n}\frac{\partial^2 Y}{\partial a_{k} \partial a_{j}} \cdot \frac{\partial a_{k}^*}{\partial \tau_{i}}\right) \\ & + \frac{\partial Y}{\partial a_{j}}\sum_{s \in \mathcal{S}}P'_{s}(Y)v_{s}u_{j}'(\tau_{j}v_{s}) +  \frac{\partial Y}{\partial a_{j}} \cdot \frac{\partial Y}{\partial \tau_{i}}\sum_{s \in \mathcal{S}}P_{s}''(Y)u_{j}(\tau_{j}v_{s}).
\end{align}  We can combine (\ref{eq:implicitgeneraljneqiequity}) and (\ref{eq:implicitgeneraljeqiequity}) in vector form: \begin{align*}
    \frac{\partial \bm{a}^*}{\partial \tau_{i}} &= \left[\bm{H}-\bm{U}\bm{G}\right]^{-1} \begin{bmatrix}
                 \bm{0}  \\
                 \frac{\partial Y}{\partial a_{i}}\sum_{s \in \mathcal{S}}P'_{s}(Y)v_{s}u_{i}'(\tau_{i}v_{s}) \\
                 \bm{0}
            \end{bmatrix} + \frac{\partial Y}{\partial \tau_{i}}\left[\bm{H}-\bm{U}\bm{G}\right]^{-1}\bm{d},
\end{align*} where $\bm{d}$ is a vector with $d_{j} = \frac{\partial Y}{\partial a_{j}} \sum_{s \in \mathcal{S}}P_{s}''(Y)u_{j}(\tau_{j}v_{s})$. The expression in (\ref{eq:eqlbactionsderivativeequity}) follows.

Substituting (\ref{eq:eqlbactionsderivativeequity}) into (\ref{eq:changeteamperformanceequity}),  the change in team performance as the equity payment $\tau_{i}$ increases is\begin{align*}
    \frac{\partial Y}{\partial \tau_{i}}& = \nabla Y(\eqlbactions)^{T}  \left[\bm{I}-\bm{H}^{-1} \bm{U} \network \right]^{-1} \bm{H}^{-1}\begin{bmatrix}
                 \bm{0}  \\
                 \frac{\partial Y}{\partial a_{i}}\sum_{s \in \mathcal{S}}P'_{s}(Y)v_{s}u_{i}'(\tau_{i}v_{s}) \\
                 \bm{0}
            \end{bmatrix}  \\ &+ \frac{\partial Y}{\partial \tau_{i}}\nabla Y(\eqlbactions)^{T}\left[\bm{H}-\bm{U}\bm{G}\right]^{-1}\bm{d}.
\end{align*} Applying the definitions of $\productivity_{i}$ and $\centrality_{i}$, we obtain \begin{align*}
    \frac{\partial Y}{\partial \tau_{i}} = \productivity_{i}\centrality_{i} \sum_{s \in \mathcal{S}}P'_{s}(Y)v_{s}\cdot \frac{u_{i}'(\tau_{i}^*v_{s})}{C''_i(a_i^*)} + \frac{\partial Y}{\partial \tau_{i}}\nabla Y(\eqlbactions)^{T}\left[\bm{H}-\bm{U}\bm{G}\right]^{-1}\bm{d}.
\end{align*}
Rearranging,
$$    \frac{\partial Y}{\partial \tau_{i}} = \frac{1}{1-\nabla Y(\eqlbactions)^{T}\left[\bm{H}-\bm{U}\bm{G}\right]^{-1}\bm{d}} \cdot \productivity_{i}\centrality_{i}\sum_{s \in \mathcal{S}}P'_{s}(Y)v_{s}\cdot \frac{u_{i}'(\tau_{i}^*v_{s})}{C''_i(a_i^*)}
.$$
Setting $l=\frac{1}{1-\nabla Y(\eqlbactions)^{T}\left[\bm{H}-\bm{U}\bm{G}\right]^{-1}\bm{d}}$ and observing $l$ does not depend on $i$, we obtain the desired result.
\end{proof}

    The expected payoff for the principal under equity payment $\bm{\tau}$ and corresponding equilibrium actions $\eqlbactions$ is \begin{align*}
\left(1-\sum_{i \in N}\tau_{i}\right)\sum_{s \in \mathcal{S}}v_{s}P_{s}(Y(\eqlbactions)).
    \end{align*} Suppose $\bm{\tau}^*$ is an optimal equity contract inducing equilibrium $\eqlbactions(\bm{\tau}^*)$ with team performance $Y^*$. Consider agent $i$ such that $\tau_{i}^*>0$. 
    Then the first-order condition for $\tau^*_i$ implies that \begin{align*}
        \frac{d Y}{d \tau_{i}} \cdot \underbrace{\left(1-\sum_{i \in N}\tau_{i}^*\right)\sum_{s \in \mathcal{S}}v_{s}P'_{s}(Y^*)}_{D} = \sum_{s \in \mathcal{S}}v_{s}P_{s}(Y^*).
    \end{align*}
The left-hand side is the benefit from increasing $\tau^*_i$ while the right-hand side is the expected additional transfer required. Since each outcome occurs with non-zero probability, the summation labeled $D$ is nonzero.
    
Substituting \Cref{l:changeteamperformanceequity} in the above equation, we obtain 
\begin{align*} l \productivity_{i}\centrality_{i} \sum_{s \in \mathcal{S}}P'_{s}(Y^*)v_{s}\cdot \frac{u_{i}'(\tau_{i}^*v_{s})}{C''_i(a_i^*)} &= \frac{\sum_{s \in \mathcal{S}}v_{s}P_{s}(Y^*)}{D}, \\ 
\iff \productivity_{i}\centrality_{i} \sum_{s \in \mathcal{S}}P'_{s}(Y^*)v_{s}\cdot \frac{u_{i}'(\tau_{i}^*v_{s})}{C''_i(a_i^*)}&= \balanceconstant,
\end{align*} where $\balanceconstant = \sum_{s \in \mathcal{S}}v_{s}P_{s}(Y^*) /(lD)$. Observing that $\balanceconstant$ is independent of $i$ and the outcome $s$, the statement of the result follows.
\end{proof}
\end{document}